\documentclass[12pt,a4paper]{article}
\usepackage[utf8]{inputenc}
\usepackage[sectionbib]{natbib}
\usepackage[margin=1in]{geometry}
\usepackage{graphicx}
\usepackage{mathrsfs}
\usepackage{multirow}
\usepackage{amsfonts}
\usepackage{amsmath}
\usepackage{comment}
\usepackage{amsthm}
\usepackage{color}
\usepackage{mathtools}
\usepackage{algorithm,algorithmic}
\usepackage{sectsty}
\usepackage[colorlinks]{hyperref}
\usepackage{amssymb}
\usepackage[toc,page]{appendix}
\usepackage[mathscr]{euscript}
\usepackage[toc]{appendix}

\usepackage{chngcntr}
\usepackage{apptools}
\usepackage{booktabs}
\usepackage{lscape}
\usepackage{arydshln}
\usepackage{soul}
\usepackage{titling}
\usepackage{euscript}
\usepackage{enumitem}
\usepackage{xr}

\newtheorem{assumption}{Assumption}

\newtheorem{lemma}{Lemma}
\newtheorem{proposition}{Proposition}
\newtheorem{theorem}{Theorem}
\newtheorem{corollary}{Corollary}
\theoremstyle{definition}
\newtheorem{remark}{Remark}

\hypersetup{
    colorlinks=true,
    linkcolor=black,
    citecolor=black,
    filecolor=magenta,
    urlcolor=cyan,
}

\DeclareMathOperator*{\argmin}{arg\,min}

\DeclareMathOperator*{\diag}{diag}
\newcommand{\bm}{\boldsymbol}
\newcommand{\bb}{\mathbb}
\newcommand{\cm}[1]{\mbox{\boldmath$\mathscr{#1}$}}
\newcommand{\pr}{\prime}

\mathtoolsset{showonlyrefs}

\title{High-dimensional vector autoregressive time series modeling via tensor decomposition}
\author{Di Wang, Yao Zheng, Heng Lian and Guodong Li
	\\ \textit{University of Hong Kong, University of Connecticut}
    \\ \textit{and City University of Hong Kong}}

\begin{document}

\setlength{\parindent}{16pt}

\setlength{\droptitle}{-6em}
\maketitle

\begin{abstract}

The classical vector autoregressive model is a fundamental tool for multivariate time series analysis. However, it involves too many parameters when the number of time series and lag order are even moderately large. This paper proposes to rearrange the transition matrices of the model into a tensor form such that the parameter space can be restricted along three directions simultaneously via tensor decomposition. 
In contrast, the reduced-rank regression method can restrict the parameter space in only one direction. Besides achieving substantial dimension reduction, the proposed model is interpretable from the factor modeling perspective. 
Moreover, to handle high-dimensional time series, this paper considers imposing sparsity on factor matrices to improve the   model interpretability and estimation efficiency, which leads to a sparsity-inducing estimator. For the low-dimensional case, we derive asymptotic properties of the proposed least squares estimator and introduce an alternating least squares algorithm. For the high-dimensional case, we establish non-asymptotic properties of the sparsity-inducing estimator and propose an ADMM algorithm for regularized estimation. Simulation experiments and a real data example demonstrate the advantages of the proposed approach over various existing methods.
\end{abstract}

\textit{Keywords}: Factor model; High-dimensional time series; Reduced-rank regression; Tucker decomposition; Variable selection.

\newpage
\section{Introduction}

High-dimensional time series is one of the most common types of ``big data" and can be found in many areas including meteorology, genomics, finance and economics \citep{Hallin_Lippi2013}.
The classical vector autoregressive (VAR) model is fundamental to multivariate time series modeling and has recently been applied to the high-dimensional case under certain structural assumptions, e.g., the banded structure \citep{Guo_Wang_Yao2016}, network structure \citep{Zhu_Pan_Li_Liu_Wang2017}, and linear restrictions \citep{zheng20}. Consider the VAR model of the form \citep{Lutkepohl2005, Tsay2010}:
\begin{equation}
	\label{eq:VAR}
	\bm{y}_t=\bm{A}_1\bm{y}_{t-1}+\cdots+\bm{A}_P\bm{y}_{t-P}+\bm{\epsilon}_t, \hspace{5mm}1\leq t\leq T,
\end{equation}
where $\{\bm{y}_t\}$ is the observed time series with $\bm{y}_t=(y_{1t},\dots,y_{Nt})^\pr\in\mathbb{R}^N$, $\{\bm{\epsilon}_t\}$ are independent and identically distributed ($i.i.d.$) innovations with  $\bm{\epsilon}_t=(\epsilon_{1t},\dots,\epsilon_{Nt})^\pr$, $\mathbb{E}(\bm{\epsilon}_t)=0$ and $\text{var}(\bm{\epsilon}_t)<\infty$, $\bm{A}_j$s are $N\times N$ transition matrices of unknown parameters, and $T$ is the sample size. It can be difficult to perform the estimation even when the dimensions $N$ and $P$ are moderately large \citep{DeMol08,Carriero11,Koop13}.

On the other hand, compared with model \eqref{eq:VAR}, the vector autoregressive moving average (VARMA) model usually performs better in practice since it can provide a more flexible autocorrelation structure \citep{athanasopoulos2008varma,Chan16}.
However, the VARMA model may have a serious identification problem \citep{Chan16,Basu-varma,Dias18}, and its estimation is often unstable since the corresponding objective function involves a high-order polynomial. 
As a result, it is common in practice to employ a VAR model to approximate VARMA processes, and the order $P$ may be very large in order to provide a better fit for the data \citep{Ravenna07}.
For example, to guarantee the approximation accuracy, we need to assume that $P\rightarrow\infty$ and $PT^{-1/3}\rightarrow 0$ as $T\rightarrow \infty$ for univariate and multivariate cases \citep{Said_Dickey1984,Li_Leng_Tsai2014}. This makes the number of parameters in model \eqref{eq:VAR}, $N^2P$, much larger.

Therefore, to make inference on the VAR model for high-dimensional time series, it is necessary to restrict the parameter space of model \eqref{eq:VAR} to a reasonable number of degrees of freedom.
A direct method is to assume that the transition matrices $\bm{A}_j$s are sparse and apply sparsity-inducing regularized estimation, e.g., the $\ell_1$-regularization (Lasso or Dantzig selector) for VAR models  \citep{Kock_Callot2015,Davis_Zang_Zheng2016,Basu15,Han15,Wu16}.
However, unlike the traditional linear regression, time series data have non-negligible temporal and cross-sectional dependencies, which will seriously affect the accuracy of the regularized estimation.
Moreover, as explained in Remark \ref{rem1} in Section 2, the stationarity of the VAR model essentially entails that the average magnitude of parameters is bounded by $O(N^{-1/2})$.
This makes the variable selection much more challenging and hence limits the popularity of sparsity-inducing regularized estimation for time series data. 

Another important approach to reducing the dimensionality of model \eqref{eq:VAR} arises naturally from the reduced-rank regression \citep{Yuan07,Negahban11,Chen13,Raskutti17,basu2019low}. The VAR model in \eqref{eq:VAR} can be rewritten as
\begin{equation}\label{model-vec}
\bm{y}_t=\bm{A}^{(C)}\bm{x}_{t}+\bm{\epsilon}_t,
\end{equation}
where $\bm{x}_t=(\bm{y}_{t-1}',\ldots,\bm{y}_{t-P}')'$, and $\bm{A}^{(C)}=(\bm{A}_1,\ldots,\bm{A}_P)$ is assumed to have a low rank \citep{Velu86,Velu98}.
Based on the reduced-rank VAR model in \eqref{model-vec}, \cite{Carriero11} considered a Bayesian method to predict large macroeconomic data, and both the number of variables $N$ and the sample size $T$ diverge to infinity.
However, unlike the reduced-rank regression, we may have alternative ways to define the low-rankness of parameter matrices $\bm{A}_j$s with $P>1$.
Specifically, the rank of $\bm{A}^{(C)}$ is the dimension of the column space of $\bm{A}_j$s.
Denote $\bm{A}^{(R)}=(\bm{A}_1',\bm{A}_2',\dots,\bm{A}_P')$ and $\bm{A}^{(L)}=(\text{vec}(\bm{A}_1),\text{vec}(\bm{A}_2),\dots,\text{vec}(\bm{A}_P))'$, where $\mathrm{vec}(\bm{A}_j)$ is the vectorization of $\bm{A}_j$.
The ranks of $\bm{A}^{(R)}$ and $\bm{A}^{(L)}$ are then the dimensions of the row space and vectorized matrix space of $\bm{A}_j$s, respectively.
The three dimensions are different in general, and the corresponding low-rank structures have different physical interpretations; see Section 2 for details.
Similarly to model \eqref{model-vec} above, \cite{reinsel1983some} proposed an autoregressive index model, where the low-rank assumption was imposed on $\bm{A}^{(R)}$. 
Moreover, the transition matrices $\bm{A}_j$s may have a low-rank structure along different lags, i.e. $\bm{A}^{(L)}$ may be low-rank.
In fact, the VARMA model can be treated as a parsimonious formulation for VAR models, since it restricts the degrees of freedom on transition matrices over different lags \citep{Tsay2010}.

It is noteworthy that imposing the low-rank assumption on any one of $\bm{A}^{(C)}$, $\bm{A}^{(R)}$ and $\bm{A}^{(L)}$ leads to a different physical interpretation as it amounts to reducing the dimensionality along one of the three different directions. This inspires us to rearrange the transition matrices $\bm{A}_j$s into a tensor. Interestingly, the corresponding mode-1, -2 and -3 matricizations of the tensor happen to be $\bm{A}^{(C)}$, $\bm{A}^{(R)}$ and $\bm{A}^{(L)}$, respectively; see \cite{Kolda09} and Section 2.
By adopting the standard Tucker decomposition for the transition tensor, different low-rank structures can be assumed simultaneously along the three directions, and hence the parameter space of the VAR model can be efficiently restricted. We call the resulting model the multilinear low-rank VAR model, since the Tucker ranks are also called multilinear ranks.

In the literature, low-rank structures of high-dimensional time series are commonly explored through factor models \citep{SW05, BN08, SW11, BW16}. Similarly, as a means of low-rank discovery for VAR processes, the proposed model is naturally interpretable from the  factor modeling perspective. As we will discuss in Section \ref{subsec:MLRVAR}, by imposing the low-rankness along three directions, the proposed model can extract different dynamic factors across  response variables, predictor variables, and predictor time lags. Indeed, the proposed model can be written as a static factor model (SFM) \citep{BW16} yet endowed with additional low-rank structures for more substantial dimension reduction.
However, in contrast to factor models which are mainly used for interpretation, it is worth noting that the proposed model can be used for forecasting. On the other hand,  the dynamic factor model (DFM) in the literature can be constructed by combining the SFM with a certain dynamic structure for the latent factors \citep{SW11}. Compared to the DFM with VAR latent factors \citep{AW07}, the proposed model may be more flexible in the sense that it  can extract different sets of dynamic factors from the response $\bm{y}_{t}$ and the lagged predictors $\bm{y}_{t-j}$s, whereas the DFM restricts them to be identical.  In addition, the proposed model  can capture the possible low-rank structure across the $P$ time lags.

Another important contribution of this paper is to introduce a sparse decomposition for the transition tensor to further increase the estimation efficiency for much higher-dimensional time series data. In the literature, sparsity-inducing regularization has been widely considered in reduced-rank regression to improve  interpretability and efficiency. For example, \cite{ChenHuang12} and \cite{Bunea12} considered row-wise sparsity in singular value decomposition, where zero rows imply irrelevance of the corresponding predictors to the responses;
\cite{Lian15} proposed to directly restrict the rank of the coefficient matrix with entry-wise sparsity, which however does not lead to a sparse decomposition;
\cite{ChenChan12} obtained a sparse singular value decomposition of the coefficient matrix by slightly relaxing the strict orthogonality; and \cite{Uematsu17} achieved the sparsity and strict orthogonality simultaneously. Note that as in \cite{Uematsu17}, our estimation method is able to keep the strict orthogonality of the factor matrices in the tensor decomposition.

Our work is also related to the fast-growing literature on tensor regression; see, e.g., \citet{Zhou13}, \citet{Li13}, \citet{Li17}, \citet{Sun17} and \citet{Raskutti17}. Whereas most of the existing work focuses on tensor-valued predictors or responses, we employ tensor decomposition as a novel approach to the dimensionality reduction of vector-valued time series models. To summarize, the proposed methods have the following attractive features:
\begin{itemize}
	\item [(a)] The proposed model substantially reduces the dimension  along three directions of the transition tensor, allowing each direction to have a different low-rank structure. This results in interpretable physical  structures and interesting connections with factor models in the literature, and allows us to handle much higher dimensional data than the reduced-rank VAR model in \eqref{model-vec}. 
	
	\item [(b)] Through the sparsity assumption on the three factor matrices, the proposed  high-dimensional method further improves the model interpretability and estimation efficiency  by selecting important variables for each response, predictor or temporal factor. The corresponding estimation can be accomplished by an ADMM algorithm which effectively untangles the $\ell_1$-regularization and orthogonality constraints.
\end{itemize}

The rest of the paper is organized as follows. Section 2 introduces the proposed model and discusses its connections with factor models.
Section 3 presents asymptotic properties of the least squares estimator in low dimensions and an alternating least squares algorithm.
For the high-dimensional case, the sparse higher-order reduced-rank estimation is proposed in Section 4,  taking into account both the orthogonality and sparsity. Its non-asymptotic properties are established, and an ADMM algorithm is developed. A consistent rank selection method is proposed in Section 5.
Simulation experiments and real data analysis are presented in Sections 6 and 7, respectively. A short discussion is given in Section 7.  All technical proofs are given in the Appendix.

\section{Multilinear low-rank vector autoregression}

\subsection{Tensor decomposition}
Tensors, also known as multidimensional arrays, are natural higher-order extensions of matrices. A multidimensional array $\cm{X}\in\mathbb{R}^{p_1\times\cdots\times p_K}$ is called a $K$th-order tensor, and the order of a tensor is known as the dimension, way or mode;  we refer readers to \citet{Kolda09} for a detailed review on tensor notations and operations. This paper will focus on third-order tensors.

Throughout the paper, we denote vectors by small boldface letters  $\bm{y}$, $\bm{x}, \dots$, matrices by capital letters $\bm{Y}$, $\bm{X}, \dots$, and tensors by Euler script capital letters $\cm{Y}$, $\cm{X}, \dots$. For a vector $\bm{x}$, denote by $\|\bm{x}\|_1$ and $\|\bm{x}\|_2$ its $\ell_1$ and $\ell_2$ norms, respectively. For a matrix $\bm{X}$, denote by $\|\bm{X}\|_{\textup{F}}$, $\|\bm{X}\|_1$, $\|\bm{X}\|_0$,  $\|\bm{X}\|_\text{op}$, $\|\bm{X}\|_*$, $\text{vec}(\bm{X})$, $\bm{X}'$ and  $\sigma_j(\bm{X})$ its Frobenius norm, vectorized $\ell_1$ norm (i.e. $\|\bm{X}\|_1=\|\text{vec}(\bm{X})\|_1$),  $\ell_0$ ``norm", spectral norm, nuclear norm, vectorization, transpose and  the $j$-th largest singular value,  respectively. For two symmetric matrices $\bm{X}$ and $\bm{Y}$, we write $\bm{X}\leq\bm{Y}$ if $\bm{Y}-\bm{X}$ is positive semidefinite. Furthermore, for a tensor $\cm{X}\in\mathbb{R}^{p_1\times p_2\times p_3}$, let $\|\cm{X}\|_{\textup{F}}=\left(\sum_{i=1}^{p_1}\sum_{j=1}^{p_2}\sum_{k=1}^{p_3}\cm{X}_{ijk}^2\right)^{1/2}$ and
$\|\cm{X}\|_0=\sum_{i=1}^{p_1}\sum_{j=1}^{p_2}\sum_{k=1}^{p_3}1(\cm{X}_{ijk}\neq0)$ be its Frobenius norm and $\ell_0$ ``norm", respectively.

For a tensor $\cm{X}\in\mathbb{R}^{p_1\times p_2\times p_3}$, its mode-1 matricization $\cm{X}_{(1)}$ is defined as the $p_1$-by-$(p_2p_3)$ matrix  whose $\{i,(k-1)p_3+j\}$-th entry is $\cm{X}_{ijk}$, for $1\leq i\leq p_1,1\leq j\leq p_2$ and $1\leq k\leq p_3$,
and $\cm{X}_{(1)}$  contains all mode-1 fibers $\{(\cm{X}_{[:,i_2,i_3]})\in\mathbb{R}^{p_1}:1\leq i_2\leq p_2,1\leq i_3\leq p_3\}$. The mode-2 and mode-3 matricizations can be defined similarly. The matricization of tensors helps to link the concepts and properties of matrices to those of tensors. The mode-1 multiplication $\times_1$ of a tensor $\cm{X}\in\mathbb{R}^{p_1\times p_2\times p_3}$ and a matrix $\bm{Y}\in\mathbb{R}^{q_1\times p_1}$ is defined as
\begin{equation}
	\cm{X}\times_1\bm{Y}=\left(\sum_{i=1}^{p_1}\cm{X}_{ijk}\bm{Y}_{si}\right)_{1\leq s\leq q_1,1\leq j\leq p_2,1\leq k\leq p_3}.
\end{equation}
Multiplications $\times_2$ and $\times_3$ can be defined similarly.

Unlike matrices, there is no universal definition of the rank for tensors. In this paper, we consider the multilinear ranks $(r_1,r_2,r_3)$ of a tensor $\cm{X}\in\mathbb{R}^{p_1\times p_2\times p_3}$, where
\begin{equation}
		 r_1=\text{rank}_1(\cm{X}):=\text{rank}(\cm{X}_{(1)})=\dim(\text{span}\{\cm{X}_{[:,i_2,i_3]}\in\mathbb{R}^{p_1}:1\leq i_2\leq p_2,1\leq i_3\leq p_3\}),
\end{equation}
and $r_2$ and $r_3$ are the ranks of $\cm{X}_{(2)}$ and $\cm{X}_{(3)}$, respectively. Note that $r_1$, $r_2$ and $r_3$ are analogous to the row rank and column rank of a matrix, but these three ranks are not necessarily equal. The multilinear ranks are also known as Tucker ranks, as they are closely related to the Tucker decomposition.

For a tensor $\cm{X}\in\mathbb{R}^{p_1\times p_2\times p_3}$, if $\text{rank}_j(\cm{X})=r_j$ for $1\leq j\leq 3$, then there exists a Tucker decomposition \citep{Tucker66,DeLathauwer00},
	\begin{equation*}
		\cm{X}=\cm{Y}\times_1\bm{Y}_1\times_2\bm{Y}_2\times_3\bm{Y}_3,
	\end{equation*}
where $\cm{Y}\in \mathbb{R}^{r_1\times r_2\times r_3}$ is the core tensor, $\bm{Y}_j\in\mathbb{R}^{p_j\times r_j}$ with $1\leq j\leq 3$ are factor matrices, and the above decomposition can also be denoted by $\cm{X}=[\![\cm{Y};\bm{Y}_1,\bm{Y}_2,\bm{Y}_3]\!]$.

\subsection{Multilinear low-rank vector autoregression \label{subsec:MLRVAR}}

\begin{figure}[t]
	\begin{center}
		\includegraphics[width=15cm]{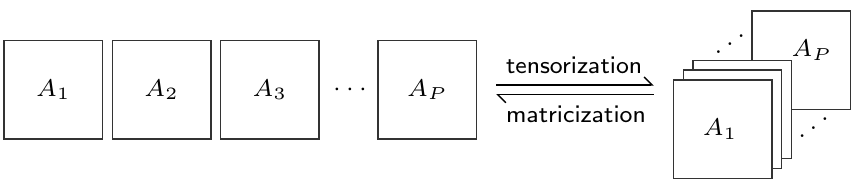}
		\label{fg:tensor}
		\caption{\label{fig1}Rearranging $P$ transition matrices of a VAR model into a tensor.}
	\end{center}
\end{figure}

For the VAR model in \eqref{eq:VAR}, we can rearrange its transition matrices into a tensor $\cm{A}\in\mathbb{R}^{N\times N\times P}$; see Figure \ref{fig1} for an illustration. Denote by $\cm{A}_{(j)}$ the mode-$j$ matricization of $\cm{A}$, where $1\leq j\leq 3$. It can be verified that $\cm{A}_{(1)}=(\bm{A}_1,\ldots,\bm{A}_P)$, $\cm{A}_{(2)}=(\bm{A}_1',\bm{A}_2',\dots,\bm{A}_P')$ and $\cm{A}_{(3)}=(\text{vec}(\bm{A}_1),\text{vec}(\bm{A}_2),\dots,\text{vec}(\bm{A}_P))'$, which correspond to the column, row and vectorized matrix spaces of $\bm{A}_j$s, respectively.

If the transition tensor $\cm{A}$ has multilinear low ranks  $(r_1,r_2,r_3)$, i.e. $\text{rank}(\cm{A}_{(j)})=r_j$ for  $1\leq j\leq 3$, then there exists a Tucker decomposition, $\cm{A}=\cm{G}\times_1\bm{U}_1\times_2\bm{U}_2\times_3\bm{U}_3$ or $\cm{A}=[\![\cm{G};\bm{U}_1,\bm{U}_2,\bm{U}_3]\!]$,
where $\cm{G}\in\mathbb{R}^{r_1\times r_2\times r_3}$ is the core tensor, and
$\bm{U}_1\in\mathbb{R}^{N\times r_1}$, $\bm{U}_2\in\mathbb{R}^{N\times r_2}$ and $\bm{U}_3\in\mathbb{R}^{P\times r_3}$ are factor matrices. As a result, model \eqref{eq:VAR} can be written as
\begin{equation}
\label{eq:MLRVAR}
\bm{y}_t=(\cm{G}\times_1\bm{U}_1\times_2\bm{U}_2\times_3\bm{U}_3)_{(1)}\bm{x}_t+\bm{\epsilon}_t,
\end{equation}
where $\bm{x}_t=(\bm{y}_{t-1}',\dots,\bm{y}_{t-P}')'$. For simplicity, we call model \eqref{eq:MLRVAR} the multilinear low-rank VAR model. 

In addition, since $(\cm{G}\times_1\bm{U}_1\times_2\bm{U}_2\times_3\bm{U}_3)_{(1)}=\bm{U}_1\cm{G}_{(1)}(\bm{U}_3\otimes\bm{U}_2)'$, where $\otimes$ is the Kronecker product, model \eqref{eq:MLRVAR} also has the following equivalent forms
\begin{equation}
\label{eq:MLRVAR_1}
\bm{y}_t=\bm{U}_1\cm{G}_{(1)}(\bm{U}_3\otimes\bm{U}_2)'\bm{x}_t+\bm{\epsilon}_t=\bm{U}_1\cm{G}_{(1)}\text{vec}(\bm{U}_2'\bm{X}_t\bm{U}_3)+\bm{\epsilon}_t,
\end{equation}
where $\bm{X}_t=(\bm{y}_{t-1},\dots,\bm{y}_{t-P})$.

\begin{assumption}
	\label{asmp:stationary}
	All roots of the matrix polynomial $\mathcal{A}(z)=\bm{I}_N-\bm{A}_1z-\dots-\bm{A}_Pz^P$, $z\in\mathbb{C}$, are outside the unit circle, where $\mathbb{C}$ is the set of complex numbers. 
\end{assumption}

Assumption \ref{asmp:stationary} is the sufficient and necessary condition for the existence of a unique strictly stationary solution to model \eqref{eq:VAR}. When $P=1$, Assumption \ref{asmp:stationary} is equivalent to $\rho(\bm{A}_1)<1$, where $\rho(\bm{A}_1)$ denotes the spectral radius of $\bm{A}_1$.

\begin{remark}\label{rem1}
To gain insight into the effect of the stationarity condition on the entries of $\bm{A}_1$, we may consider the following result regarding random matrices.  Suppose that the entries of $\bm{A}_1$ are $i.i.d.$ with mean zero and variance $\sigma^2$, that is, they are equally important. Then, by \cite{bai1997circular},  $N^{-1/2}\rho(\bm{A}_1) \rightarrow \sigma$ in probability as $N\rightarrow\infty$. In other words, when $\rho(\bm{A}_1)<1$, a larger $N$ will shrink the entries of   $\bm{A}_1$ towards zero.
\end{remark}

Note that the Tucker decomposition in \eqref{eq:MLRVAR} is not unique since $[\![\cm{G};\bm{U}_1,\bm{U}_2,\bm{U}_3]\!]=[\![\cm{G}\times_1\bm{O}_1\times_2\bm{O}_2\times_3\bm{O}_3;\bm{U}_1\bm{O}_1^{-1},\bm{U}_2\bm{O}_2^{-1},\bm{U}_3\bm{O}_3^{-1}]\!]$
for any nonsingular matrices $\bm{O}_1\in\mathbb{R}^{r_1\times r_1}$, $\bm{O}_2\in\mathbb{R}^{r_2\times r_2}$ and $\bm{O}_3\in\mathbb{R}^{r_3\times r_3}$. Hence, we consider a special Tucker decomposition: the higher-order singular value decomposition (HOSVD); see \cite{DeLathauwer00}. Specifically,
we let $\bm{U}_j$ be a tall matrix consisting  of the top $r_j$ left singular vectors of $\cm{A}_{(j)}$ for each $1\leq j\leq 3$, where $(r_1,r_2,r_3)$ are the multilinear ranks of the tensor $\cm{A}$. Let the core tensor  $\cm{G}=\cm{A}\times_1\bm{U}_1'\times_2\bm{U}_2'\times_3\bm{U}_3'$. Then $\cm{G}$ has the following \textit{all-orthogonal} property: for each $1\leq j\leq 3$, the rows of $\cm{G}_{(j)}$ are pairwise orthogonal. 

\begin{remark}
	Due to the HOSVD, the proposed multilinear low-rank VAR model in \eqref{eq:MLRVAR} has only $r_1r_2r_3+(N-r_1)r_1+(N-r_2)r_2+(P-r_3)r_3$ parameters, i.e. the dimension increases  linearly in $N$ and $P$; see \cite{Zhang18}. By contrast, model \eqref{eq:VAR} has $N^2P$ parameters, while the reduced-rank VAR model in \eqref{model-vec} has $(NP+N-r_1)r_1$ parameters, where $r_1=\textup{rank}(\cm{A}_{(1)})$.
\end{remark}

Since $\bm{U}_1$ is orthonormal, it follows from \eqref{eq:MLRVAR_1} that
\begin{equation}
\label{eq:factormodel_summary}
\bm{U}_1'\bm{y}_t=\cm{G}_{(1)}(\bm{U}_3\otimes\bm{U}_2)'\bm{x}_t+\bm{U}_1'\bm{\epsilon}_t=\cm{G}_{(1)}\text{vec}(\bm{U}_2'\bm{X}_t\bm{U}_3)+\bm{U}_1'\bm{\epsilon}_t.
\end{equation}
The above representation reveals an interesting dynamic factor based interpretation for the proposed model. Specifically, $\bm{U}_1'\bm{y}_t:=\bm{f}_t^{\text{Response}}=(f_{1,t}^{\text{Response}}, \dots, f_{r_1,t}^{\text{Response}})^\prime\in\mathbb{R}^{r_1}$ represents  $r_1$ \textit{response factors} across the $N$ variables of  $\bm{y}_t$, where $f_{j,t}^{\text{Response}}=\bm{u}_{1,j}^\prime \bm{y}_t=\sum_{i=1}^{N}(\bm{U}_1)_{ij} y_{it}$ is the $j$th response factor, for $1\leq j\leq r_1$. Thus, if the $(i,j)$th entry of $\bm{U}_1$ is zero, i.e., $(\bm{U}_1)_{ij}=0$, then  $y_{it}$ is irrelevant to $\bm{f}_{j,t}^{\text{Response}}$.  In other words, $\bm{U}_1$ can be interpreted as the loadings of the response factors.

On the right side of \eqref{eq:factormodel_summary}, the predictor has the bilinear form $\bm{U}_2'\bm{X}_t\bm{U}_3$. On the one hand, $\bm{U}_2'\bm{X}_t:=\bm{F}_{t}^{\text{Predictor}}=(\bm{f}_{1,t}^{\text{Predictor}},\dots,\bm{f}_{r_2,t}^{\text{Predictor}})^\prime\in\mathbb{R}^{r_2\times P}$ represents $r_2$ \textit{predictor factors} across the $N$ variables (rows) of the predictor matrix $\bm{X}_t$, where $\bm{f}_{j,t}^{\text{Predictor}}=\sum_{i=1}^{N}(\bm{U}_2)_{ij} \bm{x}_{it}$ is the $j$th predictor factor, for $j=1,\dots, r_2$, with $\bm{x}_{it}=(y_{i,t-1},\ldots,y_{i,t-P})'$ for $1\leq i\leq N$. Hence,  if $(\bm{U}_2)_{ij}=0$, then  $\bm{x}_{it}$ is irrelevant to $\bm{f}_{j,t}^{\text{Predictor}}$.  On the other hand, $\bm{U}_3'\bm{X}_t':=\bm{F}_{t}^{\text{Lag}}=(\bm{f}_{1,t}^{\text{Lag}}, \dots, \bm{f}_{r_3,t}^{\text{Lag}})^\prime\in\mathbb{R}^{r_3\times N}$ represents $r_3$ \textit{temporal factors} across the $P$ time lags (columns) of the predictor matrix $\bm{X}_t$, where 
$\bm{f}_{j,t}^{\text{Lag}}=\sum_{i=1}^{P}(\bm{U}_3)_{ij} \bm{y}_{t-i}$ is the  $j$th temporal  factor, for $j=1,\dots, r_3$. As a result, $(\bm{U}_3)_{ij}=0$ implies that the $i$-th lagged predictor $\bm{y}_{t-i}$ is irrelevant to $\bm{f}_{j,t}^{\text{Lag}}$.  Therefore, $\bm{U}_2$ and $\bm{U}_3$  can be interpreted as the loadings of the predictor and temporal factors, respectively.

For simplicity, we call $r_1$, $r_2$ and $r_3$ the  response, predictor and temporal ranks, respectively.  Similar formulations can be found in matrix variate regressions \citep[e.g.,][]{Zhao_Leng2014, Ding_Cook2018}. The response, predictor and temporal factors interpretations of \eqref{eq:factormodel_summary} reveal that the proposed model is related to factor modeling, one of the most widely used techniques for high-dimensional time series.  We will explore the similarities and differences between them in the next subsection.

\subsection{Connections with factor modeling for time series \label{subsec:fm}}
In the literature, low-rank structures of high-dimensional time series are commonly explored through factor models \citep{SW05, BN08, SW11, BW16}.  The multilinear low-rank assumption of $\cm{A}$ in the proposed model fulfills a similar purpose as it extracts dynamic factors along three dimensions, as shown in our discussion about \eqref{eq:factormodel_summary}. Meanwhile, the proposed model
can  be used directly for forecasting, which is another attractive feature compared to factor models. In the following, we take a closer look at the factor structures of the proposed model and both static and dynamic factor models in the literature, and discuss some interesting connections between them.

The static factor model (SFM) is commonly written as
\begin{equation}\label{eq:factor_model_basic}
\bm{y}_t=\bm{\Lambda}\bm{f}_t+\bm{e}_t,
\end{equation}
where $\bm{y}_t\in\mathbb{R}^N$ is the observed time series, $\bm{f}_t\in\mathbb{R}^{r}$ are $r$ latent factors with $r\ll N$, $\bm{\Lambda}\in\mathbb{R}^{N\times r}$ is the factor loading matrix, and $\bm{e}_t\in\mathbb{R}^N$ is the random error. 
The usual normalization restrictions require that $\bm{F}^\prime \bm{F}/T=\bm{I}_{r}$ and that $\bm{\Lambda}^\prime\bm{\Lambda}\in\mathbb{R}^{r\times r}$ is a full-rank diagonal matrix, where $\bm{F}=(\bm{f}_1, \dots, \bm{f}_T)^\prime$; see \cite{BW16}.  

We can show that the proposed  model in \eqref{eq:MLRVAR_1} has an SFM representation. Specifically, as shown in Section \ref{appendix:factor} of the Appendix, there exist $\bm{\Lambda}\in\mathbb{R}^{N\times r_1}$ and  $\bm{f}_t\in\mathbb{R}^{r_1}$  such that
\begin{equation}
\label{eq:MLRVAR_factor}
\bm{y}_t=\bm{U}_1\cm{G}_{(1)}(\bm{U}_3\otimes\bm{U}_2)'\bm{x}_t+\bm{\epsilon}_t=\bm{\Lambda}\bm{f}_t+\bm{\epsilon}_t,
\end{equation}
for $t=1,\dots, T$, where $\bm{\Lambda}$ and the resulting $\bm{F}$ satisfy the aforementioned normalization restrictions, and $\bm{f}_t$ is the normalized version of $\cm{G}_{(1)}(\bm{U}_3\otimes\bm{U}_2)'\bm{x}_t\in\mathbb{R}^{r_1}$.
Let $\textup{span}(\cdot)$ denote the column space of a matrix, and it can be verified that $\textup{span}(\bm{\Lambda})=\textup{span}(\bm{U}_1)$.

\begin{remark}\label{remark:U_space}
Consider $\{\bm{y}_t\}$ generated by the proposed model. A useful by-product of representation \eqref{eq:MLRVAR_factor} is that  the low-dimensional subspace $\textup{span}(\bm{U}_1)$ can actually be estimated by  $\textup{span}(\widehat{\bm{\Lambda}})$, where $\widehat{\bm{\Lambda}}$ is the estimator of $\bm{\Lambda}$ obtained by fitting an SFM with $r=r_1$. Moreover, let $\bm{\Lambda}_1=\bm{\Lambda}(\bm{\Lambda}^\prime\bm{\Lambda})^{-1/2}$ be the orthonormalization of $\bm{\Lambda}$. Then $\textup{span}(\bm{\Lambda}_1)=\textup{span}(\bm{\Lambda})=\textup{span}(\bm{U}_1)$, and their orthogonal projectors are identical, namely $\bm{\Lambda}_1\bm{\Lambda}_1^\prime=\bm{U}_1\bm{U}_1^\prime$. Thus, the estimation error of $\textup{span}(\bm{U}_1)$ can be measured by the commonly-used subspace distance $\|\bm{\widehat{\Lambda}}_1\bm{\widehat{\Lambda}}_1'-\bm{U}_1\bm{U}_1'\|_\text{F}^2$, where $\bm{\widehat{\Lambda}}_1=\bm{\widehat{\Lambda}}(\bm{\widehat{\Lambda}}^\prime\bm{\widehat{\Lambda}})^{-1/2}$; see \cite{vu2013minimax}.
\end{remark}

On the other hand, the dynamic factor model (DFM) can be defined by combining model \eqref{eq:factor_model_basic} with  a certain dynamic structure,  e.g., the VAR, for the latent factor process $\bm{f}_t$ \citep{AW07}. To fix ideas, suppose that $\bm{f}_t$ evolves as the VAR(1),
\begin{equation}\label{eq:factor_VAR1}
\bm{f}_t = \bm{B} \bm{f}_{t-1} + \bm{\xi}_t,
\end{equation}
where $\bm{B}\in\mathbb{R}^{r\times r}$ is the transition matrix, and $\bm{\xi}_t\in\mathbb{R}^r$ is the random error. 
Let $\bm{w}_t=\bm{\Lambda}\bm{f}_t$ and $\bm{u}_t=\bm{\Lambda}\bm{\xi}_t$.  Then, the conjunction of \eqref{eq:factor_model_basic} and \eqref{eq:factor_VAR1} can also be written as
\begin{equation}\label{eq:VAR_ME}
\bm{y}_t=\bm{w}_t+\bm{e}_t, \quad \bm{w}_t=\bm{V}\bm{C}\bm{V}' \bm{w}_{t-1}+\bm{u}_t,
\end{equation}
where $\bm{D}=\bm{\Lambda}'\bm{\Lambda}$ is diagonal, $\bm{V}=\bm{\Lambda}\bm{D}^{-1/2}$ is orthonormal, and $\bm{C}=\bm{D}^{1/2}\bm{B}\bm{D}^{-1/2}\in\mathbb{R}^{r\times r}$. Interestingly, \eqref{eq:VAR_ME} resembles the VAR with measurement error, where $\bm{y}_t$ is the observed outcome of the true VAR(1) process $\bm{w}_t$ subject to measurement error $\bm{e}_t$. Note that the naive estimation ignoring the  measurement error of the autoregressive process will result in asymptotic biases; see, e.g., \cite{Staudenmayer05}.

However, if $\bm{e}_t=0$, we may gain more insights by comparing the DFM in \eqref{eq:VAR_ME} to the proposed model of lag order one. Note that when $P=1$ the latter reduces to the reduced-rank VAR,
\begin{equation*}
\bm{y}_t=\bm{U}_1\cm{G}_{(1)}\bm{U}_2'\bm{y}_{t-1}+\bm{\epsilon}_t \hspace{5mm}\text{or}\hspace{5mm} \bm{U}_1'\bm{y}_t=\cm{G}_{(1)}\bm{U}_2'\bm{y}_{t-1}+\bm{U}_1'\bm{\epsilon}_t,
\end{equation*}
with $\bm{U}_1$ and $\bm{U}_2$ being orthonormal and $r_1=r_2$, while the DFM model in \eqref{eq:VAR_ME} with $\bm{e}_t=0$ has the form of
\[
\bm{V}^\prime \bm{y}_t=\bm{C}\bm{V}' \bm{y}_{t-1}+\bm{V}^\prime\bm{u}_t.
\]
Hence, we may argue that the proposed model is more flexible than the DFM in \eqref{eq:VAR_ME}, as the former can accommodate different low-dimensional patterns for the response $y_t$ and predictors $y_{t-j}$s, whereas the latter  requires the subspaces of $y_t$ and $y_{t-j}$s to be identical. It is also worth noting that when $P>1$, another advantage of the proposed model is that it can capture the possible low-rank structure across time lags of the predictors; see \eqref{eq:factormodel_summary} in the previous subsection. 
Lastly, we note that the proposed model may be extended along the line of the factor augmented VAR models (FAVAR) \citep{BBE05} by incorporating known low-dimensional factors. 

\begin{remark}
In contrast to the proposed model, the classical factor model in the general form of  \eqref{eq:factor_model_basic} is not specific to VAR models, since it allows for  general latent factors.  However, the general factor model in \eqref{eq:factor_model_basic} cannot be directly used for forecasting unless an additional dynamic structure is imposed on the latent factor process, e.g., \eqref{eq:factor_VAR1}. As discussed above, if the multilinear low-rank assumption holds, the proposed model can be more favorable than the DFM. 
\end{remark}

\section{Low-dimensional time series modeling}

\subsection{Multilinear low-rank least squares estimation}

For the multilinear low-rank VAR model in \eqref{eq:MLRVAR} with ranks $(r_1,r_2,r_3)$,  the multilinear low-rank (MLR) least squares estimator can be defined as
\begin{equation}
\label{eq:LSE}
\cm{\widehat{A}}_{\textup{MLR}} \equiv [\![\cm{\widehat{G}};\bm{\widehat{U}}_1,\bm{\widehat{U}}_2,\bm{\widehat{U}}_3]\!] =
\argmin L(\cm{G},\bm{U}_1,\bm{U}_2,\bm{U}_3),
\end{equation}
where
\begin{equation}\label{eq:obj1}
L(\cm{G},\bm{U}_1,\bm{U}_2,\bm{U}_3)=\frac{1}{T}\sum_{t=1}^T \|\bm{y}_t-(\cm{G}\times_1\bm{U}_1\times_2\bm{U}_2\times_3\bm{U}_3)_{(1)}\bm{x}_t\|_2^2.
\end{equation}
We will derive asymptotic properties of $\cm{\widehat{A}}_{\textup{MLR}}$ when both $N$ and $P$ are fixed and the true multilinear ranks $(r_1,r_2, r_3)$ are known. Note that the minimization in \eqref{eq:LSE} is unconstrained, so the Tucker decomposition $[\![\cm{\widehat{G}};\bm{\widehat{U}}_1,\bm{\widehat{U}}_2,\bm{\widehat{U}}_3]\!]$ of $\cm{\widehat{A}}_{\text{MLR}}$ is not unique. 

Let $\bm{\phi}=(\text{vec}(\cm{G}_{(1)})^{\prime},\text{vec}(\bm{U}_1)',\text{vec}(\bm{U}_2)',\text{vec}(\bm{U}_3)')'$ be the true value of the vectorized HOSVD components and $\bm{\widehat{\phi}}_{\textup{MLR}} =(\text{vec}(\cm{\widehat{G}}_{(1)})^{\prime}, \text{vec}(\widehat{\bm{U}}_1)^{\prime}, \text{vec}(\widehat{\bm{U}}_2)^{\prime}, \text{vec}(\widehat{\bm{U}}_3)^{\prime})^{\prime}$ be the corresponding  estimator.
Let $\bm{h}(\bm{\phi})=\text{vec}(\cm{A}_{(1)})=\text{vec}(\bm{U}_1\cm{G}_{(1)}(\bm{U}_3\otimes\bm{U}_2)')$ be a function of $\bm{\phi}$.
Let $\bm{\Sigma}_{\bm{\epsilon}}=\textrm{var}(\bm{\epsilon}_t)$, $\bm{\Gamma}_j=\textrm{cov}(\bm{y}_{t+j},\bm{y}_t)$ with $j\geq 0$,
\[
\bm{\Gamma}^*=
\begin{bmatrix}
\bm{\Gamma}_0 & \bm{\Gamma}_1 & \dots & \bm{\Gamma}_{P-1}\\
\bm{\Gamma}_1' & \bm{\Gamma}_0 & \dots & \bm{\Gamma}_{P-2}\\
\vdots & \vdots & \ddots & \vdots\\
\bm{\Gamma}_{P-1}' & \bm{\Gamma}_{P-2}' & \dots & \bm{\Gamma}_0\\
\end{bmatrix},
\]
and $\bm{J}=\bm{\Sigma}_{\bm{\epsilon}}^{-1}\otimes\bm{\Gamma}^*$.
Denote
\begin{equation}\begin{split}
\bm{H}=\frac{\partial\bm{h}}{\partial\bm{\phi}}=\Big((\bm{U}_3\otimes\bm{U}_2\otimes\bm{U}_1),&[(\bm{U}_3\otimes\bm{U}_2)\cm{G}_{(1)}']\otimes\bm{I}_N,
\bm{T}_{21}\{[(\bm{U}_1\otimes\bm{U}_3)\cm{G}_{(2)}']\otimes\bm{I}_N\},\\
&\bm{T}_{31}\{[(\bm{U}_1\otimes\bm{U}_2)\cm{G}_{(3)}']\otimes\bm{I}_P\}\Big),
\end{split}\end{equation}
where $\bm{T}_{ij}$ is an $(N^2P)\times(N^2P)$ permutation matrix such that
$\text{vec}(\cm{A}_{(j)})=\bm{T}_{ij}\text{vec}(\cm{A}_{(i)})$ with $1\leq i,j\leq 3$.

\begin{theorem}
	\label{thm:Asymptotic}
	Suppose that the  time series $\{\bm{y}_t\}$ is generated by model \eqref{eq:MLRVAR} with $\mathbb{E}\|\bm{\epsilon}_t\|_2^4<\infty$, both $N$ and $P$ are fixed, and  $(r_1,r_2, r_3)$ are known. Then, under Assumption \ref{asmp:stationary},
	\begin{equation}
		\sqrt{T}\{\textup{vec}((\cm{\widehat{A}}_{\textup{MLR}})_{(1)})-\textup{vec}(\cm{A}_{(1)})\}\rightarrow N(\bm{0},\bm{\Sigma}_{\emph{MLR}})
	\end{equation}
	in distribution as $T\rightarrow\infty$, where $\bm{\Sigma}_{\emph{MLR}}=\bm{H}(\bm{H}'\bm{J}\bm{H})^{\dagger}\bm{H}'$, and $\dagger$ denotes the Moore-Penrose inverse.
\end{theorem}

The proof of Theorem \ref{thm:Asymptotic} relies on the technique  for deriving asymptotic distributions of overparameterized models in \cite{Shapiro86}. It does not require that $\cm{G}$ and $\bm{U}_j$s  are identifiable, nor does it require imposing identification constraints on the estimation in \eqref{eq:LSE}.

However, if we are further interested in estimating the true components $\cm{G}$ and $\bm{U}_j$s in the HOSVD of $\cm{A}$, the identifiability of these components, i.e., the uniqueness of the HOSVD,  will be required. This is guaranteed by the following assumption.

\begin{assumption}\label{asmp:svd}
	For each $1\leq j\leq 3$, (i) the singular values of $\cm{A}_{(j)}$ are distinct, and (ii) the first element in each column of $\bm{U}_j$ is positive.
\end{assumption} 

In Assumption 2, Condition (i) avoids indeterminacy of the factor loading vectors and holds generally in practice. Condition (ii) rules out sign switches in $\bm{U}_j$ and is commonly used in low-rank matrix models \citep{li2016supervised}.  

Accordingly, based on the unconstrained estimator $\cm{\widehat{A}}_{\textup{MLR}}$, we can define each $\bm{\widehat{U}}_j$ uniquely as the top $r_j$ left singular vectors of $(\cm{\widehat{A}}_{\textup{MLR}})_{(j)}$ such that the first element in each column of $\bm{\widehat{U}}_j$ is positive, and set $\cm{\widehat{G}}=[\![\cm{\widehat{A}}_{\textup{MLR}};\bm{\widehat{U}}_1',\bm{\widehat{U}}_2',\bm{\widehat{U}}_3']\!]$. As a result, the estimators $\cm{\widehat{G}}$ and $\bm{\widehat{U}}_j$s are consistent and asymptotically normal.

\begin{corollary}\label{cor1}
	Suppose that the conditions of Theorem \ref{thm:Asymptotic} and Assumption \ref{asmp:svd} hold. Then $\sqrt{T}\{\textup{vec}(\cm{\widehat{G}})-\textup{vec}(\cm{G})\}$, $\sqrt{T}\{\textup{vec}(\bm{\widehat{U}}_1)-\textup{vec}(\bm{U}_1)\}$,
	$\sqrt{T}\{\textup{vec}(\bm{\widehat{U}}_2)-\textup{vec}(\bm{U}_2)\}$, and $\sqrt{T}\{\textup{vec}(\bm{\widehat{U}}_3)-\textup{vec}(\bm{U}_3)\}$ converge to normal distributions with mean zero as $T\rightarrow\infty$.
\end{corollary}

The next corollary shows that the proposed estimator $\cm{\widehat{A}}_{\textup{MLR}}$ is asymptotically more efficient than the ordinary least squares (OLS) estimator
\begin{equation*}
\label{eq:OLS}
\bm{\widehat{A}}_{\textup{OLS}}=\argmin_{\bm{B}\in\mathbb{R}^{N\times NP}}\sum_{t=1}^T\|\bm{y}_t-\bm{B}\bm{x}_t\|_2^2
\end{equation*}
for the full VAR model in \eqref{eq:VAR} and the reduced-rank regression (RRR) estimator
\begin{equation*}
\label{eq:RRR}
\bm{\widehat{A}}_{\textup{RRR}}=\argmin_{\bm{B}\in\mathbb{R}^{N\times NP},\text{ rank}(\bm{B})\leq r_1}\sum_{t=1}^T\|\bm{y}_t-\bm{B}\bm{x}_t\|_2^2
\end{equation*}
for the  reduced-rank VAR model in \eqref{model-vec}, where $r_1$ is the rank of $\cm{A}_{(1)}$.
Denote by $\cm{\widehat{A}}_{\textup{OLS}}$ and $\cm{\widehat{A}}_{\textup{RRR}}$ the transition tensors formed by $\bm{\widehat{A}}_{\textup{OLS}}$ and $\bm{\widehat{A}}_{\textup{RRR}}$, respectively.

\begin{corollary}
	\label{cor:comparison}
Under the conditions of Theorem \ref{thm:Asymptotic},
$\sqrt{T}\{\textup{vec}((\cm{\widehat{A}}_{\textup{OLS}})_{(1)})-\textup{vec}(\cm{A}_{(1)})\}\rightarrow N(\bm{0},\bm{\Sigma}_{\emph{OLS}})$ and $\sqrt{T}\{\textup{vec}((\cm{\widehat{A}}_{\textup{RRR}})_{(1)})-\textup{vec}(\cm{A}_{(1)})\}\rightarrow N(\bm{0},\bm{\Sigma}_{\emph{RRR}})$ 	
	in distribution as $T\rightarrow\infty$.
	Moreover, it holds that $\bm{\Sigma}_{\emph{MLR}}\leq \bm{\Sigma}_{\emph{RRR}}\leq \bm{\Sigma}_{\emph{OLS}}$.
\end{corollary}

\subsection{Alternating least squares algorithm}

Let $\mathcal{F}_t=\sigma(\bm{\epsilon}_t,\bm{\epsilon}_{t-1},\cdots)$ be the $\sigma$-field generated by $\{\bm{\epsilon}_s,s\leq t\}$ and recall that $\bm{X}_t=(\bm{y}_{t-1},\dots,\bm{y}_{t-P})$. The objective function in \eqref{eq:LSE} is a nonlinear function of $\cm{G}$, $\bm{U}_1$, $\bm{U}_2$ and $\bm{U}_3$. However, from model \eqref{eq:MLRVAR}, we have
\begin{equation}
\begin{split}
\mathbb{E}(\bm{y}_t|\mathcal{F}_{t-1})&=\left ((\bm{x}_t'(\bm{U}_3\otimes\bm{U}_2)\cm{G}_{(1)}' )\otimes\bm{I}_N\right)\text{vec}(\bm{U}_1)\\
&=\bm{U}_1\cm{G}_{(1)}((\bm{U}_3'\bm{X}_t')\otimes\bm{I}_{r_2})\text{vec}(\bm{U}_2')\\
&=\bm{U}_1\cm{G}_{(1)}(\bm{I}_{r_3}\otimes(\bm{U}_2'\bm{X}_t))\text{vec}(\bm{U}_3)\\
&=(((\bm{U}_3\otimes\bm{U}_2)'\bm{x}_t)'\otimes\bm{U}_1)\text{vec}(\cm{G}_{(1)}),
\end{split}
\end{equation}
which implies that the objective function in \eqref{eq:LSE} is linear with respect to any of $\cm{G}$, $\bm{U}_1$, $\bm{U}_2$ and $\bm{U}_3$ when the other three are fixed.

Given the multilinear ranks $(r_1, r_2, r_3)$, we can employ Algorithm \ref{alg:LSE} to find $\cm{\widehat{A}}_{\textup{MLR}}$. Note that this is an alternating least squares algorithm where each step has a closed-form solution.
In practice,  the multilinear ranks need to be selected consistently, and we relegate the details to Section 5. The following proposition gives the convergence property of Algorithm \ref{alg:LSE}.

\begin{proposition}\label{prop3}
	Suppose that the stationary points of the objective function in \eqref{eq:LSE} are isolated, up to an arbitrary nonsingular linear transformation.
	Then $\bm{\phi}^{(k)}$ converges to a stationary point as $k\rightarrow\infty$, where $\bm{\phi}^{(k)}=(\textup{vec}(\cm{G}^{(k)})',\textup{vec}(\bm{U}_1^{(k)})',\textup{vec}(\bm{U}_2^{(k)})',\textup{vec}(\bm{U}_3^{(k)})')'$.
	Moreover, let $\bm{\phi}^{(\infty)}=(\textup{vec}(\cm{G}^{(\infty)})',\textup{vec}(\bm{U}_1^{(\infty)})',\textup{vec}(\bm{U}_2^{(\infty)})',\textup{vec}(\bm{U}_3^{(\infty)})')'$ be a strict local minimum of the objective function.
	Then $\{\bm{\phi}^{(k)}\}$ will be attracted to $\bm{\phi}^{(\infty)}$ if the initial value $\bm{\phi}^{(0)}$ is sufficiently close to $\bm{\phi}^{(\infty)}$.
\end{proposition}

\begin{algorithm}[t]
	\caption{Alternating least squares algorithm for $\cm{\widehat{A}}_{\textup{MLR}}$}
	\label{alg:LSE}
	Initialize: $\cm{A}^{(0)}$\\ 
	HOSVD: $\cm{A}^{(0)}\approx\cm{G}^{(0)}\times_1\bm{U}_1^{(0)}\times_2\bm{U}_2^{(0)}\times_3\bm{U}_3^{(0)}$ with multilinear ranks $(r_1,r_2,r_3)$\\
	\textbf{repeat} $k=0,1,2,\dots$\\
	\hspace*{1cm}$\bm{U}_1^{(k+1)}\leftarrow\arg\min_{\bm{U}_1}\sum_{t=1}^T\|\bm{y}_t-  ((\bm{x}_t'(\bm{U}^{(k)}_3\otimes\bm{U}^{(k)}_2)\cm{G}^{(k)\prime}_{(1)})\otimes\bm{I}_N)\text{vec}(\bm{U}_1)\|_2^2$\\
	\hspace*{1cm}$\bm{U}_2^{(k+1)} \leftarrow \arg\min_{\bm{U}_2}\sum_{t=1}^T\|\bm{y}_t-\bm{U}_1^{(k+1)}\cm{G}^{(k)}_{(1)}((\bm{X}_t\bm{U}_3^{(k)})'\otimes\bm{I}_{r_2})\text{vec}(\bm{U}_2')\|_2^2$\\
	\hspace*{1cm}$\bm{U}_3^{(k+1)} \leftarrow \arg\min_{\bm{U}_3}\sum_{t=1}^T\|\bm{y}_t-\bm{U}_1^{(k+1)}\cm{G}^{(k)}_{(1)}(\bm{I}_{r_3}\otimes(\bm{U}_2^{(k+1)'}\bm{X}_t))\text{vec}(\bm{U}_3)\|_2^2$\\
	\hspace*{1cm}$\cm{G}^{(k+1)}\leftarrow \arg\min_{\footnotesize{\cm{G}}}\sum_{t=1}^T\|\bm{y}_t-(((\bm{U}_3^{(k+1)}\otimes\bm{U}_2^{(k+1)})'\bm{x}_t)'\otimes \bm{U}_1^{(k+1)})\text{vec}(\cm{G}_{(1)})\|_2^2$\\
	\hspace*{1cm}$\cm{A}^{(k+1)}\leftarrow\cm{G}^{(k+1)}\times_1\bm{U}_1^{(k+1)}\times_2\bm{U}_2^{(k+1)}\times_3\bm{U}_3^{(k+1)}$\\
	\textbf{until convergence}\\
	Finalize: $\bm{\widehat{U}}_i\leftarrow$ top $r_i$ left singular vectors of $\cm{\widehat{A}}_{(i)}$ with positive first elements, $1\leq i\leq 3$\\
	\hspace*{1.6cm} $\cm{\widehat{G}}\leftarrow [\![\cm{\widehat{A}};\bm{\widehat{U}}_1',\bm{\widehat{U}}_2',\bm{\widehat{U}}_3']\!]$
\end{algorithm}

\begin{remark}\label{remark:init1}
If the sample size is sufficiently large, by Corollary \ref{cor:comparison},  $\cm{\widehat{A}}_{\textup{OLS}}$ can be used as the initial value $\cm{A}^{(0)}$ of Algorithm \ref{alg:LSE}. For smaller sample sizes,  $\cm{\widehat{A}}_{\textup{RRR}}$ or the nuclear norm estimator to be  discussed in  Section 5 can be employed instead. 
Moreover, Algorithm \ref{alg:LSE} does not guarantee convergence to the global solution defined in \eqref{eq:LSE}. As a result, in practice, we recommend a random initialization method with $\cm{A}^{(0)}=\cm{\widehat{A}}_{\textup{pre}}+T^{-1/2}\cm{T}$, where $\cm{\widehat{A}}_{\textup{pre}}$ is a preliminary estimate, say, $\cm{\widehat{A}}_{\textup{OLS}}$ or $\cm{\widehat{A}}_{\textup{RRR}}$, and $\cm{T}\in\mathbb{R}^{N\times N\times P}$ is a random perturbation whose entries are drawn independently from $N(0,1)$. Many randomized initial values can be tried, and the solution which yields the smallest value for the objective function will be adopted. 
\end{remark}

\begin{remark}
Algorithm \ref{alg:LSE} corresponds to the unconstrained estimation in \eqref{eq:LSE}. Thus, we do not need the orthogonality constraints of $\cm{G}$ and $\bm{U}_i$s. The unidentifiability of the Tucker decomposition does not affect the convergence of the algorithm, since Proposition \ref{prop3} does not require that the convergent sequence $\bm{\phi}^{(k)}$ is unique. Moreover, note that the final estimates $\cm{\widehat{G}}$ and $\bm{\widehat{U}}_i$s in Algorithm \ref{alg:LSE} are obtained from the unconstrained estimate of  $\cm{A}$, which is consistent with the definitions of $\cm{\widehat{G}}$ and $\bm{\widehat{U}}_i$s in Corollary \ref{cor1}. Similar alternating algorithms without imposing identification constraints can be found in the literature of tensor decomposition; see, e.g. \citet{Zhou13} and \cite{Li13}.
\end{remark}

\section{High-dimensional time series modeling}

\subsection{Sparse higher-order reduced-rank VAR}

As discussed in Section 2.2, the proposed model can effectively capture the dynamic information  along three dimensions by response, predictor and temporal factors, with  $\bm{U}_1, \bm{U}_2$ and $\bm{U}_3$ representing the the corresponding factor loadings. 
However, when the dimensions $N$ and/or $P$ are very large,  the fitted loading matrices often contain many small values, indicating relatively insignificant contribution of certain variables or lags to the factors. 
For example, if the $(i,j)$th entry of $\bm{U}_1$ is very small, then $y_{it}$ may be irrelevant to the $j$th response factor, with $1\leq i\leq N$ and $1\leq j \leq r_1$; see also the discussion below \eqref{eq:factormodel_summary} for similar interpretations regarding $\bm{U}_2$ and $\bm{U}_3$.

To improve the interpretability, we may shrink the small values in the factor  loading matrices to zero by imposing  sparsity assumptions on $\bm{U}_i$s. This allows us to substantially reduce the number of unknown parameters  while performing data-driven variable selection for each factor, and hence the estimation efficiency is also improved; see
\cite{ChenChan12} and \cite{Uematsu17}.

Specifically, we introduce the following $\ell_1$-penalized Sparse Higher-Order Reduced-Rank (SHORR) estimator:
\begin{equation}
\label{eq:SparseHOSVD}
\begin{split}
&\cm{\widehat{A}}_{\textup{SHORR}}\equiv[\![\cm{\widehat{G}};\bm{\widehat{U}}_1,\bm{\widehat{U}}_2,\bm{\widehat{U}}_3]\!]
=\underset{{\scriptsize\cm{G}},\bm{U}_1,\bm{U}_2,\bm{U}_3}{\argmin}\left\{L(\cm{G},\bm{U}_1,\bm{U}_2,\bm{U}_3)
+\lambda\|\bm{U}_3\otimes\bm{U}_2\otimes\bm{U}_1\|_1\right\}
\end{split}
\end{equation}
subject to 
\begin{equation}\label{eq:ortho}
	\cm{G}\in\text{AO}(r_1,r_2,r_3)\hspace{5mm}\text{and}\hspace{5mm}\bm{U}_i'\bm{U}_i=\bm{I}_{r_i},\hspace{5mm} i=1,2,3,
\end{equation}
where
$L(\cm{G},\bm{U}_1,\bm{U}_2,\bm{U}_3)$ is defined as in \eqref{eq:obj1}, and $\text{AO}(r_1,r_2,r_3)=\{\cm{G}\in\mathbb{R}^{r_1\times r_2\times r_3}: \cm{G}_{(i)}\text{ is row-orthogonal},~i=1,2,3\}$.
Unlike the unconstrained estimation in \eqref{eq:LSE}, the orthogonality constraints in \eqref{eq:ortho} are necessary; otherwise, the sparsity patterns of $\bm{U}_i$ cannot be identified. As in Section 3, we  will derive the statistical properties of the proposed estimator under the true multilinear ranks $(r_1, r_2, r_3)$, while a consistent rank selection procedure will be discussed in Section 5.

\begin{remark}
The proposed SHORR estimation method is different from the row-sparse reduced-rank regression that has been studied extensively in the literature \citep{ChenHuang12,Bunea12}.  We avoid imposing the row-sparsity because (1) it would restrict the flexibility and interpretability of the VAR model, and (2) with a row-sparse response factor matrix $\bm{U}_1$, those unselected time series cannot be predicted at all. Thus, we consider the general sparsity structure for  $\bm{U}_i$s rather than the row-sparsity. 
\end{remark}

\begin{remark}
Alternatively, one  might consider penalizing each $\bm{U}_i$ individually, e.g. with the penalty term $\sum_{i=1}^{3}\lambda_i\|\bm{U}_i\|_1$. Unfortunately, the three tuning parameters will bring about much higher computational costs and significant theoretical difficulties. To circumvent this problem, the SHORR estimator induces sparsity for $\bm{U}_1$, $\bm{U}_2$ and $\bm{U}_3$ jointly since $\|\bm{U}_3\otimes\bm{U}_2\otimes\bm{U}_1\|_1=\|\bm{U}_3\|_1\|\bm{U}_2\|_1\|\bm{U}_1\|_1$. 
Implementation of this joint penalty is convenient through the alternating algorithm to be introduced in Section 4.3. Similar ideas of joint penalization can be found in the literature, e.g., the joint Lasso penalty in \citet{Zhao_Leng2014} and the joint penalty for left and right singular vectors for sparse SVD  in \citet{ChenChan12}. 
Moreover, when $P$ is relatively small, we might wish to impose sparsity on $\bm{U}_1$ and $\bm{U}_2$ only, and then $\|\bm{U}_3\otimes\bm{U}_2\otimes\bm{U}_1\|_1$ can be replaced by $\|\bm{U}_2\otimes\bm{U}_1\|_1$.
\end{remark}

\subsection{Theoretical properties of the SHORR estimator}
To derive the non-asymptotic estimation and prediction error bounds of the SHORR estimator, we make the following assumptions.

\begin{assumption} (Gaussian error)
	\label{asmp:subgau}
	The errors $\{\bm{\epsilon}_t\}$ are $i.i.d.$ Gaussian random vectors with mean zero and positive definite covariance matrix $\bm{\Sigma}_{\bm{\epsilon}}$.
\end{assumption}

\begin{assumption} (Sparsity)
	\label{asmp:SSS}
	Each column of the factor matrices $\bm{U}_i$ has at most $s_i$ nonzero entries, for $i=1,2,3$.
\end{assumption}

\begin{assumption} (Restricted parameter space)
	\label{asmp:RPS}
	The parameter space for $\cm{G}$ and $\bm{U}_i$ with $1\leq i\leq 3$ is $\Omega=\{\cm{G}\in\textup{AO}(r_1,r_2,r_3):\sigma_1(\cm{G}_{(j)})\leq \bar{g}<\infty,~\textup{for}~1\leq j\leq 3\}\times\mathcal{U}_1\times\mathcal{U}_2\times\mathcal{U}_3$, where $\mathcal{U}_i=\{\bm{U}\in\mathbb{R}^{p_i\times r_i}:\bm{U}'\bm{U}=\bm{I}_{r_i},~\textup{and}~\bm{U}_{ij}^2\geq\nu>0~\textup{or}~\bm{U}_{ij}=0\}$ with $p_1=p_2=N$ and $p_3=P$,
	and $\nu$ is a uniform lower threshold for elements of $\bm{U}_{i}$s.
\end{assumption}

\begin{assumption} (Relative spectral gap)
	\label{asmp:RSG}
	The nonzero singular values of $\cm{A}_{(i)}$ satisfy that $\sigma^2_{j-1}(\cm{A}_{(i)})-\sigma^2_{j}(\cm{A}_{(i)})\geq\delta\sigma^2_{j-1}(\cm{A}_{(i)})$ for $2\leq j\leq r_i$ and $1\leq i \leq 3$, where $\delta$ is a positive constant.
\end{assumption}

Assumption \ref{asmp:subgau} enables us to apply the concentration inequalities for VAR models in \citet{Basu15}. The Gaussian condition may be relaxed to sub-Gaussianity by techniques in \citet{zheng2019testing}. Assumption \ref{asmp:SSS} states the sparsity of each factor matrix. Assumption \ref{asmp:RPS} imposes an upper bound on the core tensor $\cm{G}$, which is not a stringent assumption since large singular values in $\cm{G}$ could cause nonstationarity of the VAR process. The lower threshold $\nu$ for the $\bm{U}_i$s is essential to restrict the estimation error to a subspace such that the restricted eigenvalue condition \citep{Bickel09} can be established. Note that $\nu$ may shrink to zero as the dimension increases, so this condition is not too stringent.  Assumption \ref{asmp:RSG} guarantees that the singular values of each $\cm{A}_{(i)}$  are well separated. This rules out unidentifiability and allows us to derive the upper bound for the perturbation errors in Lemma \ref{lemma:perturbation} in Section \ref{append:lems} of the Appendix.

Assumption \ref{asmp:stationary} guarantees that the eigenvalues of the Hermitian matrix $\mathcal{A}^*(z)\mathcal{A}(z)$ over the unit circle $\{z\in\mathbb{C}:|z|=1\}$ are all positive, where $\mathcal{A}^*(z)$ denotes the conjugate transpose of $\mathcal{A}(z)$. Following \citet{Basu15}, let
\[
\mu_{\min}(\mathcal{A})=\underset{|z|=1}{\min}\lambda_{\min}(\mathcal{A}^*(z)\mathcal{A}(z))
\hspace{5mm}\text{and}\hspace{5mm}
\mu_{\max}(\mathcal{A})=\underset{|z|=1}{\max}\lambda_{\max}(\mathcal{A}^*(z)\mathcal{A}(z)),
\]
where $\lambda_{\min}(\cdot)$ and $\lambda_{\max}(\cdot)$ denote the minimum and maximum eigenvalues of a matrix, respectively.
It holds that
\begin{equation}
\mu_{\min}(\mathcal{A})=\underset{\theta\in[-\pi,\pi]}{\min}\lambda_{\min}\left(\left(\bm{I}_N-\sum_{p=1}^P\bm{A}_p'e^{ip\theta}\right)\left(\bm{I}_N-\sum_{p=1}^P\bm{A}_p'e^{-ip\theta}\right)\right)
\end{equation}
and
\begin{equation}
\mu_{\max}(\mathcal{A})=\underset{\theta\in[-\pi,\pi]}{\max}\lambda_{\max}\left(\left(\bm{I}_N-\sum_{p=1}^P\bm{A}_p'e^{ip\theta}\right)\left(\bm{I}_N-\sum_{p=1}^P\bm{A}_p'e^{-ip\theta}\right)\right).
\end{equation}

\begin{theorem}
	\label{thm:errorbound}
	Suppose that Assumptions \ref{asmp:stationary} and 3-6 hold, and  $(r_1,r_2, r_3)$ are known. If $\lambda\gtrsim\mathcal{M}\sqrt{\log(N^2P)/T}$ and  $T\gtrsim\log(N^2P)+\mathcal{M}^2d\min[\log(NP),\log(cNP/d)]$,
	then
	\begin{equation}
	\label{eq:L2est}
	\|\cm{\widehat{A}}_{\textup{SHORR}}-\cm{A}\|_{\textup{F}}\leq C_1\tau\sqrt{S}\lambda/\alpha,
	\end{equation}
	and
	\begin{equation}
	\label{eq:L2pred}
	T^{-1}\sum_{t=1}^T\|(\cm{\widehat{A}}_{\textup{SHORR}}-\cm{A})_{(1)}\bm{x}_t\|_2^2\leq C_2\tau^2S\lambda^2/\alpha,
	\end{equation}
	with probability at least $1-C\exp[-c\log(N^2P)]-C\exp\{-cd\min[\log(NP),\log(cNP/d)]\}$,
	where $c, C, C_1, C_2>0$ are absolute constants, $\mathcal{M}=\lambda_{\max}(\bm{\Sigma}_{\bm{\epsilon}})\left(1+\mu_{\max}(\mathcal{A})/\mu_{\min}(\mathcal{A})\right)$,  $d=\nu^{-2}r_1r_2r_3$, $\tau=\delta^{-1}r_1r_2r_3\sum_{i=1}^3\eta_i/\sqrt{r_i}$ with $\eta_i=(\sum_{j=1}^{r_i}\sigma_1^2(\cm{A}_{(i)})/\sigma_j^2(\cm{A}_{(i)}))^{1/2}$,
	$S=s_1s_2s_3$ and $\alpha=\lambda_{\min}(\bm{\Sigma}_{\bm{\epsilon}})/\mu_{\max}(\mathcal{A})$.
\end{theorem}

Theorem \ref{thm:errorbound} gives the non-asymptotic error upper bounds under high-dimensional scaling. When the multilinear ranks $(r_1, r_2, r_3)$ and lower threshold $\nu$ are fixed, \eqref{eq:L2est} shows that $\cm{\widehat{A}}_{\textup{SHORR}}$ is a consistent estimator if $T\gtrsim S\log(N^2P)$. In this setting, the estimation and prediction error bounds in \eqref{eq:L2est} and \eqref{eq:L2pred} become $O_p(\sqrt{S\log(N^2P)/T})$ and $O_p(S\log(N^2P)/T)$, respectively.

\begin{remark}
\citet{Basu15} considers estimation of stationary Gaussian VAR($P$) models with sparse transition matrices  such that $\|\cm{A}\|_0=k$. For the Lasso estimator $\cm{\widehat{A}}_{\textup{LASSO}}=\argmin T^{-1}\sum_{t=1}^T\|\bm{y}_t-\cm{A}_{(1)}\bm{x}_t\|_2^2+\lambda\|\cm{A}_{(1)}\|_1$, it was shown that  $\|\cm{\widehat{A}}_\textup{LASSO}-\cm{A}\|_\textup{F}\lesssim \sqrt{k\log(N^2P)/T}$ and $T^{-1}\sum_{t=1}^T\|(\cm{\widehat{A}}_\textup{LASSO}-\cm{A})_{(1)}\bm{x}_t\|_2^2\lesssim k\log(N^2P)/T$ with high probability, which are consistent with the regular error bounds for the Lasso as $N^2 P$ corresponds to the number of parameters \citep[e.g.,][]{wang2015high}. In contrast, we assume that $\cm{A}$ admits an HOSVD with sparse factor matrices $\bm{U}_i$, but $\cm{A}$ itself is not necessarily sparse. When each $\bm{U}_i$ is row-sparse with $s_i$ nonzero rows, for $i=1,2,3$, it can be checked that $\cm{A}$ is a sparse tensor with sparsity level $S$. In this case, the SHORR estimator has the same error bounds as the Lasso estimator. However, in the general case, even when $\cm{A}$ has a sparse HOSVD,  $\cm{A}$ may not be highly sparse, i.e., $k$ is larger than $S$, so $\cm{\widehat{A}}_{\textup{SHORR}}$ may be more efficient than  $\cm{\widehat{A}}_\textup{LASSO}$.
\end{remark}

\subsection{ADMM algorithm}

There are two major challenges in developing an efficient algorithm for the SHORR estimator. First, the core tensor $\cm{G}$ is subject to the  all-orthogonal constraint in \eqref{eq:ortho} which cannot be handled in a straightforward way. Second, the $\ell_1$-regularization in \eqref{eq:SparseHOSVD}  and the orthogonality constraints  in \eqref{eq:ortho} are imposed jointly on  $\bm{U}_i$s. The former is nonsmooth while the latter is nonconvex.  To deal with  these challenges, we adopt the alternating direction method of multipliers (ADMM) algorithm \citep{Boyd11} to update $\bm{U}_i$s and $\cm{G}$ alternatingly; see Algorithm \ref{alg:ADMM}. 

\algsetup{indent=2em}
\begin{algorithm}[t]
	\caption{ADMM algorithm for SHORR estimator}
	\label{alg:ADMM}
	\begin{algorithmic}[1]
		\STATE Initialize: $\cm{A}^{(0)}$\\ 
		\STATE HOSVD: $\cm{A}^{(0)}\approx\cm{G}^{(0)}\times_1\bm{U}_1^{(0)}\times_2\bm{U}_2^{(0)}\times_3\bm{U}_3^{(0)}$ with multilinear ranks $(r_1,r_2,r_3)$.
		\REPEAT 
		\STATE$\bm{U}_1^{(k+1)}\leftarrow\underset{\bm{U}_1'\bm{U}_1=\bm{I}_{r_1}}{\arg\min}\left\{L(\cm{G}^{(k)},\bm{U}_1,\bm{U}_2^{(k)},\bm{U}_3^{(k)})+\lambda\|\bm{U}_1\|_1\|\bm{U}_2^{(k)}\|_1\|\bm{U}_3^{(k)}\|_1\right\}$
		\STATE$\bm{U}_2^{(k+1)}\leftarrow\underset{\bm{U}_2'\bm{U}_2=\bm{I}_{r_2}}{\arg\min}\left\{L(\cm{G}^{(k)},\bm{U}_1^{(k+1)},\bm{U}_2,\bm{U}_3^{(k)})+\lambda\|\bm{U}_1^{(k+1)}\|_1\|\bm{U}_2\|_1\|\bm{U}_3^{(k)}\|_1\right\}$
		\STATE$\bm{U}_3^{(k+1)}\leftarrow\underset{\bm{U}_3'\bm{U}_3=\bm{I}_{r_3}}{\arg\min}\left\{L(\cm{G}^{(k)},\bm{U}_1^{(k+1)},\bm{U}_2^{(k+1)},\bm{U}_3)+\lambda\|\bm{U}_1^{(k+1)}\|_1\|\bm{U}_2^{(k+1)}\|_1\|\bm{U}_3\|_1\right\}$                         
		\STATE$\cm{G}^{(k+1)}\leftarrow \arg\min\Big\{ L(\cm{G},\bm{U}_1^{(k+1)},\bm{U}_2^{(k+1)},\bm{U}_3^{(k+1)})+\sum_{i=1}^3\varrho_i\|\cm{G}_{(i)}-\bm{D}_i^{(k)}\bm{V}_i^{(k)\prime}$\\
		\hspace{35mm}$+(\cm{C}_i^{(k)})_{(i)}\|_{\textup{F}}^2\Big\}$\\
		\FOR{$i\in\{1,2,3\}$}
		\STATE$\bm{D}_i^{(k+1)}\leftarrow\underset{\bm{D}_i=\diag(\bm{d}_i)}{\argmin}\|\cm{G}_{(i)}^{(k+1)}-\bm{D}_i\bm{V}_i^{(k)\prime}+(\cm{C}_i^{(k)})_{(i)}\|_{\textup{F}}^2$
		\STATE$\bm{V}_i^{(k+1)}\leftarrow\underset{\bm{V}_i'\bm{V}_i=\bm{I}_{r_i}}{\argmin}\|\cm{G}_{(i)}^{(k+1)}-\bm{D}_i^{(k+1)}\bm{V}_i^\prime+(\cm{C}_i^{(k)})_{(i)}\|_{\textup{F}}^2$
		\STATE$(\cm{C}_{i}^{(k+1)})_{(i)}\leftarrow(\cm{C}_{i}^{(k)})_{(i)}+\cm{G}_{(i)}^{(k+1)}-\bm{D}_i^{(k+1)}\bm{V}_i^{(k+1)\prime}$
		\ENDFOR
		\STATE$\cm{A}^{(k+1)}\leftarrow\cm{G}^{(k+1)}\times_1\bm{U}_1^{(k+1)}\times_2\bm{U}_2^{(k+1)}\times_3\bm{U}_3^{(k+1)}$
		\UNTIL{\textbf{convergence}}
	\end{algorithmic}
\end{algorithm}

Firstly,  to tackle the all-orthogonal constraint of $\cm{G}$, our idea is to separate it into three orthogonality constraints on the matricizations $\cm{G}_{(i)}$ for $1\leq i\leq 3$. This is to say that $\cm{G}_{(i)}$ can be decomposed as
$\cm{G}_{(i)}=\bm{D}_i\bm{V}_i'$,
where $\bm{D}_i\in\mathbb{R}^{r_i\times r_i}$  is a diagonal matrix,  and  $\bm{V}_1\in\mathbb{R}^{r_2 r_3\times r_1}$, $\bm{V}_2\in\mathbb{R}^{r_1 r_3\times r_2}$, and $\bm{V}_3\in\mathbb{R}^{r_1 r_2\times r_3}$ are orthonormal matrices with $\bm{V}_i'\bm{V}_i=\bm{I}_{r_i}$. 
Then, the augmented Lagrangian corresponding to the objective function in \eqref{eq:SparseHOSVD}  can be written as
\begin{equation*}
\begin{split}
\mathcal{L}_{\bm{\varrho}}(\cm{G},\{\bm{U}_i\},\{\bm{D}_i\},\{\bm{V}_i\};\{\cm{C}_i\})=&L(\cm{G},\bm{U}_1,\bm{U}_2,\bm{U}_3)+\lambda\|\bm{U}_3\otimes\bm{U}_2\otimes\bm{U}_1\|_1\\
&+2\sum_{i=1}^3\varrho_i\langle(\cm{C}_i)_{(i)},\cm{G}_{(i)}-\bm{D}_i\bm{V}_i'\rangle+\sum_{i=1}^3\varrho_i\|\cm{G}_{(i)}-\bm{D}_i\bm{V}_i'\|_\textup{F}^2,
\end{split}
\end{equation*}
where  $\cm{C}_1$, $\cm{C}_2, \cm{C}_3 \in\mathbb{R}^{r_1\times r_2\times r_3}$ are the tensor-valued  dual variables, and $\bm{\varrho}=(\varrho_1,\varrho_2,\varrho_3)'$ is the set of regularization parameters. This leads us to Algorithm \ref{alg:ADMM}. Note that all-orthogonal constraint of  $\cm{G}$ has been transferred to the matrices $\bm{V}_i$s in line 10, so no constraint is needed for updating $\cm{G}$ in line 7 of Algorithm \ref{alg:ADMM}.

Secondly, we consider the update of $\bm{U}_i$s. Since $L(\cm{G},\bm{U}_1,\bm{U}_2,\bm{U}_3)$ in \eqref{eq:obj1} is a least squares loss function with respect to each $\bm{U}_i$,  the $\bm{U}_i$-update steps in lines 4-6 of Algorithm \ref{alg:ADMM} are $\ell_1$-regularized least squares problems subject to an orthogonality constraint, which can be written in the general form:
\begin{equation}
\label{eq:SparseOrthogonal}
\underset{B}{\min}\left\{n^{-1}\|\bm{y}-\bm{X}\text{vec}(\bm{B})\|_2^2+\lambda\|\bm{B}\|_1\right\},~~\text{s.t.}~\bm{B}'\bm{B}=\bm{I}.
\end{equation}
Since the $\ell_1$-regularization and the orthogonality constraint for $\bm{B}$ are difficult to handle jointly, we adopt an  ADMM subroutine to separate them into two steps.
Specifically, we introduce the dummy variable  $\bm{W}$  as a surrogate for $\bm{B}$ and write problem \eqref{eq:SparseOrthogonal} into the equivalent form as follows:
\begin{equation}
\underset{\bm{B},\bm{W}}{\min}\{n^{-1}\|\bm{y}-\bm{X}\text{vec}(\bm{B})\|_2^2+\lambda\|\bm{W}\|_1\},~~\text{s.t.}~\bm{B}'\bm{B}=\bm{I} ~\text{and}~ \bm{B}=\bm{W}.
\end{equation}
Then the corresponding  augmented Lagrangian formulation is
\begin{equation}\label{eq:augl}
\underset{\bm{B},\bm{W}}{\min}\{n^{-1}\|\bm{y}-\bm{X}\text{vec}(\bm{B})\|_2^2+\lambda\|\bm{W}\|_1+2\kappa\langle \bm{M},\bm{B}-\bm{W}\rangle+\kappa\|\bm{B}-\bm{W}\|_{\textup{F}}^2\},
\end{equation}
where $\bm{M}$ is the dual variable, and $\kappa$ is a regularization parameter. The ADMM subroutine for \eqref{eq:augl} is presented in Algorithm \ref{alg:admm}. This yields solutions to the $\bm{U}_i$-update subproblems in Algorithm \ref{alg:ADMM}.

Note that the $\bm{B}$-update step in Algorithm \ref{alg:admm} and the $\bm{V}_i$-update step in line 10 of Algorithm \ref{alg:ADMM} are least squares problems with an orthogonality constraint.  
Hence, they can be solved efficiently by the splitting orthogonality constraint (SOC) method \citep{Lai14}. The $\bm{W}$-update step in Algorithm \ref{alg:admm} is an $\ell_1$-regularized minimization, which can be solved by the explicit soft-thresholding. The $\cm{G}$- and $\bm{D}_i$-update steps in lines 7 and 9 of Algorithm \ref{alg:ADMM} are simple least squares problems.

\algsetup{indent=2em}
\begin{algorithm}[t]
	\caption{ADMM subroutine for sparse and orthogonal regression}
	\label{alg:admm}
	\begin{algorithmic}[1]
		\STATE Initialize: $\bm{B}^{(0)}=\bm{W}^{(0)}$, $\bm{M}^{(0)}=\bm{0}$\\
		\REPEAT 
		\STATE$\bm{B}^{(k+1)}\leftarrow\arg\min_{\bm{B}^\prime \bm{B}=\bm{I}}\left\{n^{-1}\|\bm{y}-\bm{X}\text{vec}(\bm{B})\|^2_2+\kappa\|\bm{B}-\bm{W}^{(k)}+\bm{M}^{(k)}\|^2_{\textup{F}}\right\}$
		\STATE$\bm{W}^{(k+1)}\leftarrow\arg\min_{\bm{W}}\left\{\kappa\|\bm{B}^{(k+1)}-\bm{W}+\bm{M}^{(k)}\|^2_{\textup{F}}+\lambda\|\bm{W}\|_1\right\}$
		\STATE$\bm{M}^{(k+1)}\leftarrow\bm{M}^{(k)}+\bm{B}^{(k+1)}-\bm{W}^{(k+1)}$                        
		\UNTIL{\textbf{convergence}}
	\end{algorithmic}
\end{algorithm}

For general nonconvex problems, it is well known that ADMM algorithms need not converge, and even if they do, they need not converge to an optimal solution. A comprehensive algorithmic convergence analysis for Algorithm \ref{alg:ADMM} is challenging due to both the nested ADMM subroutine, Algorithm \ref{alg:admm},  and its interplay with the outer loop of Algorithm \ref{alg:ADMM}. 

\cite{WangYin19} gives a rigorous convergence analysis of multi-block ADMMs for nonconvex nonsmooth optimization with linear equality constraints. Their theory would be applicable to  Algorithm \ref{alg:admm} if the $\bm{B}$-update step in line 3 were exact. 
The extension to the inexact $\bm{B}$-update step would require a sophisticated analysis of the optimization error of the SOC method. We do not delve into the development of the convergence theory further in this paper.
Nonetheless, similarly to the analysis in \cite{Uematsu17}, under some high-level assumptions on $\mathcal{L}_{\bm{\varrho}}(\cdot)$, we can still  obtain the following convergence result for Algorithm \ref{alg:ADMM}. 

\begin{proposition}\label{prop4}
Let $\Delta\mathcal{L}_{\bm{\varrho}}(\cdot)$ be the decrease in the augmented Lagrangian $\mathcal{L}_{\bm{\varrho}}(\cdot)$ by a block update.
If $\sum_{k=1}^{\infty} \{[\Delta \mathcal{L}_{\bm{\varrho}}(\cm{G}^{(k)})]^{1/2}+\sum_{i=1}^3[\Delta \mathcal{L}_{\bm{\varrho}}(\bm{U}_i^{(k)})]^{1/2}+\sum_{i=1}^3[\Delta \mathcal{L}_{\bm{\varrho}}(\bm{D}_i^{(k)})]^{1/2}+\sum_{i=1}^3[\Delta \mathcal{L}_{\bm{\varrho}}(\bm{V}_i^{(k)})]^{1/2}\}<\infty$, then the sequence generated by Algorithm \ref{alg:ADMM} converges to a local solution of problem \eqref{eq:SparseHOSVD}.
\end{proposition}

\begin{remark}\label{remark:init2}
The initial value $\cm{A}^{(0)}$ for Algorithm 2 can be set to   the nuclear norm (NN) estimator $\cm{\widehat{A}}_{\textup{NN}}$ for low-rank VAR models \citep{Negahban11}, and it holds $\lVert\cm{\widehat{A}}_{\textup{NN}} - \cm{A} \rVert_{\textup{F}}=O_p(\sqrt{r_1NP/T})$; see also Section 5. Consequently, if one searches the SHORR estimator within a neighborhood of $\cm{\widehat{A}}_{\textup{NN}}$ of radius $O(\sqrt{r_1NP/T})$, then  all iterates $\cm{A}^{(k)}$ will satisfy $\lVert \cm{A}^{(k)}- \cm{A} \rVert_{\textup{F}} \leq  \lVert \cm{A}^{(k)}- \cm{\widehat{A}}_{\textup{NN}} \rVert_{\textup{F}} + \lVert \cm{\widehat{A}}_{\textup{NN}} - \cm{A} \rVert_{\textup{F}}=O_p(\sqrt{r_1NP/T})$. Additionally, Theorem \ref{thm:errorbound} implies $\lVert \cm{A}^{(k)}- \cm{\widehat{A}}_{\textup{SHORR}} \rVert_{\textup{F}} \leq \lVert \cm{A}^{(k)}- \cm{A} \rVert_{\textup{F}} + \lVert \cm{\widehat{A}}_{\textup{SHORR}}- \cm{A} \rVert_{\textup{F}} =O_p(\sqrt{r_1 NP/T})$, where  $\cm{\widehat{A}}_{\textup{SHORR}}$ is the global solution.
A similar convex relaxation based initialization approach is used by \cite{Uematsu17} for a nonconvex optimization problem with jointly imposed sparsity and orthogonality constraints.  
Moreover, since Algorithm 2 and Proposition \ref{prop4} do not guarantee the convergence to a global solution, similarly to the random initialization method in Remark \ref{remark:init1}, in practice we can try many randomized initial values $\cm{A}^{(0)}=\cm{\widehat{A}}_{\textup{NN}}+(NP/T)^{1/2}\cm{T}$, where the entries of the perturbation $\cm{T}\in\mathbb{R}^{N\times N\times P}$ are drawn independently from $N(0,(N^2P)^{-1})$ such that $\|\cm{T}\|_{\mathrm{F}}=O_p(1)$, and then select the final solution as the one with the smallest value for the objective function. 
\end{remark}

\begin{remark}
	The above algorithms are presented under known multilinear ranks and a fixed tuning parameter $\lambda$. In practice, to save computational costs, we recommend a two-step procedure: first select the ranks by the method to be introduced in Section 5, and then fixing these rank, select the tuning parameter $\lambda$  by a fine grid search with information criterion such as the BIC or its high-dimensional extensions. Although the degrees of freedom in a sparse and orthogonal matrix are unclear, the total number of nonzero elements in $\cm{G}$, $\bm{U}_1$, $\bm{U}_2$ and $\bm{U}_3$ could be used as proxies.
\end{remark}

\section{Rank Selection}

The theoretical results we derived for MLR and SHORR estimators  hinge on correct multilinear ranks. This section introduces a procedure for consistent rank selection.

Suppose that $\cm{\widehat{A}}$ is a consistent initial estimator  of $\cm{A}$.  We propose the following ridge-type ratio estimator \citep{xia2015consistently}  to estimate the multilinear ranks,
\[
\widehat{r}_i=\argmin_{1\leq j\leq p_i-1} \frac{\sigma_{j+1}(\cm{\widehat{A}}_{(i)})+c}{\sigma_{j}(\cm{\widehat{A}}_{(i)})+c},
\]
for $1\leq i\leq 3$, where $p_1=p_2=N$, $p_3=P$, and $c$ is a parameter that needs to be well chosen; see the assumption  below. Here we allow $N$, $P$ and the multilinear ranks to diverge with $T$. For $i=1,2,3$, denote
\[
\varsigma_i=\frac{1}{\sigma_{r_i}(\cm{A}_{(i)})}\cdot \max_{1\leq j< r_i} \frac{\sigma_{j}(\cm{A}_{(i)})}{\sigma_{j+1}(\cm{A}_{(i)})}.
\]

\begin{assumption} \label{asmp:minimal}
The parameter $c>0$ is chosen such that (i) $\|\cm{\widehat{A}}-\cm{A}\|_\textup{F}/c=o_p(1)$ and (ii) $c\max_{1\leq i\leq 3}\varsigma_i=o(1)$.
\end{assumption}

\begin{remark}\label{remark:rank}
In Assumption \ref{asmp:minimal}, Condition (i) states that the estimation error is dominated by $c$, while Condition (ii) requires that $c$  grows much slower than  $\varsigma_i$s. Roughly speaking,  Condition (ii) may be violated if the smallest nonzero singular value of $\cm{A}_{(i)}$ is too small, or if there is a big drop from $\sigma_{j}(\cm{A}_{(i)})$ to $\sigma_{j+1}(\cm{A}_{(i)})$, for some $1\leq j< r_i$ and $1\leq i\leq 3$. In either case, it will be more difficult for the ridge-type ratio to select the rank correctly. Note that if all the  nonzero singular values are bounded above and away from zero, then Condition (ii) simply becomes $c=o(1)$.
\end{remark}

Similar to the minimal signal assumption for variable selection consistency of sparsity-inducing estimators, Assumption \ref{asmp:minimal} is essential to the rank selection consistency:

\begin{theorem}\label{thm:rank}
Under Assumption \ref{asmp:minimal} and the conditions of Theorem \ref{thm:errorbound}, $\mathbb{P}(\widehat{r}_1=r_1,\widehat{r}_2=r_2,\widehat{r}_3=r_3)\to1$ as $T\to\infty$. 
\end{theorem}

For the initial estimator, in this paper we use the nuclear norm (NN) estimator for low-rank VAR models defined as
\[\cm{\widehat{A}}_{\text{NN}}=\argmin\frac{1}{T}\sum_{t=1}^T\|\bm{y}_t-\cm{A}_{(1)}\bm{x}_t\|_2^2+\lambda\|\cm{A}_{(1)}\|_*.\]
Note that the estimation error rate derived in \cite{Negahban11} for VAR(1) models can be readily extended to VAR($P$) cases, which yields $\|\cm{\widehat{A}}_{\textup{NN}}-\cm{A}\|_\text{F}=O_p(\sqrt{r_1NP/T})$; see also Remark \ref{remark:init2}. Then, the rank selection consistency in Theorem \ref{thm:rank} would hold for a relatively large range of $c$. In practice, we recommend using $c=\sqrt{NP\log(T)/10T}$, which is shown to  perform satisfactorily in the first simulation experiment of Section 6. 

\section{Simulation experiments\label{sec:sim}}
\subsection{Rank selection consistency}
As the rank selection method proposed  in Section 5 will be used throughout all the following simulations and real data analysis in the next section, we first conduct an experiment to evaluate its consistency.

The data are generated from the proposed  model in \eqref{eq:MLRVAR} with dimensions $(N,P)=(10,5)$,  multilinear ranks $(r_1,r_2,r_3)=(3,3,3)$, and $\bm{\epsilon}_t \overset{i.i.d.}{\sim} N(\bm{0},\bm{I}_N)$. To examine how the singular values of $\cm{A}_{(i)}$s impact the rank selection performance, we let $\cm{G}$ be a diagonal cube with superdiagonal elements $(\cm{G}_{111},\cm{G}_{222},\cm{G}_{333})=(2,2,2)$ (case a), $(4,3,2)$ (case b), $(1,1,1)$ (case c), or $(2,1,0.5)$ (case d). As a result, the three nonzero singular values of every $\cm{A}_{(i)}$ are exactly $\cm{G}_{111},\cm{G}_{222}$ and $\cm{G}_{333}$. We generate the orthonormal factor matrices $\bm{U}_i$s as the first $r_i$ left singular vectors of Gaussian random  matrices while ensuring that the stationarity condition in Assumption \ref{asmp:stationary} holds. The parameter $c$ for the proposed  ridge-type ratio estimator is set to $\sqrt{NP\log(T)/10T}$. Figure \ref{fig:E4} presents the proportion of correct rank selection, i.e., the event $\{(\widehat{r}_1,\widehat{r}_2, \widehat{r}_3)=(r_1,r_2,r_3)\}$, across different sample sizes $T\in[50,400]$ based on 1000 replications for each setting. First, it can be seen that the proportion increases as $T$ increases and reaches almost one when $T=400$ for all cases. Second, as noted in Remark \ref{remark:rank}, the rank selection may be more difficult  if the smallest nonzero singular value $\sigma_{r_i}(\cm{A}_{(i)})$ is too small, or if there is a big gap between any two consecutive nonzero singular values.  Thus, the better performance of cases a and b may be due to their larger $\sigma_{r_i}(\cm{A}_{(i)})$ compared to the other two cases. Moreover, it can be seen that cases a and c outperform cases b and d, respectively, which may be explained by the equality of the singular values $\cm{G}_{111},\cm{G}_{222}$ and $\cm{G}_{333}$ in the former cases.

\subsection{Performance of MLR and SHORR estimators \label{subsec:sim2}}
We conduct two experiments to verify the theoretical properties of the proposed MLR and SHORR estimators. 

We first verify the asymptotic results in Section 3 for the proposed MLR estimator $\cm{\widehat{A}}_{\textup{MLR}}$ in comparison with the other two  low-dimensional estimators, $\cm{\widehat{A}}_{\textup{OLS}}$ and $\cm{\widehat{A}}_{\textup{RRR}}$. The data are generated from model \eqref{eq:MLRVAR} with $(N,P)=(10,5)$, $\bm{\epsilon}_t \overset{i.i.d.}{\sim} N(\bm{0},\bm{I}_N)$, $r_1=r_2=3$, and $r_3=2,3$ or 4.  We generate $\cm{G}$ by scaling a randomly generated tensor with independent standard normal entries such that $\min_{1\leq i\leq 3}\sigma_{r_i}(\cm{G}_{(i))})=1$, and generate $\bm{U}_i$s by the same method as in the previous experiment. There are 1000 replications for each setting. Throughout this and all following experiments, the multilinear ranks are selected by the method in Section 5. For each estimator, i.e., $\cm{\widehat{A}}=\cm{\widehat{A}}_{\textup{OLS}}, \cm{\widehat{A}}_{\textup{RRR}}$  or $\cm{\widehat{A}}_{\textup{MLR}}$, we calculate the average bias across all elements of $\cm{\widehat{A}}$ and  all replications. The square of this average bias is plotted against  $T\in[2000,4000]$ in the upper panels of Figure \ref{fig:E1}. We also calculate the empirical and asymptotic variances  for each element of  $\cm{\widehat{A}}$ according to Theorem \ref{thm:Asymptotic} and Corollary \ref{cor:comparison}. The averages of these empirical and asymptotic variances over all elements of $\cm{\widehat{A}}$ and  all replications, denoted by EVar and AVar, respectively, are plotted against $T$ in the lower panels of Figure \ref{fig:E1}. It can be seen that $\cm{\widehat{A}}_{\textup{MLR}}$ has much smaller squared bias,  EVar and AVar than $\cm{\widehat{A}}_{\textup{OLS}}$ and $\cm{\widehat{A}}_{\textup{RRR}}$. In addition, the EVar generally matches the corresponding AVar well, with their difference getting smaller as $T$ increases, although the EVar tends to overestimate the variances for all cases due to the large $(N,P)$ relative to the sample size. In sum, the asymptotic theory of the proposed MLR estimator in Section 3 is confirmed by this experiment.

The goal of the next experiment is to verify the non-asymptotic error bound of the proposed SHORR estimator. We consider two settings of the multilinear ranks, $(r_1,r_2,r_3)=(2,2,2)$ and $(3,3,3)$, and  the following four cases of $(N,P, s_1,s_2,s_3)$ for model \eqref{eq:MLRVAR}. For case a, we set $(N,P)=(10,5)$ and $(s_1,s_2,s_3)=(3,3,2)$. Then, cases b-d are defined by changing one of the settings in case a while keeping all others fixed. Specifically, we set $(s_1,s_2,s_3)=(2,2,2)$ in case b, $N=20$ in case c, and $P=10$ in case d. The core tensor $\cm{G}$ is generated in the same way as in the previous experiment, and the sparse orthonormal factor matrices $\bm{U}_i$s are generated randomly by the method  given in Section \ref{appendix: sim} of the Appendix. The regularization parameter $\lambda$ is selected by the BIC. By Theorem \ref{thm:errorbound}, fixing the multilinear ranks, it holds $\|\cm{\widehat{A}}_{\textup{SHORR}}-\cm{A}\|_{\textup{F}}^2=O_p(S\log(N^2P)/T)$, where $S=s_1s_2s_3$.  Thus, we denote $\gamma=S\log(N^2P)/T$ and set the sample size $T$ such that $\gamma=0.05$, 0.1, 0.15, 0.2, and 0.25. The mean squared error $\|\cm{\widehat{A}}_{\textup{SHORR}}-\cm{A}\|_{\textup{F}}^2$, averaged over 500 replications, is plotted against $\gamma$ in Figure \ref{fig:E2}. It is shown that the mean squared error generally increases linearly in $\gamma$, and the four lines in each plot almost coincide. These findings support the error bound  in Theorem \ref{thm:errorbound}.

\subsection{Comparison with existing estimation methods}

In the following experiment, we compare the performance of the proposed MLR and SHORR estimators with those of four existing ones for low-rank and/or sparse VAR models, including (i) Lasso \citep{Tibshirani96,Basu15}; (ii) nuclear norm \citep[NN]{Negahban11}; (iii) regression with a sparse SVD \citep[RSSVD]{ChenChan12}; and (iv) sparse and orthogonal factor regression \citep[SOFAR]{Uematsu17}.   

The data are generated from model \eqref{eq:MLRVAR} with $(N,P)=(10,5)$ (case a) or $(15,8)$ (case b). For both cases, we let $(r_1,r_2,r_3)=(3,3,3)$, $(s_1,s_2,s_3)=(3,3,2)$ and $\bm{\epsilon}_t \overset{i.i.d.}{\sim} N(\bm{0},\bm{I}_N)$.  For case a, $\cm{G}$ and $\bm{U}_i$s are generated by the same methods as in the previous subsection, and in case b, zeros rows are added below the $\bm{U}_i$s in case a. In both cases, entry-wisely $\|\cm{A}\|_0=500$. Hence, it is not sparse in case a, but is sparse in case b due to the zero rows of $\bm{U}_i$s.  Figure \ref{fig:E3} plots the estimation error $\|\cm{\widehat{A}}-\cm{A}\|_{\textup{F}}$ averaged over 500 replications  against $T\in[500,900]$ and $T\in[800,1200]$ for the smaller and larger $(N,P)$ cases, respectively. The error bars representing $\pm$ one standard deviation are also displayed for the proposed estimators, and suppressed for the others for clearer presentation.

Under both smaller and larger $(N,P)$, Figure \ref{fig:E3} shows that  both $\cm{\widehat{A}}_{\textup{MLR}}$ and $\cm{\widehat{A}}_{\textup{SHORR}}$ significantly outperform the other estimators which either   consider the low-rankness along only one direction or ignore it completely. Moreover, $\cm{\widehat{A}}_{\textup{SHORR}}$ consistently outperforms $\cm{\widehat{A}}_{\textup{MLR}}$ as the former exploits the sparsity of $\bm{U}_i$s in addition to the low-rankness along three dimensions. It is also interesting to note the different performances of $\cm{\widehat{A}}_{\textup{LASSO}}$ and $\cm{\widehat{A}}_{\textup{NN}}$ in Figure \ref{fig:E3}. Since  $\cm{\widehat{A}}_{\textup{LASSO}}$ only exploits entry-wise sparsity of $\cm{A}$, it has the worst performance when $\cm{A}$ is not sparse, as shown in the left panel.  In contrast, $\cm{\widehat{A}}_{\textup{NN}}$ only takes into account the low-rankness, so it performs best among the four existing estimators when $\cm{A}$ is not sparse, and yet becomes the worst when $\cm{A}$ is sparse as is the case for the right panel. This suggests that higher efficiency can be achieved by incorporating both the low-rankness and sparsity, which is the key advantage of the proposed $\cm{\widehat{A}}_{\textup{SHORR}}$.

\subsection{Comparison with factor models \label{subsec:sim_fac}}
The final experiment aims to compare the proposed model to the static and dynamic factor models, namely SFM and DFM, discussed in Section 2.3. Note that the SFM  cannot be directly used for forecasting since it does not impose an explicit model on the latent factors. However, for data generated by both the proposed model and the DFM, the SFM can be used to estimate the low-dimensional subspace where the conditional mean $\mathbb{E}(\bm{y}_t|\mathcal{F}_{t-1})$ lies; see Remark \ref{remark:U_space}.  

We consider four data generating processes for $\{\bm{y}_t\}$ with dimension $N=10$. Two of them are generated by DFMs with $\bm{e}_t\overset{i.i.d.}{\sim}N(\bm{0}, 0.5\bm{I}_{N})$: 
\begin{itemize}
	\item DFM-1:  The DFM with $r=1$, specified jointly by the SFM $\bm{y}_t=\bm{\Lambda}f_t+\bm{e}_t$ and the autoregressive latent factor $f_t = Bf_{t-1} + \xi_t$, where $\bm{\Lambda}\in\mathbb{R}^{10\times1}$ is a randomly generated vector with unit Euclidean norm, $B=0.5$, and $\xi_t\overset{i.i.d.}{\sim}N(0, 1)$.
	\item DFM-2: The DFM  with $r=3$, specified jointly by the SFM $\bm{y}_t=\bm{\Lambda}\bm{f}_t+\bm{e}_t$ and the VAR(1) process for the latent factors $\bm{f}_t = \bm{B}\bm{f}_{t-1} + \bm{\xi}_t$, where $\bm{\Lambda}\in\mathbb{R}^{10\times3}$ is a randomly generated orthonormal matrix, $\bm{B}=\textup{diag}(0.6,0.5,0.4)$, and $\bm{\xi}_t\overset{i.i.d.}{\sim}N(\bm{0}, \bm{I}_3)$.
\end{itemize}
The other two are generated by the proposed model with  $P=3$ and $\bm{\epsilon}_t \overset{i.i.d.}{\sim} N(\bm{0},\bm{I}_N)$:
\begin{itemize}
	\item MLR-1: The proposed multilinear low-rank VAR model in \eqref{eq:MLRVAR} with $(r_1,r_2,r_3)=(2,2,2)$. The core tensor $\cm{G}$ and factor matrices $\bm{U}_i$s are generated in the same way as the first experiment in Section 6.2.
	\item MLR-2: Same as MLR-1 except for $(r_1,r_2,r_3)=(3,3,3)$.
\end{itemize}

We first compare the performance of the proposed model and the SFM in terms of the estimation accuracy of the conditional mean subspace. The estimation of the SFM is conducted by the principal component method in \cite{BW16}. The subspace estimation error  can be measured by   $\|\bm{\widehat{\Lambda}}_1\bm{\widehat{\Lambda}}_1'-\bm{\Lambda}\bm{\Lambda}'\|_\text{F}^2$ for DFM-1 and DFM-2, and $\|\bm{\widehat{\Lambda}}_1\bm{\widehat{\Lambda}}_1'-\bm{U}_1\bm{U}_1'\|_\text{F}^2$ for MLR-1 and MLR-2, where  $\bm{\widehat{\Lambda}}_1=\bm{\widehat{\Lambda}}(\bm{\widehat{\Lambda}}'\bm{\widehat{\Lambda}})^{-1/2}$ is the normalized version of $\bm{\widehat{\Lambda}}$ for the fitted SFM. The results based on 1000 replications are displayed in Figure \ref{fig:E4a}.  It can be seen from the upper panels of the figure  that the conditional mean subspaces for DFM-1 and DFM-2 can be consistently estimated by fitting the corresponding SFMs. 
However, as discussed in Section \ref{subsec:fm}, for data generated by the DFM, fitting the proposed model will lead to  model misspecification. This may explain why the subspace estimation error for the MLR method is much larger and seems to persist for large $T$ in the upper panels of Figure \ref{fig:E4a}.  
On the other hand, when the data generating process is MLR-1 or MLR-2, the lower panels of  Figure \ref{fig:E4a} show that both methods can estimate the subspace consistently, although the MLR method is more efficient. This agrees with our observation that the proposed model admits an SFM representation.

We next compare the performance of the proposed model and the DFM through the  prediction error of the conditional mean $\mathbb{E}(\bm{y}_{T+1}|\mathcal{F}_{T})$. The DFM is estimated by a two-step approach, where we first obtain the estimated factors $\bm{\widehat{f}}_t$ by fitting an SFM, and then fit a (vector) autoregressive model to $\{\bm{\widehat{f}}_t\}$. Figure \ref{fig:E4b} displays the prediction error $\|\widehat{\mathbb{E}}(\bm{y}_{T+1}|\mathcal{F}_T)-\mathbb{E}(\bm{y}_{T+1}|\mathcal{F}_T)\|_2$ based on 1000 replications. Remarkably, as shown in the upper panels, even if the data are generated from the DFM-1 or DFM-2, the proposed model exhibits competitive forecasting performance despite the model misspecfication.  On the other hand, as shown in the lower panels, when the data are generated from MLR-1 or MLR-2, the  forecasting performance of the fitted DFM is rather poor. 
As discussed in Section 2.3, the proposed model can accommodate different low-dimensional patterns for the response $\bm{y}_t$ and predictors $\bm{y}_{t-j}$s, whereas the DFM requires the subspaces of $\bm{y}_t$ and $\bm{y}_{t-j}$s to be identical. 
When the DGP is the proposed model with distinct $\bm{U}_1$ and $\bm{U}_2$, the DFM will forecast $\bm{\widehat{U}}_1'\bm{y}_t$ based on $\bm{\widehat{U}}_1'\bm{y}_{t-j}$. However, the true conditional expectation of $\bm{y}_t$ is dependent on $\bm{U}_2'\bm{y}_{t-j}$, and the latter could be different from, or even orthogonal to, $\bm{\widehat{U}}_1'\bm{y}_{t-j}$. Consequently, when the response's low-dimensional subspace is applied to the predictors, the DFM may have no predictive power at all.
This explains why forecasting based on the DFM leads to considerable prediction errors. The robust forecasting performance of the proposed model reflects that its low-dimensional structure can be much more flexible than that of the DFM. 

\section{Real data analysis\label{sec:realdata}}

This section applies the proposed estimation methods to jointly model 40 quarterly macroeconomic sequences of the United States from 1959 to 2007, with 194 observed values for each variable \citep{Koop13}. All series are seasonally adjusted except for financial variables, transformed to stationarity, and standardized to zero mean and unit variance. These variables capture many aspects of the economy, and can be classified into eight categories: (i) GDP and its decomposition, (ii) National Association of Purchasing Managers (NAPM) indices, (iii) industrial production, (iv) housing, (v) money, credit and interest rate, (vi) employment, (vii) prices and wages, and (viii) others. The VAR model has been widely applied to fit these series in empirical econometric studies for structural analysis and forecasting; see \cite{Stock09} and \cite{Koop13}. Table \ref{tbl:macro} gives more details about these macroeconomic variables.

We first apply the SHORR estimation to the entire data set, with the lag order fixed at $P=4$ for the fitted VAR model as suggested by \citet{Koop13}.
Since the number of variables $N=40$ is much larger than the lag order $P=4$,
we do not perform variable selection for the factor matrix related to lags; that is, we replace $\|\bm{U}_3\otimes\bm{U}_2\otimes\bm{U}_1\|_1$ with $\|\bm{U}_2\otimes\bm{U}_1\|_1$ in the penalty term. The multilinear ranks are selected by the ridge-type ratio estimator, which results in $(r_1, r_2, r_3)=(4,3,2)$, and the tuning parameter $\lambda$ is selected by BIC.

The $\ell_1$ penalty yields sparse estimated  factor matrices $\widehat{\bm{U}}_1$ and $\widehat{\bm{U}}_2$, and the estimated coefficients are presented in Figure \ref{fg:estimates}.
The factor loading provides insights into the dynamic relationship among the 40 macroeconomic variables. The four response factors, denoted by $R_i$ for $1\leq i\leq 4$, contain nearly all of the variables and encapsulate different aspects of the economy: $R_1$ is mostly related to investments, imports, industrial production and employments; $R_2$ includes personal consumption, housing starts, and labor productivity; $R_3$ includes manufacturing, housing starts, and treasury bill yield rates; and $R_4$ includes NAPM indices, housing starts, and price index. Each response factor covers multiple categories of macroeconomic indices, and no clear group structure can be observed. However, it is noteworthy that only twelve variables are selected by the three predictor factors, and the sparse formulations of the predictor factors mainly consist of variables from the first four categories, including real GDP, private investment, NAPM indices, manufacturing and housing starts. The above result leads to an interesting interpretation: the activeness of production and investment serves as the driving force of the whole economy and usually precedes changes in other economic aspects such as the price indices, financial indices, and labor markets.

We next evaluate the forecasting performance of $\cm{\widehat{A}}_{\textup{MLR}}$ and  $\cm{\widehat{A}}_{\textup{SHORR}}$ in comparison with the competing estimators considered in Section \ref{sec:sim}.
The following rolling forecasting procedure is adopted: first, use the historical data with the end point rolling from Q4-2000 to Q3-2007 to fit the models; and then, conduct one-step-ahead forecasts based on the fitted models.
The selected ranks and tuning parameters for $\cm{\widehat{A}}_{\textup{SHORR}}$ are preserved from the analysis of the entire data set, i.e. $(r_1, r_2, r_3)=(4,3,2)$, and the selected ranks for MLR and RRR estimation are also fixed accordingly.

The $\ell_2$ and $\ell_\infty$ norms of the forecast errors for various methods are displayed in Table \ref{table:forecast}.  It can be seen that  the proposed MLR and SHORR estimators have much smaller forecast errors than competing ones, including the DFM with $r=4$ and the regularized and unregularized estimation methods for the VAR model.
This can be explained by the capability of the proposed estimators to substantially reduce the dimensionality along three directions simultaneously. The SHORR estimator performs best among all estimators as it enforces sparsity of the factor matrices and hence prevents overfitting most effectively.

\section{Conclusion and discussion}

For a large VAR($P$) model, its reduced-rank structure can be defined in three different ways. The novelty of the proposed approach lies in its ability to jointly enforce three different reduced-rank structures. This is made possible by rearranging the transition matrices of the VAR model into a tensor such that the Tucker decomposition can be conducted. As a result, the parameter space is restricted effectively along three directions, and the capability of the classical VAR model for modeling large-scale time series is substantially expanded. 

Moreover, for the high-dimensional setup, this paper further proposes a sparsity-inducing estimator to improve the model interpretability and estimation efficiency. An ADMM algorithm is developed to tackle the computational challenges due to the all-orthogonal constraints on $\cm{G}$  as well as the jointly imposed $\ell_1$-regularization and orthogonality constraints on $\bm{U}_i$s. It is worth noting that this paper has a different focus than most work on tensor regression: here we employ the tensor technique as a novel approach to the dimension reduction problem in classical VAR time series modeling.

This paper may be extended in three possible directions. Firstly, the proposed estimators do not take into account the possible correlation structure among components of $\bm{\epsilon}_t$, which will reduce the estimation efficiency.
Let $\widehat{\bm{\Sigma}}_{\bm{\epsilon}}$ be an estimator of ${\bm{\Sigma}}_{\bm{\epsilon}}$. As in \cite{Davis_Zang_Zheng2016}, we may alternatively consider the generalized least squares loss
$\sum_{t=1}^T(\bm{y}_t-\cm{A}_{(1)}\bm{x}_t)\widehat{\bm{\Sigma}}_{\bm{\epsilon}}^{-1}(\bm{y}_t-\cm{A}_{(1)}\bm{x}_t)$ rather than $\sum_{t=1}^T\|\bm{y}_t-\cm{A}_{(1)}\bm{x}_t\|_2^2$.  However, the difficulty would be to find a good estimator $\widehat{\bm{\Sigma}}_{\bm{\epsilon}}$.
Secondly, the tensor technique potentially can be applied to many variants of the VAR model, e.g., those with a nonlinear dynamic structure such as the threshold VAR model \citep{Tsay1998} and the varying coefficient VAR model \citep{Lutkepohl2005}. For instance, consider the time-varying coefficient VAR model with lag one, $\bm{y}_t=\bm{A}_t\bm{y}_{t-1}+\bm{\epsilon}_t$. Similarly, the coefficient matrices can be rearranged into a tensor $\cm{A}$ with $\cm{A}_{(1)}=(\bm{A}_1,...,\bm{A}_T)$. If $\cm{A}$ has multilinear low ranks $(r_1,r_2,r_3)$, then the number of parameters will be $r_1r_2r_3+(N-r_1)r_1+(N-r_2)r_2+(T-r_3)r_3\lesssim NT$. Moreover, a fourth-order tensor can be used to handle the case of lag order $P>1$.
Lastly, the proposed model can be generalized to a tensor autoregressive model for matrix-valued or tensor-valued time series; see \cite{wang2016factor} for a related work.

\section*{Acknowledgements}

We are grateful to the joint editor, the associate editor and three anonymous referees for their valuable comments which led to substantial improvement of this paper. 
Lian and Li are co-corresponding authors and contributed to the paper equally.
This research was partially supported by GRF grants 11300519 and 17305319 from the Hong Kong Research Grant Council and a Key Program grant 72033002 from National Natural Science Foundation of China.

\bibliography{mybib,HDVAR}

\newpage

\begin{figure}[ht]
	\centering
	\includegraphics[width=0.9\textwidth]{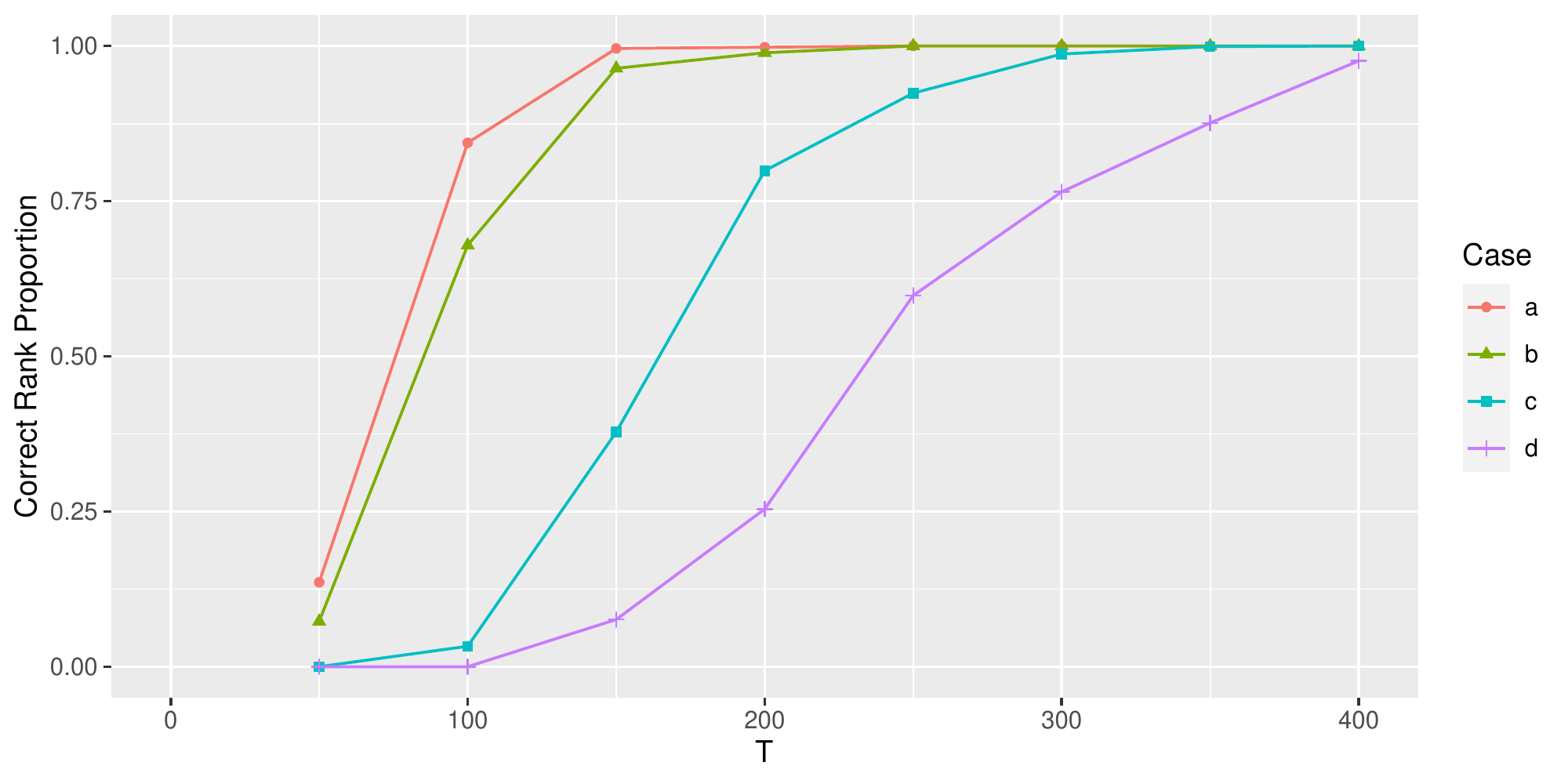}
	\caption{Proportion of correct rank selection when the three nonzero singular values of each $\cm{A}_{(i)}$ are $(2,2,2)$ (case a), $(4,3,2)$ (case b), $(1,1,1)$ (case c), or $(2,1,0.5)$ (case d).}
	\label{fig:E4}
\end{figure}

\begin{figure}[!htbp]
	\centering
	\includegraphics[width=0.9\textwidth]{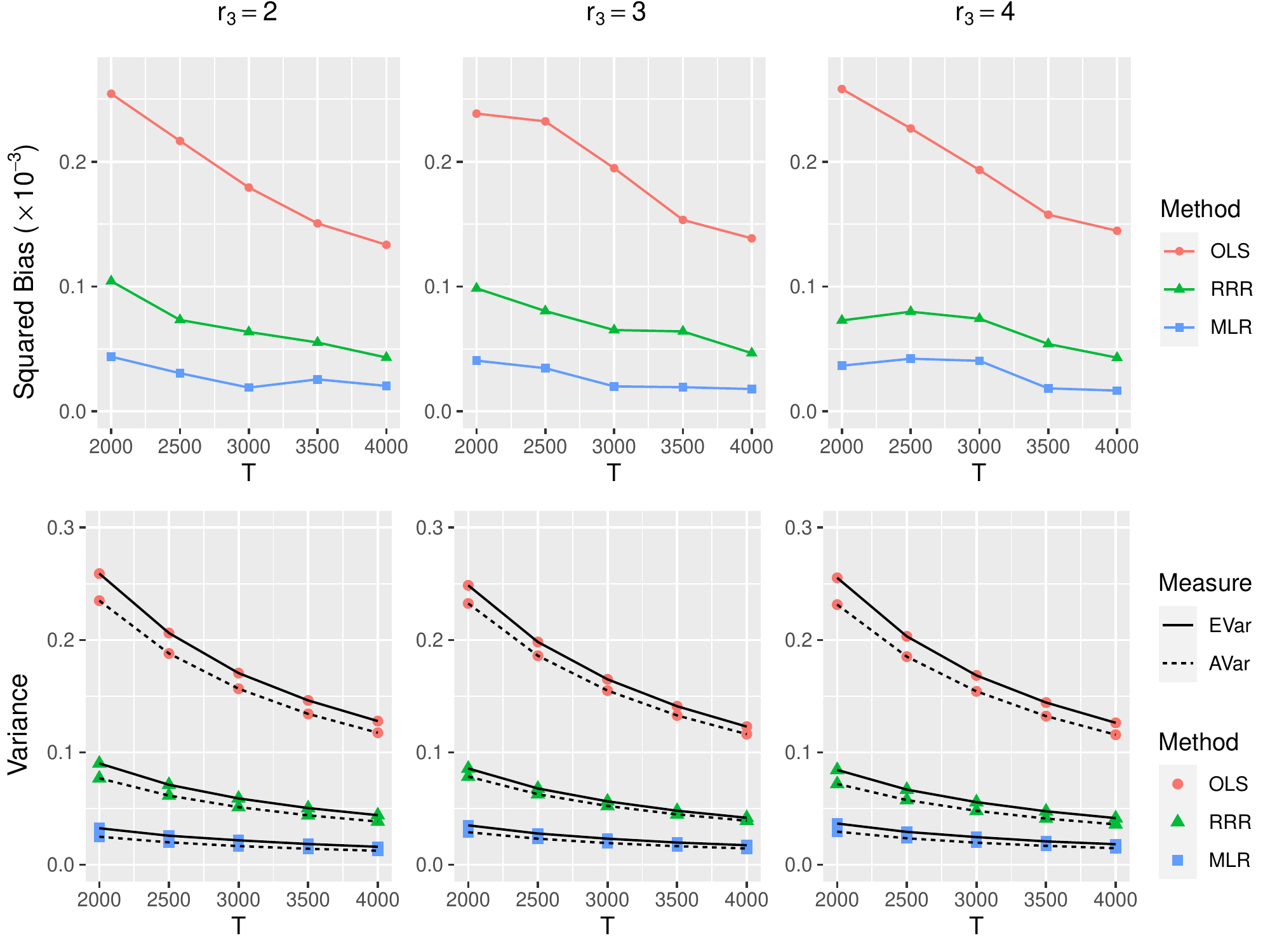}
	\caption{Squared bias, empirical  variance (EVar) and asymptotic variance (AVar) for $\cm{\widehat{A}}_{\textup{OLS}}$  $\cm{\widehat{A}}_{\textup{RRR}}$, and $\cm{\widehat{A}}_{\textup{MLR}}$ under various multilinear ranks.}
	\label{fig:E1}
\end{figure}

\clearpage
\begin{figure}[ht]
	\centering
	\includegraphics[width=\textwidth]{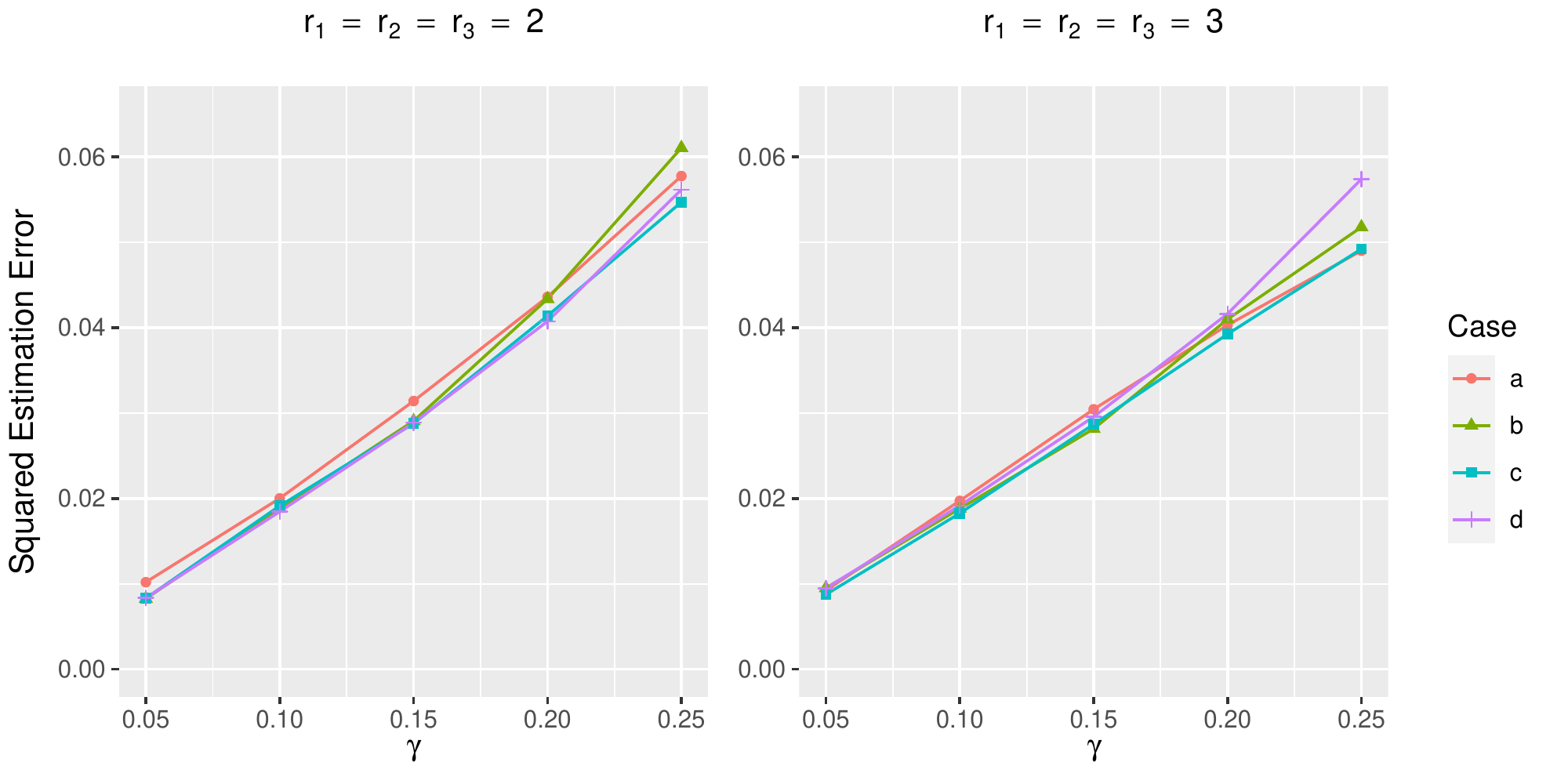}
	\caption{Plots of the squared estimation error $\|\cm{\widehat{A}}_{\textup{SHORR}}-\cm{A}\|_{\textup{F}}^2$ against  $\gamma=S\log(N^2P)/T$ for four cases of $(N,P, s_1,s_2,s_3)$ under two settings of  multilinear ranks.}
	\label{fig:E2}
\end{figure}

\begin{figure}[!htbp]
	\centering
	\includegraphics[width=\textwidth]{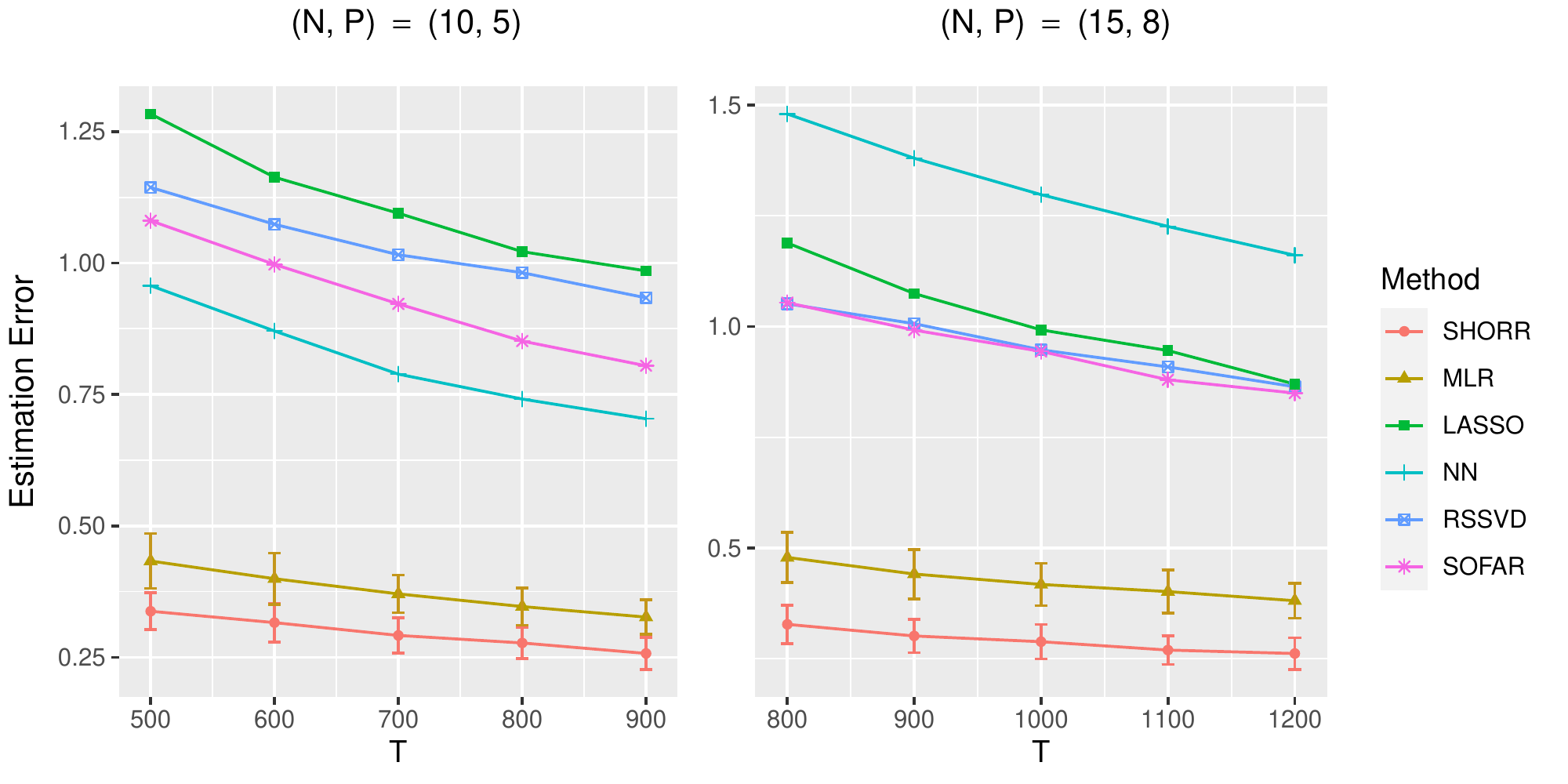}
	\caption{Plots of the estimation error  $\|\cm{\widehat{A}}-\cm{A}\|_{\textup{F}}$ against $T$ for six estimation methods under two settings of $(N,P)$.}
	\label{fig:E3}
\end{figure}

\begin{figure}[!htbp]
	\centering
	\includegraphics[width=0.8\textwidth]{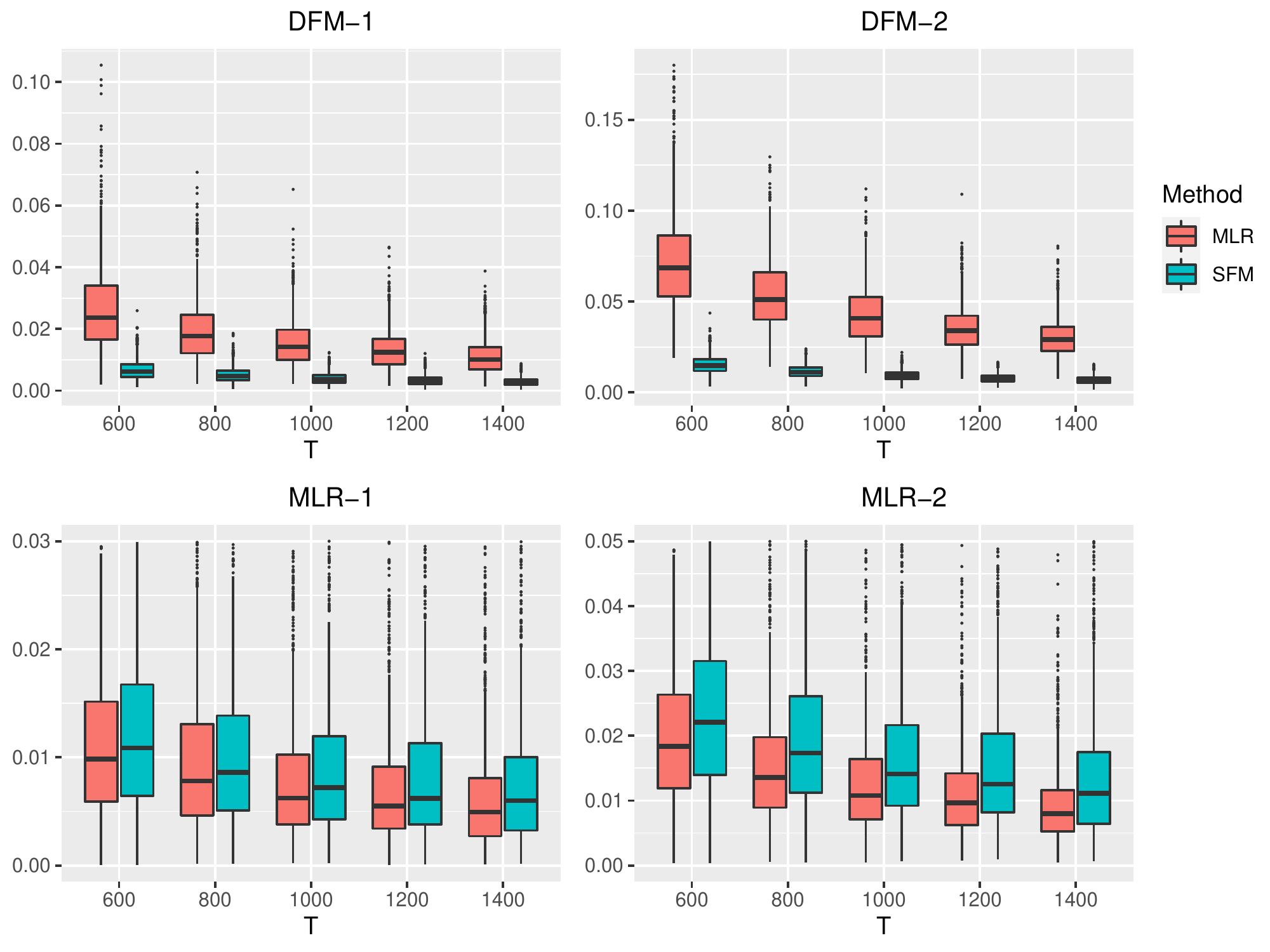}
	\vspace{-5mm}
	\caption{Subspace estimation error for four data generating processes based on two methods: fitting the proposed model by the MLR method, or fitting the static factor model (SFM).}
	\label{fig:E4a}
\end{figure}

\begin{figure}[!htbp]
	\centering
	\includegraphics[width=0.8\textwidth]{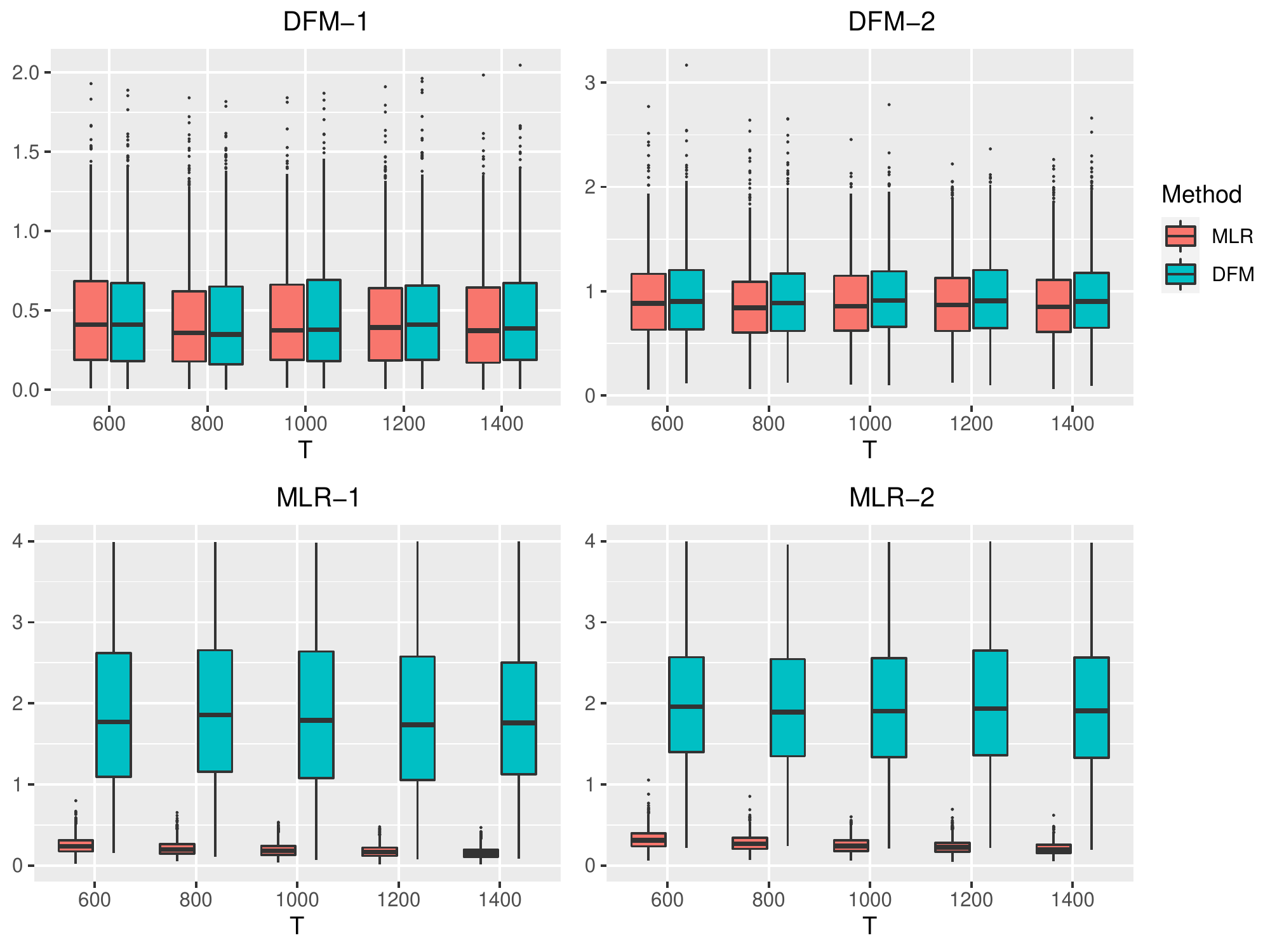}
	\vspace{-5mm}
	\caption{Prediction error for four data generating processes based on two methods: fitting the proposed model by the MLR method, or fitting the dynamic factor model (DFM).}
	\label{fig:E4b}
\end{figure}

\newpage
\begin{figure}[t]
	\centering\vspace{-10mm}
	\includegraphics[width=0.75\textwidth]{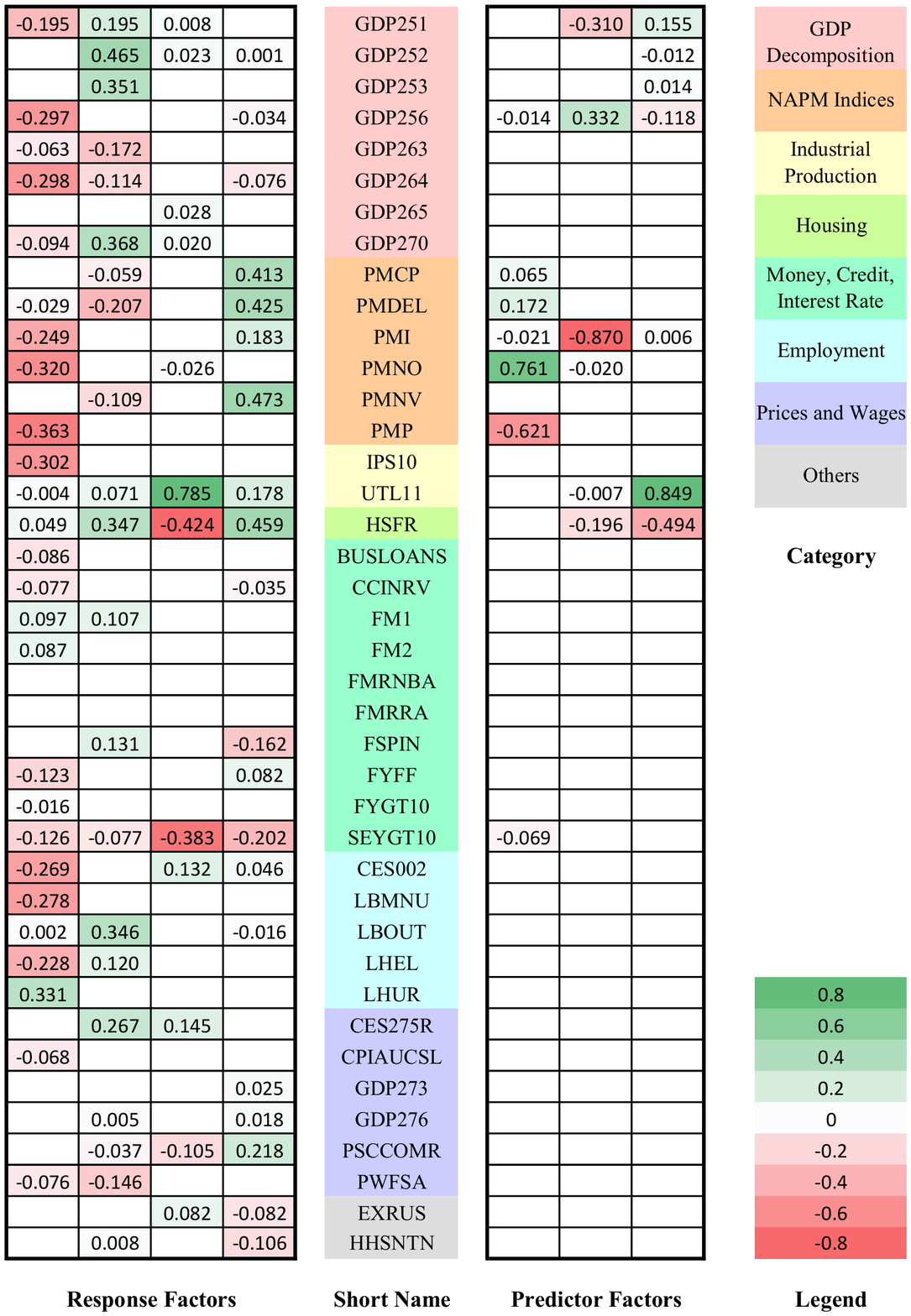}
	\vspace{-5mm}
	\caption{Estimated coefficients in the response and predictor factor loading matrices.}
	\label{fg:estimates}
\end{figure}

\begin{table}[h]
	\begin{center}\vspace{-5mm}
		\caption{Forecasting error for forty quarterly macroeconomic sequences of the United States from 1959 to 2007. The best cases among (un)regularized methods are marked in bold. \label{table:forecast}}	
		\vspace{-5mm}	
		\begin{tabular}{cccccccccccc}
			\hline\hline
			& \multicolumn{4}{c}{Unregularized methods}&&\multicolumn{6}{c}{Regularized methods}\\
			\cline{2-5}\cline{7-12}
			Criterion & OLS & RRR & DFM & MLR && SHORR & LASSO & NN & RSSVD  & SOFAR\\
			\hline
			$\ell_2$ norm  & 20.16 & 13.31 & 6.36 & \bf{5.81} && $\bf{5.35}$ & 6.72 & 8.16 & 6.33 & 6.28\\
			$\ell_\infty$ norm & 8.32 & 4.55 & 2.85 & \bf{2.56} &&$\bf{2.44}$ & 3.06 & 3.36 & 3.02 & 3.02
			\\\hline
		\end{tabular}
	\end{center}
	\vspace{-10mm}
\end{table}

\begin{landscape}
	\begin{table}[]
		\small
		\centering
		\caption{Forty quarterly macroeconomic variables belonging to 8 categories. Category code (C) represents: 1 = GDP and its decomposition, 2 = national association of purchasing managers (NAPM) indices, 3 = industrial production, 4 = housing, 5 = money, credit, interest rates, 6 = employment, 7 = prices and wages, 8 = others. Variables are seasonally adjusted except for those in category 5. All variables are transformed to stationarity with the following transformation codes (T): 1 = no transformation, 2 = first difference, 3 = second difference, 4 = log, 5 = first difference of logged variables, 6 = second difference of logged variables.}
		\label{tbl:macro}
		\begin{tabular}{@{}llllllll@{}}
			\toprule
			Short name&C&T&Description&Short name&C&T&Description\\
			\midrule
			GDP251&1&5& Real GDP, quantity index (2000=100)            &FM2&5&6& Money stock: M2 (bil\$)                        \\[-2ex]
			GDP252&1&5& Real personal cons exp, quantity index       &FMRNBA&5&3& Depository inst reserves: nonborrowed (mil\$)  \\[-2ex]
			GDP253&1&5& Real personal cons exp: durable goods&FMRRA&5&6& Depository inst reserves: total (mil\$)        \\[-2ex]
			GDP256&1&5& Real gross private domestic investment&FSPIN&5&5& S\&P's common stock price index: industrials   \\[-2ex]
			GDP263&1&5& Real exports&FYFF&5&2& Interest rate: federal funds (\% per annum)    \\[-2ex]
			GDP264&1&5& Real imports&FYGT10&5&2& Interest rate: US treasury const. mat., 10-yr  \\[-2ex]
			GDP265&1&5& Real govt cons expenditures \&  gross investment&SEYGT10&5&1& Spread btwn 10-yr and 3-mth T-bill rates\\[-2ex]
			GDP270&1&5& Real final sales to domestic purchasers&CES002&6&5& Employees, nonfarm: total private              \\[-2ex]
			PMCP&2&1& NAPM commodity price index (\%)&LBMNU&6&5& Hrs of all persons: nonfarm business sector\\[-2ex]
			PMDEL&2&1& NAPM vendor deliveries index (\%)&LBOUT&6&5& Output per hr: all persons, business sec\\[-2ex]
			PMI&2&1& Purchasing managers' index&LHEL&6&2& Index of help-wanted ads in newspapers\\[-2ex]
			PMNO&2&1& NAPM new orders index (\%)&LHUR&6&2& Unemp. rate: All workers, 16 and over (\%)     \\[-2ex]
			PMNV&2&1& NAPM inventories index (\%)&CES275R&7&5& Real avg hrly earnings, non-farm prod. workers \\[-2ex]
			PMP&2&1& NAPM production index (\%)&CPIAUCSL&7&6& CPI all items                                  \\[-2ex]
			IPS10&3&5& Industrial production index: total             &GDP273&7&6& Personal consumption exp.: price index         \\[-2ex]
			UTL11&3&1& Capacity utilization: manufacturing (SIC)      &GDP276&7&6& Housing price index\\[-2ex]
			HSFR&4&4& Housing starts: Total (thousands)              &PSCCOMR&7&5& Real spot market price index: all commodities  \\[-2ex]
			BUSLOANS&5&6& Comm. and industrial loans at all comm. Banks&PWFSA&7&6& Producer price index: finished goods           \\[-2ex]
			CCINRV&5&6& Consumer credit outstanding: nonrevolving&EXRUS&8&5& US effective exchange rate: index number       \\[-2ex]
			FM1&5&6& Money stock: M1 (bil\$)                        &HHSNTN&8&2& Univ of Mich index of consumer expectations\\[-1ex]
			\bottomrule
		\end{tabular}
	\end{table}
\end{landscape}

\appendix

\section{Proofs of Theorems 1-3\label{append:A1}}
\begin{proof}[Proof of Theorem 1]
	
	The proof generally follows from Proposition 4.1 in \citet{Shapiro86} for overparameterized models. The VAR($P$) model can be written as the linear regression problem
	\begin{equation}\label{eq:lr}
	\underbrace{\begin{bmatrix}
		\bm{y}_{1}'\\\bm{y}_{2}'\\ \vdots\\ \bm{y}_T'
		\end{bmatrix}}_{\bm{Y}}=
	\underbrace{\begin{bmatrix}
		\bm{y}_{0}'&\bm{y}_{-1}'&\dots&\bm{y}_{-P+1}'\\
		\bm{y}_{1}'&\bm{y}_{0}'&\dots&\bm{y}_{-P+2}'\\
		\vdots & \vdots & \ddots & \vdots\\
		\bm{y}_{T-1}'&\bm{y}_{T-2}'&\dots&\bm{y}_{T-P}'\\
		\end{bmatrix}}_{\bm{X}}\underbrace{\begin{bmatrix}
		\bm{A}_1'\\\bm{A}_2'\\\vdots\\\bm{A}_P'
		\end{bmatrix}}_{\cm{A}_{(1)}'}+\underbrace{\begin{bmatrix}
		\bm{\epsilon}_{P+1}'\\\bm{\epsilon}_{P+2}'\\\vdots\\\bm{\epsilon}_{T}'
		\end{bmatrix}}_{\bm{E}}
	\end{equation}
	Let $\bm{\phi}$ denote the component parameters in the Tucker decomposition forms, let $\bm{h}(\bm{\phi})$ denote the true parameter $\text{vec}(\cm{A}_{(1)})=\text{vec}(\bm{U}_1\cm{G}_{(1)}(\bm{U}_3\otimes\bm{U}_2)')$, and let $\bm{\widehat{h}}_{\textup{OLS}}$ denote the vectorized OLS estimates $\text{vec}(\bm{\widehat{A}}_{\textup{OLS}})$ without constraint.
	With Assumption 1, according to classical asymptotic theory for stationary VAR model \citep{Tsay13}, as $T\to\infty$,
	\begin{equation}
	\begin{split}
	(i).~~ & \bm{X}'\bm{X}/T\overset{p}{\to}\bm{\Gamma}^*\equiv
	\begin{bmatrix}
	\bm{\Gamma}_0 & \bm{\Gamma}_1 & \dots & \bm{\Gamma}_{P-1}\\
	\bm{\Gamma}_1' & \bm{\Gamma}_0 & \dots & \bm{\Gamma}_{P-2}\\
	\vdots & \vdots & \ddots & \vdots\\
	\bm{\Gamma}_{P-1}' & \bm{\Gamma}_{P-2}' & \dots & \bm{\Gamma}_0\\
	\end{bmatrix};\\
	(ii).~~ & T^{-1/2}\text{vec}(\bm{X}'\bm{E})\overset{d}{\to}N(\bm{0},\bm{\Sigma}_{\bm{\epsilon}}\otimes\bm{\Gamma}^*);\\
	(iii).~~ & \bm{\widehat{h}}_{\textup{OLS}}\overset{p}{\to}\bm{h};\\
	(iv).~~ & \sqrt{T}(\bm{\widehat{h}}_{\textup{OLS}}-\bm{h})\overset{d}{\to}N\left(\bm{0},\bm{\Sigma}_{\bm{\epsilon}}\otimes(\bm{\Gamma}^*)^{-1}\right).
	\end{split}
	\end{equation}
	
	Consider the discrepancy function for any $\bm{h}(\bm{\phi})$,
	\begin{equation}
	F(\bm{\widehat{h}}_{\textup{OLS}},\bm{h})=\|\text{vec}(\bm{Y})-(\bm{I}_N\otimes \bm{X})\bm{h}\|^2_2-\|\text{vec}(\bm{Y})-(\bm{I}_N\otimes \bm{X})\bm{\widehat{h}}_{\textup{OLS}}\|^2_2.
	\end{equation}
	Obviously, $F(\bm{\widehat{h}}_{\textup{OLS}},\bm{h})$ is a nonnegative and twice continuously differentiable function, and equals to zero if and only if $\bm{\widehat{h}}_{\textup{OLS}}=\bm{h}$.
	
	In order to calculate the Jacobian matrix $\bm{H}$, we define the tensor matricization transformation operator $\bm{T}_{ij}(N,N,P)$ which is an $N^2P\times N^2P$ matrix and satisfies that $\text{vec}(\cm{A}_{(j)})=\bm{T}_{ij}(N,N,P)\text{vec}(\cm{A}_{(i)})$ for any tensor $\cm{A}\in\mathbb{R}^{N\times N\times P}$.
	In fact, $\bm{T}_{ij}(N,N,P)$ is a full-rank matrix indicating the corresponding position in $\text{vec}(\cm{A}_{(i)})$ of $\cm{A}$'s each entry in $\text{vec}(\cm{A}_{(j)})$, and can be regarded as the natural extension of the permutation matrix for matrix transpose. Also note that $\bm{T}_{ij}(N,N,P)$ only depends on the value of $N$ and $P$, and since we consider fixed $N$ and $P$ in this part, we simplify it to $\bm{T}_{ij}$.
	
	Therefore,
	\begin{equation}\begin{split}
	\text{vec}(\cm{A}_{(1)})=&\text{vec}(\bm{U}_1\cm{G}_{(1)}(\bm{U}_3\otimes\bm{U}_2)')=\bm{T}_{21}\text{vec}(\bm{U}_2\cm{G}_{(2)}(\bm{U}_1\otimes\bm{U}_3)')\\
	=&\bm{T}_{31}\text{vec}(\bm{U}_3\cm{G}_{(3)}(\bm{U}_1\otimes\bm{U}_2)'),
	\end{split}\end{equation}
	
	and the Jacobian matrix of $\bm{h}$ is
	\begin{equation}\begin{split}
	\bm{H}=\frac{\partial \bm{h}}{\partial\bm{\phi}}=&\Big((\bm{U}_3\otimes\bm{U}_2\otimes\bm{U}_1),[(\bm{U}_3\otimes\bm{U}_2)\cm{G}_{(1)}']\otimes\bm{I}_N,\\
	&\bm{T}_{21}\{[(\bm{U}_1\otimes\bm{U}_3)\cm{G}_{(2)}']\otimes\bm{I}_N\},\bm{T}_{31}\{[(\bm{U}_1\otimes\bm{U}_2)\cm{G}_{(3)}']\otimes \bm{I}_P\}\Big).
	\end{split}\end{equation}
	
	Then, by Proposition 4.1 in \citet{Shapiro86}, we know that the minimizer of $F(\bm{\widehat{h}}_{\textup{OLS}},\cdot)$, namely the MLR estimator, has the asymptotic normality,
	\begin{equation}
	\sqrt{T}(\bm{h}(\bm{\widehat{\phi}}_{\textup{MLR}})-\bm{h})\overset{d}{\to}N(\bm{0},\bm{\Sigma}_{\textup{MLR}})
	\end{equation}
	and $\bm{\Sigma}_{\textup{MLR}}=\bm{P}\bm{\Gamma}\bm{P}'$, where $\bm{P}=\bm{H}(\bm{H}'\bm{J}\bm{H})^\dagger\bm{H}'\bm{J}$ is the projection matrix, $\bm{J}$ is the Fisher information matrix of $\bm{h}$ as $T$ goes to infinity,
	$\bm{H}$ is the Jacobian matrix of $\bm{h}$ with respect to the overparameterized model parameters $\bm{\phi}$, $\bm{\Gamma}=\bm{\Sigma}_{\bm{\epsilon}}\otimes(\bm{\Gamma}^*)^{-1}$ is the asymptotic covariance matrix for $\bm{\widehat{h}}_{\textup{OLS}}$ and $\dagger$ denotes the Moore-Penrose inverse.
	Since $\bm{\Gamma}=\bm{J}^{-1}$ in the VAR($P$) model, we have $\bm{\Sigma}_{\textup{MLR}}=\bm{H}(\bm{H}'\bm{J}\bm{H})^\dagger\bm{H}'$.
\end{proof}

\begin{proof}[Proof of Theorem \ref{thm:errorbound}]
	
	The proof of Theorem \ref{thm:errorbound} consists of two parts.
	\begin{itemize}
		\item The first part shows the estimation error bounds given the deterministic realization of the time series process, assuming that the deviation bound condition and restricted eigenvalue condition hold.
		\item The second step is the stochastic analysis in which we show that these two regulatory conditions are satisfied with high probability converging to 1.
	\end{itemize}
	
	Based on the linear regression form \eqref{eq:lr}, we can rewrite $\bm{Y}=\bm{X}(\bm{U}_3\otimes\bm{U}_2)\cm{G}_{(1)}'\bm{U}_1'+\bm{E}$ as
	\begin{equation}
	\underbrace{\text{vec}(\bm{Y})}_{\bm{y}}=\underbrace{(\bm{I}_N\otimes\bm{X})}_{\bm{Z}}\underbrace{(\bm{U}_3\otimes \bm{U}_2\otimes\bm{U}_1)}_{\bm{U}}\underbrace{\text{vec}(\cm{G}_{(1)}')}_{\bm{g}}+\underbrace{\text{vec}(\bm{E})}_{\bm{e}}.
	\end{equation}
	
	Denote $\bm{\widehat{\Delta}}=\bm{\widehat{U}}\bm{\widehat{g}}-\bm{Ug}$, $\bm{\widehat{\Delta}}_{\bm{u}}=\bm{\widehat{U}}-\bm{U}$ and $\bm{\widehat{\Delta}}_{\bm{g}}=\bm{\widehat{g}}-\bm{g}$. By the optimality of the SHORR estimator,
	\begin{equation}
	\label{eq:shorr}
	\begin{split}
	&T^{-1}\|\bm{y}-\bm{Z}\bm{\widehat{U}}\bm{\widehat{g}}\|_2^2+\lambda\|\bm{\widehat{U}}\|_1\leq T^{-1}\|\bm{y}-\bm{ZUg}\|_2^2+\lambda\|\bm{U}\|_1\\
	\Rightarrow&T^{-1}\|\bm{Z}\bm{\widehat{\Delta}}\|_2^2+\lambda\|\bm{\widehat{U}}\|_1\leq 2\langle T^{-1}\bm{Z}'\bm{e},\bm{\widehat{\Delta}}\rangle+\lambda\|\bm{U}\|_1.
	\end{split}
	\end{equation}
	
	Note that $\bm{\widehat{\Delta}}=\bm{\widehat{U}}\bm{\widehat{g}}-\bm{Ug}=(\bm{\widehat{U}}-\bm{U})\bm{\widehat{g}}+\bm{U}(\bm{\widehat{g}}-\bm{g})=\bm{\widehat{\Delta}}_{\bm{u}}{\bm{\widehat{g}}}+\bm{U}\bm{\widehat{\Delta}}_{\bm{g}}$,
	so we can decompose $\langle T^{-1}\bm{Z}'\bm{e},\bm{\widehat{\Delta}}\rangle$ into two parts,
	\begin{equation}
	\langle T^{-1}\bm{Z}'\bm{e},\bm{\widehat{\Delta}}\rangle=\langle T^{-1}\bm{Z}'\bm{e},\bm{\widehat{\Delta}}_{\bm{u}}\bm{\widehat{g}}\rangle+\langle T^{-1}\bm{Z}'\bm{e},\bm{U}\bm{\widehat{\Delta}}_{\bm{g}}\rangle.
	\end{equation}
	and bound these two parts separately.
	We denote the event that these two inner products are bounded by $\lambda\|\bm{\widehat{\Delta}}_{\bm{u}}\|_1$ and $\lambda\|\bm{\widehat{\Delta}}_g\|_1$ as $\mathcal{I}_1$,
	\begin{equation}
	\mathcal{I}_1=\{\langle T^{-1}\bm{Z}'\bm{e},\bm{\widehat{\Delta}}_{\bm{u}}\bm{\widehat{g}}\rangle\leq\lambda\|\bm{\widehat{\Delta}}_{\bm{u}}\|_1/4\}\cap \{\langle T^{-1}\bm{Z}'\bm{e},\bm{U}\bm{\widehat{\Delta}}_{\bm{g}}\rangle\leq\lambda\|\bm{\widehat{\Delta}}_{\bm{g}}\|_1/4\}.
	\end{equation}
	
	Denote by $\bb{S}_{\bm{U}}$ the nonzero index set of $\text{vec}(\bm{U})$, and by $\bb{S}_{\bm{U}}^\complement$ is the complement of $\bb{S}_{\bm{U}}$. By the sparsity of each $\bm{U}_i$ in Assumption 4, $\text{card}(\bb{S}_{\bm{U}})=\|\bm{U}_3\otimes \bm{U}_2\otimes\bm{U}_1\|_0\leq \prod_{i=1}^3s_ir_i$.
	In the following proof, we use the abused notation. For any matrix $\bm{M}\in\mathbb{R}^{N^2P\times r_1r_2r_3}$ and any vector norm $\|\cdot\|_*$, we denote $\|(\bm{M})_{\bb{S}_{\bm{U}}}\|_{*}:=\|(\text{vec}(\bm{M}))_{\bb{S}_{\bm{U}}}\|_{*}$ and $\|(\bm{M})_{\bb{S}_{\bm{U}}^\complement}\|_{*}:=\|(\text{vec}(\bm{M}))_{\bb{S}_{\bm{U}}^\complement}\|_{*}$.
	
	On the event $\mathcal{I}_1$, if we multiply 2 to both sides of \eqref{eq:shorr} we can have
	\begin{equation}
	\label{eq:basicinequality}
	\begin{split}
	&2T^{-1}\|\bm{Z}\bm{\widehat{\Delta}}\|^2_2+2\lambda\|\bm{\widehat{U}}\|_1\leq\lambda\|\bm{\widehat{\Delta}}_{\bm{g}}\|_1+\lambda\|\bm{\widehat{\Delta}}_{\bm{u}}\|_1+2\lambda\|\bm{U}\|_1
	\end{split}
	\end{equation}
	On the left-hand side, by triangle inequality,
	\begin{equation}
	\|\bm{\widehat{U}}\|_1=\|\bm{\widehat{U}}_{\bb{S}_{\bm{U}}}\|_1+\|\bm{\widehat{U}}_{\bb{S}_{\bm{U}}^\complement}\|_1\geq \|\bm{U}_{\bb{S}_{\bm{U}}}\|_1-\|(\bm{\widehat{\Delta}}_{\bm{u}})_{\bb{S}_{\bm{U}}}\|_1+\|\bm{\widehat{U}}_{\bb{S}_{\bm{U}}^\complement}\|_1,
	\end{equation}
	whereas on the right-hand side, $\|\bm{\widehat{\Delta}}_{\bm{u}}\|_1=\|(\bm{\widehat{\Delta}}_{\bm{u}})_{\bb{S}_{\bm{U}}}\|_1+\|\bm{\widehat{U}}_{\bb{S}_{\bm{U}}^\complement}\|_1$.
	So we have
	\begin{equation}
	2T^{-1}\|\bm{Z}\bm{\widehat{\Delta}}\|_2^2+\lambda\|(\bm{\widehat{\Delta}}_{\bm{u}})_{\bb{S}_{\bm{U}}^\complement}\|_1\leq \lambda\|\bm{\widehat{\Delta}}_{\bm{g}}\|_1+3\lambda\|(\bm{\widehat{\Delta}}_{\bm{u}})_{\bb{S}_{\bm{U}}}\|_1
	\end{equation}
	
	Next, we assume that there is a lower bound for $2T^{-1}\|\bm{Z}\bm{\widehat{\Delta}}\|_2^2$. Thus, we define the event $\mathcal{I}_2=\{2T^{-1}\|\bm{Z}\bm{\widehat{\Delta}}\|_2^2\geq\alpha\|\widehat{\bm{\Delta}}\|_2^2\}$, where $\alpha=\lambda_{\min}(\bm{\Sigma}_{\bm{\epsilon}})/\mu_{\max}(\mathcal{A})$. On the event $\mathcal{I}_2$,
	\begin{equation}
	\begin{split}
	\label{eq:maininequality}
	\alpha\|\bm{\widehat{\Delta}}\|_2^2&\leq2T^{-1}\|\bm{Z}\bm{\widehat{\Delta}}\|_2^2\leq\lambda\|\bm{\widehat{\Delta}}_{\bm{g}}\|_1+3\lambda\|(\bm{\widehat{\Delta}}_{\bm{u}})_{\bb{S}_{\bm{U}}}\|_1\\
	&\leq \lambda\sqrt{r_1r_2r_3}\|\bm{\widehat{\Delta}}_g\|_2+3\lambda\sqrt{r_1r_2r_3}\sqrt{s_1s_2s_3}\|(\bm{\widehat{\Delta}}_{\bm{u}})_{\bb{S}_{\bm{U}}}\|_2\\
	&\leq \lambda\sqrt{r_1r_2r_3}\|\bm{\widehat{\Delta}}_g\|_2+3\lambda\sqrt{r_1r_2r_3}\sqrt{s_1s_2s_3}\|\bm{\widehat{\Delta}}_{\bm{u}}\|_{\textup{F}}.\\
	\end{split}
	\end{equation}
	
	By the perturbation bound for HOSVD in Lemma 1, we have
	\begin{equation}
	\begin{split}
	&\|\bm{\widehat{\Delta}}_{\bm{u}}\|_{\textup{F}}=\|\bm{U}_3\otimes\bm{U}_2\otimes\bm{U}_1-\bm{\widehat{U}}_3\otimes\bm{\widehat{U}}_2\otimes\bm{\widehat{U}}_1\|_{\textup{F}}\\
	\leq&\|\bm{U}_3\otimes\bm{U}_2\otimes\bm{U}_1-\bm{\widehat{U}}_3\otimes\bm{U}_2\otimes\bm{U}_1\|_{\textup{F}}
	+\|\bm{\widehat{U}}_3\otimes\bm{U}_2\otimes\bm{U}_1-\bm{\widehat{U}}_3\otimes\bm{\widehat{U}}_2\otimes\bm{U}_1\|_{\textup{F}}\\
	+&\|\bm{\widehat{U}}_3\otimes\bm{\widehat{U}}_2\otimes\bm{U}_1-\bm{\widehat{U}}_3\otimes\bm{\widehat{U}}_2\otimes \bm{\widehat{U}}_1\|_{\textup{F}}\\
	\leq&\sqrt{r_1r_2}\|\bm{U}_3-\bm{\widehat{U}}_3\|_{\textup{F}}+\sqrt{r_1r_3}\|\bm{U}_2-\bm{\widehat{U}}_2\|_{\textup{F}}+\sqrt{r_2r_3}\|\bm{U}_1-\bm{\widehat{U}}_1\|_{\textup{F}}\\
	\leq&c\tau\|\bm{\widehat{\Delta}}\|_2,
	\end{split}
	\end{equation}
	where $\widetilde{\tau}=\delta^{-1}(\eta_1\sqrt{r_2r_3}+\eta_2\sqrt{r_1r_3}+\eta_3\sqrt{r_1r_2})$, and
	\begin{equation}
	\|\bm{\widehat{\Delta}}_{\bm{g}}\|_2\leq C\delta^{-1}(\eta_1+\eta_2+\eta_3)\|\bm{\widehat{\Delta}}\|_2.
	\end{equation}
	
	Therefore, we have
	\begin{equation}
	\alpha\|\bm{\widehat{\Delta}}\|_2^2\leq C\tau\sqrt{s_1s_2s_3}\lambda\|\bm{\widehat{\Delta}}\|_2,
	\end{equation}
	where $\tau=\delta^{-1}r_1r_2r_3\sum_{i=1}^3\eta_i/\sqrt{r_i}$.
	
	If we denote $S=s_1s_2s_3$, we can obtain the estimation error bound and in-sample prediction error bound
	\begin{equation}
	\|\bm{\widehat{\Delta}}\|_2\leq C_1\tau\sqrt{S}\lambda/\alpha,
	~~\text{and}~~T^{-1}\|\bm{Z}\bm{\widehat{\Delta}}\|_2^2\leq C_2\tau^2S\lambda^2/\alpha,
	\end{equation}
	which conclude the deterministic analysis.
	
	In the second part, we show that the events $\mathcal{I}_1$ and $\mathcal{I}_2$ occur with high probability. In the high-dimensional regression literature, the conditions in $\mathcal{I}_1$ and $\mathcal{I}_2$ are known as deviation bound condition and restricted eigenvalue condition. We defer the proof of both conditions to Lemma \ref{lemma:DB} and \ref{lemma:RE}, where we show that $\mathcal{I}_1$ and $\mathcal{I}_2$ hold simultaneously with probability at least $1-C\exp[-c\log(N^2P)]-C\exp\{-cd\min[\log(NP),\log(cNP/d)]\}$, given that the sample size $T\gtrsim\log(N^2P)+\mathcal{M}^2d\min[\log(NP),\log(cNP/d)]$.
\end{proof}

\begin{proof}[Proof of Theorem \ref{thm:rank}]
	By definition, for any tensor $\cm{T}\in\mathbb{R}^{p_1\times p_2\times p_3}$, where $p_1=p_2=N$ and $p_3=P$,
	\begin{equation}
	\|\cm{T}\|_{\text{F}}^2=\|\cm{T}_{(i)}\|_{\text{F}}^2=\sum_{j=1}^{p_i}\sigma_j^2(\cm{T}_{(i)}),~~1\leq i\leq j.
	\end{equation}
	In other words, the Frobenius norm of the error tensor is equivalent to the $\ell_2$ norm of the singular values of any matricization. By Mirsky's singular value inequality \citep{mirsky1960symmetric},
	\begin{equation}
	\sum_{j=1}^{p_j}[\sigma_j((\cm{\widehat{A}}_0)_{(i)})-\sigma_j(\cm{A}_{(i)})]^2\leq\sum_{j=1}^{p_j}\sigma_j^2((\cm{\widehat{A}}_0)_{(i)}-\cm{A}_{(i)})=\|\cm{\widehat{A}}_0-\cm{A}\|_{\text{F}}^2.
	\end{equation}
	Obviously, the $\ell_\infty$ error bound is smaller than the $\ell_2$ error bound, so it follows the same upper bound,
	\begin{equation}
	\max_{1\leq j\leq p_j}|\sigma_j((\cm{\widehat{A}}_0)_{(i)})-\sigma_j(\cm{A}_{(i)})|\leq\left\{\sum_{j=1}^{p_j}[\sigma_j((\cm{\widehat{A}}_0)_{(i)})-\sigma_j(\cm{A}_{(i)})]^2\right\}^{1/2}\leq\|\cm{\widehat{A}}_0-\cm{A}\|_{\text{F}}\lesssim B(T).
	\end{equation}
	
	Note that $\sigma_j(\cm{\widehat{A}}_{(i)})+c=\sigma_j(\cm{A}_{(i)})+[\sigma_j(\cm{\widehat{A}}_{(i)})-\sigma_j(\cm{A}_{(i)})]+c$. For $j>r_i$, since $\sigma_j(\cm{A}_{(i)})=0$ and $\sigma_j(\cm{\widehat{A}}_{(i)})-\sigma_j(\cm{A}_{(i)})=o_p(c)$, $c$ is the dominating term in $\sigma_j(\cm{\widehat{A}}_{(i)})+c$. For $j\leq r_i$, since $\sigma_j(\cm{\widehat{A}}_{(i)})-\sigma_j(\cm{A}_{(i)})=o_p(c)$ and $c=o(\sigma_j(\cm{A}_{(i)}))$, $\sigma_j(\cm{A}_{(i)})$ is the dominating term.
	
	Hence, for $j>r_i$ as $T\to\infty$,
	\begin{equation}
	\frac{\sigma_{j+1}(\cm{\widehat{A}}_{(i)})+c}{\sigma_{j}(\cm{\widehat{A}}_{(i)})+c}\to\frac{c}{c}=1.
	\end{equation}
	For $j<r_i$,
	\begin{equation}
	\frac{\sigma_{j+1}(\cm{\widehat{A}}_{(i)})+c}{\sigma_{j}(\cm{\widehat{A}}_{(i)})+c}\to\frac{\sigma_{j+1}(\cm{A}_{(i)})}{\sigma_{j}(\cm{A}_{(i)})}.
	\end{equation}
	For $j=r_i$,
	\begin{equation}
	\frac{\sigma_{j+1}(\cm{\widehat{A}}_{(i)})+c}{\sigma_{j}(\cm{\widehat{A}}_{(i)})+c}\to\frac{c}{\sigma_{r_i}(\cm{A}_{(i)})}=o\left(\min_{1\leq j\leq r_i-1}\frac{\sigma_{j+1}(\cm{A}_{(i)})}{\sigma_{j}(\cm{A}_{(i)})}\right).
	\end{equation}

\end{proof}

\section{Proofs of Corollaries \ref{cor1} and \ref{cor:comparison}}

\begin{proof}[Proof of Corollary \ref{cor1}]
	Here we prove the asymptotic normality for $\text{vec}(\bm{\widehat{U}}_1)$, since the proofs for the $\text{vec}(\bm{\widehat{U}}_2)$ and $\text{vec}(\bm{\widehat{U}}_3)$ are similar. In this part, we simplify $\cm{\widehat{A}}_{\textup{MLR}}$ to $\cm{\widehat{A}}$. Note that $\bm{\widehat{U}}_1$ and $\bm{U}_1$ are the eigenvectors of $\cm{\widehat{A}}_{(1)}\cm{\widehat{A}}_{(1)}'$ and $\cm{A}_{(1)}\cm{A}_{(1)}'$ respectively. By Theorem 1, $\sqrt{T}\text{vec}(\cm{\widehat{A}}_{(1)}-\cm{A}_{(1)})\to_dN(\bm{0},\bm{\Sigma}_{\bm{h}})$.
	Note that
	\begin{equation}
	\begin{split}
	&\sqrt{T}(\cm{\widehat{A}}_{(1)}\cm{\widehat{A}}_{(1)}'-\cm{A}_{(1)}\cm{A}_{(1)}')\\
	=&\sqrt{T}(\cm{\widehat{A}}_{(1)}-\cm{A}_{(1)})\cm{A}_{(1)}'+\sqrt{T}\cm{A}_{(1)}(\cm{\widehat{A}}_{(1)}-\cm{A}_{(1)})'+\sqrt{T}(\cm{\widehat{A}}_{(1)}-\cm{A}_{(1)})(\cm{\widehat{A}}_{(1)}-\cm{A}_{(1)})',
	\end{split}
	\end{equation}
	so we have
	\begin{equation}
	\begin{split}
	&\sqrt{T}\text{vec}(\cm{\widehat{A}}_{(1)}\cm{\widehat{A}}_{(1)}'-\cm{A}_{(1)}\cm{A}_{(1)}')\\
	=&(\cm{A}_{(1)}\otimes\bm{I}_N)\sqrt{T}\text{vec}(\cm{\widehat{A}}_{(1)}-\cm{A}_{(1)})+(\bm{I}_N\otimes \cm{A}_{(1)})\sqrt{T}\text{vec}(\cm{\widehat{A}}_{(1)}-\cm{A}_{(1)})+O_p(T^{-1/2}).
	\end{split}
	\end{equation}
	Therefore, $\sqrt{T}\text{vec}(\cm{\widehat{A}}_{(1)}\cm{\widehat{A}}_{(1)}'-\cm{A}_{(1)}\cm{A}_{(1)}')$ is asymptotically normally distributed.
	
	By the matrix perturbation expansion \citep{Izenman75,Velu98},
	\begin{equation}
	\sqrt{T}(\bm{\widehat{U}}_{1k}-\bm{U}_{1k})=\sum_{i\neq k}\frac{1}{d_k^2-d_i^2}(\bm{U}_{1i}'\otimes \bm{U}_{1i}\bm{U}_{1i}')\sqrt{T}\text{vec}(\cm{\widehat{A}}_{(1)}\cm{\widehat{A}}_{(1)}'-\cm{A}_{(1)}\cm{A}_{(1)}')+O_p(T^{-1/2}).
	\end{equation}
	Therefore, $\sqrt{T}(\bm{\widehat{U}}_1-\bm{U}_1)$ is also asymptotically normally distributed.
	
	For $\text{vec}(\cm{\widehat{G}}_{(1)})$, by the definition of HOSVD,
	\begin{equation}
	\cm{\widehat{G}}=[\![\cm{\widehat{A}};\bm{\widehat{U}}'_1,\bm{\widehat{U}}'_2,\bm{\widehat{U}}'_3]\!],~~\text{and}~~\cm{G}=[\![\cm{A};\bm{U}'_1,\bm{U}'_2,\bm{U}'_3]\!].
	\end{equation}
	So we have
	\begin{equation}\begin{split}
	&\text{vec}(\cm{\widehat{G}}_{(1)}-\cm{G}_{(1)})\\
	=&(\bm{\widehat{U}}_3'\otimes\bm{\widehat{U}}_2'\otimes\bm{\widehat{U}}_1')\text{vec}(\cm{\widehat{A}}_{(1)})-(\bm{{U}}_3'\otimes\bm{{U}}_2'\otimes\bm{{U}}_1')\text{vec}(\cm{{A}}_{(1)})\\
	=&(\bm{\widehat{U}}_3'\otimes\bm{\widehat{U}}_2'\otimes\bm{\widehat{U}}_1'-\bm{{U}}_3'\otimes\bm{{U}}_2'\otimes\bm{{U}}_1')\text{vec}(\cm{\widehat{A}}_{(1)})
	+(\bm{{U}}_3'\otimes\bm{{U}}_2'\otimes\bm{{U}}_1')\text{vec}(\cm{\widehat{A}}_{(1)}-\cm{A}_{(1)})\\
	=&[(\bm{\widehat{U}}'_3-\bm{U}'_3)\otimes\bm{U}'_2\otimes\bm{U}'_1]\text{vec}(\cm{A}_{(1)})+[\bm{U}'_3\otimes(\bm{\widehat{U}}_2-\bm{U}_2)'\otimes\bm{U}'_1]\text{vec}(\cm{A}_{(1)})\\
	&+[\bm{U}'_3\otimes\bm{U}'_2\otimes(\bm{\widehat{U}}'_1-\bm{U}'_1)]\text{vec}(\cm{A}_{(1)})+(\bm{{U}}_3'\otimes\bm{{U}}_2'\otimes\bm{{U}}_1')\text{vec}(\cm{\widehat{A}}_{(1)}-\cm{A}_{(1)})+o_p(T^{-1/2})\\
	=&[\bm{I}_{r_3}\otimes((\bm{U}'_2\otimes\bm{U}'_1)\cm{A}'_{(3)})]\text{vec}(\bm{\widehat{U}}_3-\bm{U}_3)+[\bm{I}_{r_2}\otimes((\bm{U}'_3\otimes\bm{U}'_1)\cm{A}'_{(2)})]\text{vec}(\bm{\widehat{U}}_2-\bm{U}_2)\\
	&+[\bm{I}_{r_1}\otimes((\bm{U}'_3\otimes\bm{U}'_2)\cm{A}'_{(1)})]\text{vec}(\bm{\widehat{U}}_1-\bm{U}_1)+(\bm{{U}}_3'\otimes\bm{{U}}_2'\otimes\bm{{U}}_1')\text{vec}(\cm{\widehat{A}}_{(1)}-\cm{A}_{(1)})\\&+o_p(T^{-1/2}).
	\end{split}\end{equation}
	Therefore, $\sqrt{T}\text{vec}(\cm{\widehat{G}}_{(1)}-\cm{G}_{(1)})$ is also normally distributed with mean zero, as $T\to\infty$. For simplicity, we omit the covariance of each component, but they can be easily calculated by the above formula.
\end{proof}

\begin{proof}[Proof of Corollary \ref{cor:comparison}]
	The $\sqrt{T}$-consistency and asymptotic normality of $\cm{\widehat{A}}_{\textup{OLS}}$ has been studied in the proof of Theorem 1, with $\bm{\Sigma}_{\textup{OLS}}=\bm{J}^{-1}$.
	
	As discussed previously, $\bm{\Sigma}_{\textup{MLR}}=\bm{P}\bm{J}^{-1}\bm{P}'$ where $\bm{P}$ is a projection matrix. Note that $\bm{J}^{-1}-\bm{H}(\bm{H}'\bm{J}\bm{H})^\dagger\bm{H}'=\bm{J}^{-1/2}\bm{Q}_{\bm{J}^{1/2}\bm{H}}\bm{J}^{-1/2}$, where $\bm{Q}_{\bm{J}^{1/2}\bm{H}}$ is the projection matrix onto the orthogonal compliment of $\text{span}(\bm{J}^{1/2}\bm{H})$. Then, it is clear that $\bm{J}^{-1}\geq\bm{H}(\bm{H}'\bm{J}\bm{H})^\dagger\bm{H}'$.
	
	For the RRR estimator, the components in the SVD, $\bm{A}=\bm{U}\bm{D}\bm{V}'$, can be denoted as $\bm{\theta}=(\text{vec}(\bm{U})',\text{diag}(\bm{D})',\text{vec}(\bm{V})')'$. Therefore, the gradient matrix of the RRR is $\bm{R}=\partial\bm{h}/\partial\bm{\theta}$.
	Since $\bm{U}_1$ in Tucker decomposition is exactly the same as the left singular vectors in the SVD of $\cm{A}_{(1)}$, we can view the Tucker decomposition as a further decomposition of the matrix $\bm{DV}'$. Therefore, $\bm{H}=\partial\bm{h}/\partial\bm{\phi}=\partial\bm{h}/\partial\bm{\theta}\cdot\partial\bm{\theta}/\partial\bm{\phi}=\bm{R}\cdot\partial\bm{\theta}/\partial\bm{\phi}$.
	By similar arguments in the proof of Theorem 1, we can obtain that the RRR estimator has the asymptotic covariance $\bm{\Sigma}_{\textup{RRR}}=\bm{R}(\bm{R}'\bm{JR})^\dagger\bm{R}'$ and it is smaller than or equal to $\bm{\Sigma}_{\textup{MLR}}$ since $\text{span}(\bm{J}^{1/2}\bm{R})\subset\text{span}(\bm{J}^{1/2}\bm{H})$.
\end{proof}

\section{Proofs of Propositions 1 and 2}

\begin{proof}[Proof of Proposition 1]
	Proof of global convergence hinges on standard arguments for block relaxation algorithm \citep{Lange10}. Note that although the objective function is nonconvex, the subproblem of each updating step is well-defined, differentiable and convex. Since the algorithm decreases the objective function monotonically, the convergence is guaranteed and any convergent point is a stationary point. With a slight abuse of notation, denote the objective function of $\bm{\phi}$ by $L(\bm{\phi})=L(\cm{G},\bm{U}_1,\bm{U}_2,\bm{U}_3)$, and then global convergence is guaranteed under the following conditions: (i) $L$ is coercive; (ii) the stationary points of $L$ are isolated; (iii) the algorithm mapping is continuous; (iv) $\bm{\phi}$ is a fixed point of the algorithm if and only if it is a stationary point of $L$; (v) $L(\bm{\phi}^{(t+1)})\leq L(\bm{\phi}^{(t)})$ with equality if and only if $\bm{\phi}^{(t)}$ is a fixed point of the algorithm.
	
	Condition (i) is guaranteed by the compactness of the set $\{\bm{\phi}:L(\bm{\phi})\leq L(\bm{\phi}^{(0)})\}$. Condition (ii) is assumed. Condition (iii) follows from the implicit function theorem since the algorithmic map $M$ is a composition of four differentiable and convex maps. A fixed point $\bm{\phi}$ satisfies that $\nabla_{\cm{G}}L(\bm{\phi})=0$
	and $\nabla_{\bm{U}_i}L(\bm{\phi})=0$; therefore the fixed point of the mapping $M(\bm{\phi})$, i.e., condition (iv) is satisfied. Finally, each step monotonically decreases $L(\bm{\phi})$, so they give a strict decrease if and only if they actually change the corresponding components. Hence, condition (v) is satisfied.
	
	Proof of local convergence hinges on the Ostrowski's theorem, which states that the sequence $\bm{\phi}^{(t+1)}=M(\bm{\phi}^{(t)})$ is locally attracted to $\bm{\phi}^{(\infty)}$ if the spectral radius of the differential of the algorithmic map $\rho[dM(\bm{\phi}^{(\infty)})]<1$. The condition can be shown to be true based on the local convergence of block relaxation algorithm by \citet{Lange10}, and we omit the detailed proof here.
\end{proof}

\begin{proof}[Proof of Proposition 2]
	The proof of this proposition follows the same lines as that of Theorem 3 in \cite{Uematsu17}. In this analysis, we fix parameters $\bm{\rho}$ and Lagrangian multipliers $\cm{C}_i$, and consider the objective function $\mathcal{L}_{\bm{\varrho}}(\cm{G},\{\bm{U}_i\},\{\bm{D}_i\},\{\bm{V}_i\};\{\cm{C}_i\})$. By the nature of the ADMM algorithm applied to $\cm{G}, \bm{U}_i, \bm{D}_i, \bm{V}_i$, the resulting sequence $\{\mathcal{L}_{\bm{\varrho}}(\cdot)\}$ is non-increasing. Clearly, the value of the function $\mathcal{L}_{\bm{\varrho}}(\cm{G},\{\bm{U}_i\},\{\bm{D}_i\},\{\bm{V}_i\};\{\cm{C}_i\})$ is bounded from below. Hence, the sequence $\{\mathcal{L}_{\bm{\varrho}}(\cdot)\}$ converges. 
	
	Note that the subspace of matrices with the orthonormal constraint is a Stiefel manifold which is compact and smooth. Moreover, on the Stiefel manifold, the objective function $\mathcal{L}_{\bm{\varrho}}(\cdot)$ is strongly convex with respect to any one of the blocks $\bm{U}_i$ and $\bm{V}_i$, for $i=1,2,3$, when all the other blocks are fixed. 
	
	Now consider the $\bm{U}_1$-update first. By definition of $\bm{U}_1^{(k+1)}$ in the algorithm, the value of the gradient of $\mathcal{L}_{\bm{\varrho}}(\cm{G},\cdot,\bm{U}_2^{(k)},\bm{U}_3^{(k)},\{\bm{D}_i^{(k)}\},\{\bm{V}_i^{(k)}\})$  with respect to $\bm{U}_1$  at  $\bm{U}_1^{(k+1)}$ vanishes. Thus, it follows easily from the strong convexity of $\mathcal{L}_{\bm{\varrho}}$ that $\Delta\mathcal{L}_{\bm{\varrho}}(\bm{U}_1^{(k+1)})\geq \delta d_g^2(\bm{U}_1^{(k+1)}, \bm{U}_1^{(k)})$, where $d_g(\cdot, \cdot)$ denotes the distance function on the Stiefel manifold, and $\delta>0$ is a constant.  Then it holds
	\[
	\sum_{k=0}^{\infty}d_g(\bm{U}_1^{(k+1)}, \bm{U}_1^{(k)}) \leq \delta^{-1/2} \sum_{k=0}^{\infty} [\Delta\mathcal{L}_{\bm{\varrho}}(\bm{U}_1^{(k+1)})]^{1/2} <\infty,
	\]
	which entails that $\{\bm{U}_1^{(k)}\}$ is a Cauchy sequence on the Stiefel manifold.  Thus, the sequence $\{\bm{U}_1^{(k)}\}$ converges to a limit point on the Stiefel manifold, i.e., a blockwise local minima with respect to $\bm{U}_1$. Similar arguments can be applied to $\bm{U}_2$, $\bm{U}_3$ and $\bm{V}_i$ for $i=1,2,3$. 
	
	Since we do not impose any constraints on $\cm{G}$ and $\bm{D}_i$ blocks, the objective function along one of these blocks with all the other blocks fixed is strictly convex. Thus, by the similar arguments, the sequences $\{\cm{G}^{(k)}\}$ and $\{\bm{D}_i^{(k)}\}$ also converge to a limit point. This completes the proof of Proposition 2.
\end{proof}

\section{Five lemmas used for the proof of Theorem 2 \label{append:lems}}

\begin{lemma} (HOSVD perturbation bound)
	\label{lemma:perturbation}
	Suppose that $\cm{A}=[\![\cm{G};\bm{U}_1,\bm{U}_2,\bm{U}_3]\!]$ and $\cm{\widetilde{A}}=[\![\cm{\widetilde{G}};\bm{\widetilde{U}}_1,\bm{\widetilde{U}}_2,\bm{\widetilde{U}}_3]\!]$ are two HOSVD for $\cm{A}$ and $\cm{\widetilde{A}}$, with the same multilinear ranks $(r_1,r_2,r_3)$. Under Assumptions \ref{asmp:RPS} and \ref{asmp:RSG}, we have
	\begin{equation}
	\|\cm{\widetilde{G}}-\cm{G}\|_{\textup{F}}\leq \frac{C(\eta_1+\eta_2+\eta_3)}{\delta}\|\cm{\widetilde{A}}-\cm{A}\|_{\textup{F}}.
	\end{equation}
	and
	\begin{equation}
	\|\bm{\widetilde{U}}_i-\bm{U}_i\|_{\textup{F}}\leq \frac{C\eta_i}{\delta}\|\cm{\widetilde{A}}-\cm{A}\|_{\textup{F}},
	\end{equation}
	where $\eta_i=\sum_{j=1}^{r_i}\sigma^2_1(\cm{A}_{(i)})/\sigma^2_j(\cm{A}_{(i)})$.
\end{lemma}

\begin{proof}[Proof of Lemma 1]
	Since both $[\![\cm{G};\bm{U}_1,\bm{U}_2,\bm{U}_3]\!]$ and $[\![\cm{\widetilde{G}};\bm{\widetilde{U}}_1,\bm{\widetilde{U}}_2,\bm{\widetilde{U}}_3]\!]$ are HOSVD for $\cm{A}$ and $\cm{\widetilde{A}}$.
	Each factor matrix is exactly the left singular vectors of the corresponding tensor matricization. We can apply the matrix perturbation theory for the factor matrices.
	
	By the extension of the Davis-Kahan theorem for singular decomposition, i.e. Theorem 3 in \citet{Yu15}, under Assumption 6, we have for the $j$-th singular vector of $\cm{A}_{(1)}$,
	\begin{equation}
	\|\bm{\widetilde{U}}_{ij}-\bm{U}_{ij}\|_{\textup{F}}\leq\frac{c(2\sigma_1(\cm{G}_{(i)})+\|\cm{\widetilde{A}}_{(i)}-\cm{A}_{(i)}\|_{\textup{op}})\|\cm{\widetilde{A}}_{(i)}-\cm{A}_{(i)}\|_{\textup{F}}}{\min[\sigma^2_{j-1}(\cm{G}_i)
		-\sigma^2_{j}(\cm{G}_i),\sigma^2_{j}(\cm{G}_i)-\sigma^2_{j+1}(\cm{G}_i)]}\leq\frac{c\sigma_1(\cm{A}_{(i)})\|\cm{\widetilde{A}}-\cm{A}\|_{\textup{F}}}{\delta\sigma^2_j(\cm{A}_{(i)})}.
	\end{equation}
	Therefore,
	\begin{equation*}
	\|\bm{\widetilde{U}}_i-\bm{U}_i\|^2_{\textup{F}}= \sum_{j=1}^{r_i}\|\bm{\widetilde{U}}_{ij}-\bm{U}_{ij}\|^2_{\textup{F}}\leq \frac{c\eta_i^2}{\delta^2}\|\cm{\widetilde{A}}-\cm{A}\|_{\textup{F}}^2.
	\end{equation*}	
	
	By the HOSVD, $\cm{G}=[\![\cm{A};\bm{U}_1',\bm{U}_2',\bm{U}_3']\!]$ and $\cm{\widetilde{G}}=[\![\cm{\widetilde{A}};\bm{\widetilde{U}}_1',\bm{\widetilde{U}}_2',\bm{\widetilde{U}}_3']\!]$. Then, we have
	\begin{equation}\label{eq:split}
	\begin{split}
	\|\cm{\widetilde{G}}-\cm{G}\|_{\textup{F}}&=\|(\bm{\widetilde{U}}_3\otimes\bm{\widetilde{U}}_2\otimes\bm{\widetilde{U}}_1)'\textup{vec}(\cm{\widetilde{A}})-(\bm{U}_3\otimes\bm{U}_2\otimes\bm{U}_1)'\textup{vec}(\cm{A})\|_{\textup{F}}\\
	&=\|(\bm{\widetilde{U}}_3\otimes\bm{\widetilde{U}}_2\otimes\bm{\widetilde{U}}_1)'\textup{vec}(\cm{\widetilde{A}}-\cm{A})\|_{\textup{F}}+\|(\bm{\widetilde{U}}_3-\bm{U}_3)'\cm{A}_{(3)}(\bm{\widetilde{U}}_2\otimes\bm{\widetilde{U}}_1)\|_{\text{F}}\\
	&+\|(\bm{\widetilde{U}}_2-\bm{U}_2)'\cm{A}_{(2)}(\bm{U}_3\otimes\bm{\widetilde{U}}_1)\|_{\text{F}}+\|(\bm{\widetilde{U}}_1-\bm{U}_1)'\cm{A}_{(1)}(\bm{U}_3\otimes\bm{U}_2)\|_{\text{F}}.
	\end{split}
	\end{equation}
	Since $\bm{\widetilde{U}}_3\otimes\bm{\widetilde{U}}_2\otimes\bm{\widetilde{U}}_1$ is orthonormal, $\|(\bm{\widetilde{U}}_3\otimes\bm{\widetilde{U}}_2\otimes\bm{\widetilde{U}}_1)'\textup{vec}(\cm{\widetilde{A}}-\cm{A})\|_{\textup{F}}\leq\|\cm{\widetilde{A}}-\cm{A}\|_{\textup{F}}$. In addition $(\bm{U}_3\otimes\bm{U}_2)$ is orthonormal, so
	\begin{eqnarray}
	\|(\bm{\widetilde{U}}_1-\bm{U}_1)'\cm{A}_{(1)}(\bm{U}_3\otimes\bm{U}_2)\|_{\text{F}}\leq\|(\bm{\widetilde{U}}_1-\bm{U}_1)'\cm{A}_{(1)}\|_{\text{F}}\leq \bar{g}\|\bm{\widetilde{U}}_1-\bm{U}_1\|_{\text{F}},
	\end{eqnarray}
	where the last inequality follows from $\lambda_{\max}^{1/2}(\bm{A}_{(1)}\bm{A}_{(1)}')=\sigma_1(\bm{A}_{(1)})\leq\bar{g}$. We can also bound the other two terms in \eqref{eq:split} in the similar way.
	
	In summary, we can have
	\begin{equation}
	\|\cm{\widetilde{G}}-\cm{G}\|_{\text{F}}\leq\frac{C(\eta_1+\eta_2+\eta_3)}{\delta}\|\cm{\widetilde{A}}-\cm{A}\|_{\textup{F}}.
	\end{equation}
\end{proof}

\begin{lemma} (Deviation bound)
	\label{lemma:DB}
	Under the conditions of Theorem \ref{thm:errorbound}, for $T\gtrsim\log(N^2P)$, the following two inequalities
	\begin{equation}\begin{split}
	&\langle T^{-1}\bm{Z}'\bm{e},\bm{\widehat{\Delta}}_{\bm{u}}\bm{\widehat{g}}\rangle\leq C\mathcal{M}\sqrt{\log(N^2P)/T}\|\bm{\widehat{\Delta}}_{\bm{u}}\|_1,\\
	\text{and}~&\langle T^{-1}\bm{Z}'\bm{e},\bm{U}\bm{\widehat{\Delta}}_{\bm{g}}\rangle\leq C\mathcal{M}\sqrt{\log(N^2P)/T}\|\bm{\widehat{\Delta}}_{\bm{g}}\|_1,
	\end{split}\end{equation}
	hold with probability at least $1-C\exp(-c\log(N^2P))$.
\end{lemma}

\begin{proof}[Proof of Lemma 2]
	For the first inequality,
	\begin{equation}
	\langle T^{-1}\bm{Z}'\bm{e},\bm{\widehat{\Delta}}_{\bm{u}}\bm{\widehat{g}}\rangle=\langle T^{-1}\bm{Z}'\bm{e\widehat{g}}',\bm{\widehat{\Delta}}_{\bm{u}}\rangle\leq\|T^{-1}\bm{Z}'\bm{e\widehat{g}}'\|_\infty\|\bm{\widehat{\Delta}}_{\bm{u}}\|_1\leq \bar{g}\|T^{-1}\bm{Z}'\bm{e}\|_\infty\|\bm{\widehat{\Delta}}_{\bm{u}}\|_1,
	\end{equation}
	where
	$\|\bm{Z}'\bm{e}\|_\infty=\|\bm{X}'\bm{E}\|_\infty=\max_{1\leq i\leq NP,1\leq j\leq N}|\bm{e}_i'\bm{X}'\bm{E}\bm{e}_j|$ where $\bm{e}_i$ is a coordinate vector whose $i$-th entry is 1 and the others are 0.
	
	By Lemma 4, for any vector $\bm{u}$ and $\bm{v}$ s.t. $\|\bm{u}\|_2=\|\bm{v}\|_2=1$, and $\eta>0$,
	\begin{equation}
	\mathbb{P}\left[|\bm{u}'(\bm{X}'\bm{E}/T)\bm{v}|>2\pi\left(\lambda_{\max}(\bm{\Sigma}_{\bm{\epsilon}})\left(1+\frac{\mu_{\max}(\mathcal{A})}{\mu_{\min}(\mathcal{A})}\right)\right)\eta\right]\leq6\exp[-cT\min(\eta,\eta^2)].
	\end{equation}
	
	Therefore, if we denote $\mathcal{M}\equiv\lambda_{\max}(\bm{\Sigma}_{\bm{\epsilon}})[1+\mu_{\max}(\mathcal{A})/\mu_{\min}(\mathcal{A})]$ and take a union bound,
	\begin{equation}
	\mathbb{P}\left[\max_{1\leq i\leq NP,1\leq j\leq N}|\bm{e}_i'\bm{X}'\bm{Ee}_j/T|>2\pi\mathcal{M}\eta\right]\leq6N^2P\exp[-cT\min(\eta,\eta^2)].
	\end{equation}
	Take $\eta=\sqrt{\log(N^2P)/T}$ and we obtain
	\begin{equation}
	\mathbb{P}\left[\max_{1\leq i\leq NP,1\leq j\leq N}|\bm{e}_i'\bm{X}'\bm{Ee}_j/T|>2\pi\mathcal{M}\sqrt{\log(N^2P)/T}\right]\leq C\exp[-c\log(N^2P)].
	\end{equation}
	
	For the second inequality,
	\begin{equation}
	\begin{split}
	&\langle T^{-1}\bm{Z}'\bm{e},\bm{U}\bm{\widehat{\Delta}}_{\bm{g}}\rangle\leq\|T^{-1}\bm{U}\bm{Z}'\bm{e}\|_\infty\|\bm{\widehat{\Delta}}_{\bm{g}}\|_1\\
	=&\|T^{-1}(\bm{U}_3\otimes\bm{U}_2)'\bm{X}'\bm{E}\bm{U}_1\|_\infty\|\bm{\widehat{\Delta}}_{\bm{g}}\|_1\\
	=&\max_{1\leq i\leq r_2r_3,1\leq j\leq r_1}|T^{-1}\bm{e}_i'(\bm{U}_3\otimes\bm{U}_2)'\bm{X}'\bm{E}\bm{U}_1\bm{e}_j|\cdot\|\bm{\widehat{\Delta}}_{\bm{g}}\|_1.
	\end{split}
	\end{equation}
	
	For any orthonormal matrix $\bm{U}_3\otimes\bm{U}_2$, the spectral density of $\{\bm{X}_t(\bm{U}_3\otimes\bm{U}_2)\}$ is defined as
	\begin{equation}
	f_{\bm{X}(\bm{U}_3\otimes\bm{U}_2)}(\theta)=\frac{1}{2\pi}\sum_{\ell=-\infty}^\infty(\bm{U}_3\otimes\bm{U}_2)'\bm{\Gamma}_{\bm{X}}(\ell)(\bm{U}_3\otimes
	\bm{U}_2)e^{-i\ell\theta}=(\bm{U}_3\otimes\bm{U}_2)'f_{\bm{X}}(\theta)(\bm{U}_3\otimes\bm{U}_2),
	\end{equation}
	so we have $\mathcal{M}(f_{\bm{X}(\bm{U}_3\otimes\bm{U}_2)})\leq\mathcal{M}(f_{\bm{X}})$. Similarly, $\mathcal{M}(f_{\bm{EU}_1})\leq\mathcal{M}(f_{\bm{E}})$. Therefore, for any $\|\bm{u}\|_2\leq1$ and $\|\bm{v}\|_2\leq1$,
	\begin{equation}
	\mathbb{P}[|T^{-1}\bm{u}'((\bm{U}_3\otimes\bm{U}_2)\bm{X}'\bm{E}\bm{U}_1)\bm{v}|>2\pi\mathcal{M}\eta]\leq6\exp[-cT\min(\eta,\eta^2)].
	\end{equation}
	Taking a union bound, we can have
	\begin{equation}
	\mathbb{P}\left[\max_{1\leq i\leq r_2r_3,1\leq j\leq r_1}|T^{-1}\bm{e}_i'((\bm{U}_3\otimes\bm{U}_2)\bm{X}'\bm{E}\bm{U}_1)\bm{e}_j|>2\pi\mathcal{M}\eta\right]\leq Cr_1r_2r_3\exp[-cT\min(\eta,\eta^2)].
	\end{equation}
	Take $\eta=\sqrt{\log(N^2P)/T}$ and we obtain
	\begin{equation}
	\begin{split}
	&\mathbb{P}\left[\max_{1\leq i\leq r_2r_3,1\leq j\leq r_1}|T^{-1}\bm{e}_i'((\bm{U}_3\otimes\bm{U}_2)\bm{X}'\bm{E}\bm{U}_1)\bm{e}_j|>2\pi\mathcal{M}\sqrt{\log(N^2P)/T}\right]\\
	\leq&Cr_1r_2r_3\exp[-c\log(N^2P)]\leq CN^2P\exp[-c\log(N^2P)]\leq C'\exp[-\log(N^2P)].
	\end{split}
	\end{equation}
	The proof is complete.
\end{proof}

\begin{lemma} (Restricted eigenvalue)
	\label{lemma:RE}
	Under the conditions of Theorem \ref{thm:errorbound}, if the sample size $T\gtrsim\mathcal{M}^2d\min[\log(N^2P),\log(cN^2P/d)]$, for $\bm{\widehat{\Delta}}=\bm{\widehat{U}}\bm{\widehat{g}}-\bm{Ug}$, where $(\cm{\widehat{G}},\bm{\widehat{U}}_1,\bm{\widehat{U}}_2,\bm{\widehat{U}}_3)$ and $(\cm{G},\bm{U}_1,\bm{U}_2,\bm{U}_3)$ belong to $\Omega$,
	\begin{equation}
	T^{-1}\|(\bm{I}_N\otimes\bm{X})\bm{\widehat{\Delta}}\|_2^2\geq\alpha\|\bm{\widehat{\Delta}}\|_2^2/2,
	\end{equation}
	with probability at least $1-2\exp\{-cd\min[\log(N^2P),\log(cN^2P/d)]\}$, where $\alpha=\lambda_{\min}(\bm{\Sigma}_{\bm{\epsilon}})/\mu_{\max}(\mathcal{A})$, $d=2\nu^{-2}r_1r_2r_3$.
\end{lemma}

\begin{proof}[Proof of Lemma 3]
	Denote by $\mathcal{K}(s)=\{\bm{v}\in\mathbb{R}^{NP}:\|\bm{v}\|_0\leq s,\|\bm{v}\|_2\leq1\}$ the set of $s$-sparse vectors.
	
	If we split $\bm{\widehat{\Delta}}$ into $N$ parts, namely $(\bm{\widehat{\Delta}})=(\bm{\widehat{\Delta}}_1',\dots,\bm{\widehat{\Delta}}_N')'$, where $\bm{\widehat{\Delta}}_k\in\mathbb{R}^{NP}$, we have
	\begin{equation}
	\|(\bm{I}_N\otimes\bm{X})\hat{\bm{\Delta}}\|_2^2=\sum_{i=1}^N\|\bm{X}\bm{\widehat{\Delta}}_i\|_2^2.
	\end{equation}
	
	Correspondingly, we split $\bm{U}$ and $\bm{\widehat{U}}$ into $N$ blocks, $\bm{U}=(\bm{M}_1',\dots,\bm{M}_N')$ and $\bm{\widehat{U}}=(\bm{\widehat{M}}_1',\dots,\bm{\widehat{M}}_N')'$.
	Since $(\cm{\widehat{G}},\bm{\widehat{U}}_1,\bm{\widehat{U}}_2,\bm{\widehat{U}}_3)$ and $(\cm{G},\bm{U}_1,\bm{U}_2,\bm{U}_3)$ belong to $\Omega$, the square of smallest nonzero entries in $\bm{U}_i$ and $\bm{\widehat{U}}_i$ is at least $\nu$. Since each column in $\bm{U}_i$ or $\bm{\widehat{U}}_i$ has unit Euclidean norm, the number of nonzero entries in $\bm{U}_i$ or $\bm{\widehat{U}}_i$
	is at most $1/\nu$. By the Kronecker structure in $\bm{U}$ and $\bm{\widehat{U}}$, the number of nonzero entries in each column of $\bm{M}_i$ or $\bm{\widehat{M}}_i$ is at most $\nu^{-2}$.
	Therefore, $\|\bm{\widehat{\Delta}}_i\|_0\leq2\nu^{-2}r_1r_2r_3:=d$.
	
	Denote $\bm{\widehat{\Gamma}}=\bm{X}'\bm{X}/T$ and $\bm{\Gamma}=\mathbb{E}\bm{\widehat{\Gamma}}$. Since
	\begin{equation}
	\begin{split}
	T^{-1}\|(\bm{I}_N\otimes\bm{X})\bm{\widehat{\Delta}}\|_2^2=&\bm{\widehat{\Delta}}'(\bm{I}_N\otimes\bm{\widehat{\Gamma}})\bm{\widehat{\Delta}}\\
	=&\bm{\widehat{\Delta}}'(\bm{I}_N\otimes\bm{\Gamma})\bm{\widehat{\Delta}}+\bm{\widehat{\Delta}}'[\bm{I}_N\otimes(\bm{\widehat{\Gamma}}-\bm{\Gamma})]\bm{\widehat{\Delta}}\\
	=&\bm{\widehat{\Delta}}'(\bm{I}_N\otimes\bm{\Gamma})\bm{\widehat{\Delta}}+\sum_{i=1}^N\bm{\widehat{\Delta}}_i'(\bm{\widehat{\Gamma}}-\bm{\Gamma})\bm{\widehat{\Delta}}_i.
	\end{split}
	\end{equation}
	By the property of spectral density, $\lambda_{\min}(\bm{\Gamma})\geq\lambda_{\min}(\bm{\Sigma}_{\bm{\epsilon}})/\mu_{\max}(\mathcal{A})$, so $T^{-1}\mathbb{E}(\|(\bm{I}_N\otimes\bm{X})\bm{\widehat{\Delta}}\|_2^2)=T^{-1}\bm{\widehat{\Delta}}'(\bm{I}_N\otimes\bm{\Gamma})
	\bm{\widehat{\Delta}}\geq\lambda_{\min}(\bm{\Gamma})\geq\lambda_{\min}(\bm{\Sigma}_{\bm{\epsilon}})/\mu_{\max}(\mathcal{A})\|\bm{\widehat{\Delta}}\|_2^2=\alpha\|\bm{\widehat{\Delta}}\|_2^2$.
	
	So it remains to show that $\sup_{\bm{\widehat{\Delta}}_i\in\mathcal{K}(d)}\bm{\widehat{\Delta}}_i'(\bm{\widehat{\Gamma}}-\bm{\Gamma})\bm{\widehat{\Delta}}_i$ is close to zero. If we combine Lemma 4 and 5, we can obtain that for any $\eta>0$,
	\begin{equation}\begin{split}
	&\mathbb{P}\left[\sup_{\bm{u}\in\mathcal{K}(d)}\Big{|}\bm{u}'(\bm{\widehat{\Gamma}}-\bm{\Gamma})\bm{u}\Big{|}>2\pi\mathcal{M}(f_{\bm{X}})\eta\right]\\
	\leq&2\exp\{-cT\min(\eta,\eta^2)+d\min[\log(NP),\log(cNP/d)]\}.
	\end{split}\end{equation}
	
	Finally, if we take $\eta=\alpha/(4\pi\mathcal{M})$,
	\begin{equation}\begin{split}
	&\mathbb{P}\left[\|(\bm{I}_N\otimes\bm{X})\bm{\widehat{\Delta}}\|_2^2/T\geq\alpha\|\bm{\widehat{\Delta}}\|_2^2/2\right]\\
	\geq&\mathbb{P}\left[\sup_{\bm{\widehat{\Delta}}_i\in\mathcal{K}(d)}\Big{|}\bm{\widehat{\Delta}}_i'(\bm{\widehat{\Gamma}}-\bm{\Gamma})\bm{\widehat{\Delta}}_i\Big{|}<\alpha\|\bm{\widehat{\Delta}}_i\|_2^2/2\right]\\
	\geq&1-2\exp\{-cT\mathcal{M}^{-2}+2d\min[\log(NP),\log(cNP/d)]\}\\
	\geq&1-2\exp\{-cd\min[\log(NP),\log(cNP/d)]\}.
	\end{split}\end{equation}
	The proof is complete.
\end{proof}

Next, to make the proof self-contained, we state two lemmas to establish concentration inequalities for Gaussian time series from \citet{Basu15}. The first one is Proposition 2.4 in \citet{Basu15}.

\begin{lemma}
	For a stationary and centered Gaussian time series $\{\bm{x}_t\}$ satisfying the bounded spectral density condition, there exists a constant $c>0$ such that for any vector $\bm{v}\in\mathbb{R}^p$ with $\|\bm{u}\|_2\leq1$, $\|\bm{v}\|_2\leq1$, and any $\eta\geq0$,
	\begin{equation*}\begin{split}
	\mathbb{P}[|\bm{v}'(\widehat{\bm{\Gamma}}-\bm{\Gamma})\bm{v}|>2\pi\mathcal{M}(f_{\bm{X}})\eta]\leq2\exp[-cT\min(\eta^2,\eta)],
	\end{split}\end{equation*}
	where $\widehat{\bm{\Gamma}}=T^{-1}\bm{X}'\bm{X}$ and $\bm{X}=[\bm{x}_T,\dots,\bm{x}_1]'$.
	
	For two $p$-dimensional, centered, stationary Gaussian processes $\bm{y}_t$ and $\bm{\epsilon}_t$ such that $\text{Cov}(\bm{y}_t,\bm{\epsilon}_t)=0$ for every $t$.
	Let $\bm{X}=[\bm{x}_T,\dots,\bm{x}_1]'$ and $\bm{E}=[\bm{\epsilon}_T,\dots,\bm{\epsilon}_1]'$ be the data matrices. Then, there exists a constant $c>0$ such that for any $\bm{u}$, $\bm{v}\in\mathbb{R}^p$ with $\|\bm{u}\|_2\leq1$ and $\|\bm{v}\|_2\leq1$, we have
	\begin{equation*}\begin{split}
	\mathbb{P}\left[|\bm{u}'(\bm{X}'\bm{E}/T)\bm{v}|>2\pi\left(\lambda_{\max}(\bm{\Sigma_\epsilon})\left(1+\frac{\mu_{\max}(\mathcal{A})}{\mu_{\max}(\mathcal{A})}\right)\right)\right].
	\end{split}\end{equation*}
\end{lemma}

Finally, we state a union concentration inequality for vectors in a sparse set via discretization from Lemma F.2 in \citet{Basu15}.

\begin{lemma}
	Consider a symmetric matrix $\bm{D}_{p\times p}$. If, for any $\bm{v}\in\mathbb{R}^p$ with $\|\bm{v}\|_2\leq1$, and any $\eta\geq0$,
	\begin{equation*}
	\mathbb{P}[|\bm{v}'\bm{D}\bm{v}|>C\eta]\leq2\exp[-cT\min(\eta,\eta^2)]
	\end{equation*}
	then, for any integers $s\geq1$, we have
	\begin{equation*}
	\mathbb{P}\left[\sup_{\bm{v}\in\mathcal{K}(s)}|\bm{v}'\bm{D}\bm{v}|>C\eta\right]\leq2\exp[-cT\min(\eta^2,\eta)+s\min\{\log(p),\log(cp/s)\}].
	\end{equation*}
\end{lemma}

\section{Details about the SFM representation in \eqref{eq:MLRVAR_factor} \label{appendix:factor}} 


To verify the SFM representation  in \eqref{eq:MLRVAR_factor}, first note that the proposed model can be written in the matrix form, $\bm{Y}=\bm{X}(\bm{U}_3\otimes\bm{U}_2)\cm{G}_{(1)}'\bm{U}_1'+\bm{E}$, where $\bm{Y}=(\bm{y}_1,\dots, \bm{y}_T)^\prime$, $\bm{X}=(\bm{x}_1,\dots, \bm{x}_T)^\prime$, and $\bm{E}=(\bm{e}_1, \dots, \bm{e}_T)^\prime$. Consider the singular value decomposition
$\bm{X}(\bm{U}_3\otimes\bm{U}_2)\cm{G}_{(1)}'=\bm{U}_{\bm{x}} \bm{D}_{\bm{x}} \bm{V}_{\bm{x}}^\prime$, where $\bm{D}_{\bm{x}}\in\mathbb{R}^{r_1\times r_1}$ is diagonal, and $\bm{U}_{\bm{x}}$ and $\bm{V}_{\bm{x}}$ are orthonormal. Define $\bm{F}=\sqrt{T}\bm{U}_{\bm{x}}$ and $\bm{\Lambda}=\bm{U}_1\bm{V}_{\bm{x}}\bm{D}_{\bm{x}}/\sqrt{T}$. Note that $\bm{F}^\prime \bm{F}/T=\bm{I}_{r_1}$ and that $\bm{\Lambda}^\prime\bm{\Lambda}$ is diagonal. Thus,
$\bm{Y}=\bm{X}(\bm{U}_3\otimes\bm{U}_2)\cm{G}_{(1)}'\bm{U}_1'+\bm{E}=\bm{F}\bm{\Lambda}'+\bm{E}$,
which is the matrix form of \eqref{eq:MLRVAR_factor}.



\section{Generation of orthonormal matrices \label{appendix: sim}}
To generate an arbitrary tall orthonormal matrix $\bm{O}\in\mathbb{R}^{m\times n}$ with $m>n$, we first generate an $m\times m$ square matrix of independent standard normal random numbers, and then set its top $n$  singular vectors as the columns of $\bm{O}$.

We next generate the sparse orthonormal matrices in Sections 6.2 and 6.3. In Section 6.2, when $(r_1,r_2,r_3,s_1,s_2,s_3)=(2,2,2,3,3,2)$, let
\begin{equation}
\bm{U}_1= \begin{bmatrix}
\bm{a}_{3\times 1}  &\bm{0}_{3\times 1} \\
\hdashline
\bm{0}_{3\times 1} & \bm{b}_{3\times 1}\\
\hdashline
\bm{0}_{4\times 1} & \bm{0}_{4\times 1}   \\
\end{bmatrix} \in\mathbb{R}^{10\times 2}, \,
\bm{U}_2= \begin{bmatrix}
\bm{c}_{3\times 1}  &\bm{0}_{3\times 1} \\
\hdashline
\bm{0}_{3\times 1} & \bm{d}_{3\times 1}\\
\hdashline
\bm{0}_{4\times 1} & \bm{0}_{4\times 1}   \\
\end{bmatrix} \in\mathbb{R}^{10\times 2}, \,
\bm{U}_3= \begin{bmatrix}
1  & 0\\
\hdashline
\bm{0}_{2\times 1} & \bm{e}_{2\times 1}\\
\hdashline
\bm{0}_{2\times 2} & \bm{0}_{2\times 1}\\
\end{bmatrix} \in\mathbb{R}^{5\times 2},
\end{equation}
where $\bm{a}_{3\times 1}$, $\bm{b}_{3\times 1}$, $\bm{c}_{3\times 1}$, $\bm{d}_{3\times 1}$ and $\bm{e}_{2\times 1}$ are random Gaussian vectors scaled to have unit Euclidean norm. 

When $(r_1,r_2,r_3,s_1,s_2,s_3)=(3,3,3,3,3,2)$, let
\begin{equation}
\label{eq:gen1}
\begin{split}
\bm{U}_1&= \begin{bmatrix}
\bm{a}_{3\times 1}  &\bm{0}_{3\times 1} &\bm{0}_{3\times 1}\\
\hdashline
\bm{0}_{3\times 1} & \bm{b}_{3\times 1} &\bm{0}_{3\times 1}\\
\hdashline
\bm{0}_{3\times 1} & \bm{0}_{3\times 1} &\bm{c}_{3\times 1} \\
\hdashline
0 & 0 & 0
\end{bmatrix} \in\mathbb{R}^{10\times 3},\quad
\bm{U}_2= \begin{bmatrix}
\bm{d}_{3\times 1}  &\bm{0}_{3\times 1} &\bm{0}_{3\times 1}\\
\hdashline
\bm{0}_{3\times 1} & \bm{e}_{3\times 1} &\bm{0}_{3\times 1}\\
\hdashline
\bm{0}_{3\times 1} & \bm{0}_{3\times 1} &\bm{f}_{3\times 1} \\
\hdashline
0 & 0 & 0
\end{bmatrix} \in\mathbb{R}^{10\times 3}\\
\quad\text{and}\quad
\bm{U}_3&= \begin{bmatrix}
1  & 0 & 0\\
\hdashline
\bm{0}_{2\times 1} & \bm{g}_{2\times 1} & \bm{0}_{2\times 1}\\
\hdashline
\bm{0}_{2\times 2} & \bm{0}_{2\times 1} & \bm{h}_{2\times 1}\\
\end{bmatrix} \in\mathbb{R}^{5\times 3},
\end{split}
\end{equation}
where each nonzero part is generated by the same way as that for $(r_1,r_2,r_3)=(s_1,s_2,s_3)$.

When $(r_1,r_2,r_3,s_1,s_2,s_3)=(3,3,3,2,2,2)$, $\bm{U}_1$ and $\bm{U}_2$ are generated by the same way as that in \eqref{eq:gen1} except that the last entries in $\bm{a}_{3\times 1}$, $\bm{b}_{3\times1}$, $\bm{c}_{3\times1}$, $\bm{d}_{3\times1}$, $\bm{e}_{3\times1}$, and $\bm{f}_{3\times1}$ are zero.

For the cases with $(N,P)=(15,8)$, zero rows are added below $\bm{U}_i$s generated for $(N,P)=(10,5)$.

\end{document}